\UseRawInputEncoding
\documentclass[12pt,twoside]{article} 

\usepackage{natbib}
\usepackage{hyperref}
\usepackage{array,epsfig,fancyheadings,rotating}

\usepackage{sectsty, secdot}
\sectionfont{\fontsize{12}{14pt plus.8pt minus .6pt}\selectfont}
\renewcommand{\theequation}{\thesection\arabic{equation}}
\subsectionfont{\fontsize{12}{14pt plus.8pt minus .6pt}\selectfont}

\usepackage{amsmath,mathtools}
\usepackage{amssymb}
\usepackage{amsfonts}
\usepackage{multirow}
\usepackage{amsthm}
\usepackage{graphicx}
\usepackage{dsfont}
\usepackage{chngcntr}
\usepackage{setspace, enumerate,array}
\usepackage{subcaption}
\usepackage{csvsimple}
\usepackage[table,x11names]{xcolor}
\usepackage{soul}
\sethlcolor{LightBlue1}
\usepackage[font={small}]{caption}
\usepackage{floatrow,multirow}
\usepackage{hhline}
\usepackage[margin=1in]{geometry} 
\usepackage[compact]{titlesec}

\definecolor{myblue}{rgb}{0.56, 0.92, 0.75}
\definecolor{green1}{RGB}{44, 220, 45}
\definecolor{yellow0}{RGB}{255, 242, 94}

\newfloatcommand{capbtabbox}{table}[][\FBwidth]

\theoremstyle{plain}
\newtheorem{assump}{Assumption}

\newcommand*{\fbar}{\overline{f}}

\newtheorem{theorem}{Theorem}
\newtheorem{lemma}{Lemma}

\newtheorem{proposition}{Proposition}
\theoremstyle{definition}

\theoremstyle{remark}
\newtheorem*{remark}{Remark}
\theoremstyle{plain}

\newcolumntype{P}[1]{>{\centering\arraybackslash}p{#1}}
\DeclarePairedDelimiter\floor{\lfloor}{\rfloor}
\begin{document}


\renewcommand{\baselinestretch}{2}

\renewcommand{\sectionmark}[1]{}
\renewcommand{\subsectionmark}[1]{}
\markright{ \hbox{\footnotesize\rm Statistica Sinica
}\hfill\\[-13pt]
\hbox{\footnotesize\rm
}\hfill }

\markboth{\hfill{\footnotesize\rm WENTIAN HUANG AND DAVID RUPPERT} \hfill}
{\hfill {\footnotesize\rm COPULA-BASED FUNCTIONAL BAYES CLASSIFICATION} \hfill}

\renewcommand{\thefootnote}{}
$\ $\par


\fontsize{12}{14pt plus.8pt minus .6pt}\selectfont \vspace{0.8pc}
\centerline{\large\bf COPULA-BASED FUNCTIONAL BAYES CLASSIFICATION WITH}
\vspace{2pt} \centerline{\large\bf PRINCIPAL COMPONENTS AND PARTIAL LEAST SQUARES}
\vspace{.4cm} \centerline{Wentian Huang, David Ruppert} \vspace{.4cm} \centerline{\it
Cornell University} \vspace{.55cm} \fontsize{9}{11.5pt plus.8pt minus
.6pt}\selectfont


\begin{quotation}
\noindent {\it Abstract:}
We present a new functional Bayes classifier that uses principal component (PC) or partial least squares (PLS) scores from the common (i.e.\ pooled) covariance function, that is, the covariance function marginalized over groups.  When the groups have different covariance functions, the PC or PLS scores need not be independent or even uncorrelated.  
We use copulas to model the dependence.   Our method is semiparametric; the marginal densities are estimated nonparametrically using kernel smoothing, and the copula is modeled parametrically.  We focus on Gaussian and $t$-copulas, but other copulas can be used.   The strong performance of our methodology is demonstrated through simulation, real-data examples, and asymptotic properties.

\vspace{9pt}
\par
\end{quotation}\par

\def\thefigure{\arabic{figure}}
\def\thetable{\arabic{table}}

\renewcommand{\theequation}{\thesection.\arabic{equation}}

\fontsize{12}{14pt plus.8pt minus .6pt}\selectfont

\setcounter{section}{0} 
\setcounter{equation}{0} 

\lhead[\footnotesize\thepage\fancyplain{}\leftmark]{}\rhead[]{\fancyplain{}\rightmark\footnotesize\thepage}

\section{Introduction}\label{sec:intro}
Functional classification, where the features are continuous functions on a compact interval, is receiving increasing interest in fields such as chemometrics, medicine, economics, and environmental science.   \cite{hastie} extended the linear discriminant analysis (LDA) to functional data (FLDA), including the case where the curves are partially observed.  \cite{james02} proposed a functional version of the generalized linear model (FGLM), including functional logistic regression. Thereafter, the FGLM was further researched by, among others, \cite{muller2005generalized},  \cite{li2010generalized}, \cite{zhu2010bayesian}, \cite{mclean2014functional}, and \cite{shang2015nonparametric}. Aside from the FGLM, other classifiers have also been studied. \cite{svm} applied support vector machines (SVM) to classify infinite-dimensional data.  \cite{depth} explored the classification of functional data based on data depth.   \cite{li} suggested a functional segmented discriminant analysis combining an LDA and an \mbox{SVM}, and \cite{agg} proposed a nonlinear aggregation classifier. 

However, certain issues remain.  Current methods, such as the FLDA, SVM, and functional centroid classifier (\cite{DH2012}), distinguish groups by the differences between their functional means. They achieve satisfactory results when the location difference is the dominant feature distinguishing classes, but functional data provide more information than just group means. For example, Fig.\ \ref{ccamusd} from the example in Section~\ref{DTIandMS} compares the mean and standard deviation functions of raw and smoothed fractional anisotropy (FA) measured along the corpus callosum (cca) of $141$ subjects, $99$ with multiple sclerosis (MS) and $42$ without.  The disparity between the group standard deviations in panel (c) provides additional information that can identify MS patients.  As shown in Section~\ref{DTIandMS}, the LDA and centroid classifiers fail to capture this information, and have higher misclassification rates than the classifiers we propose.

\begin{figure}[h] 
\centering
      \includegraphics[scale=0.4]{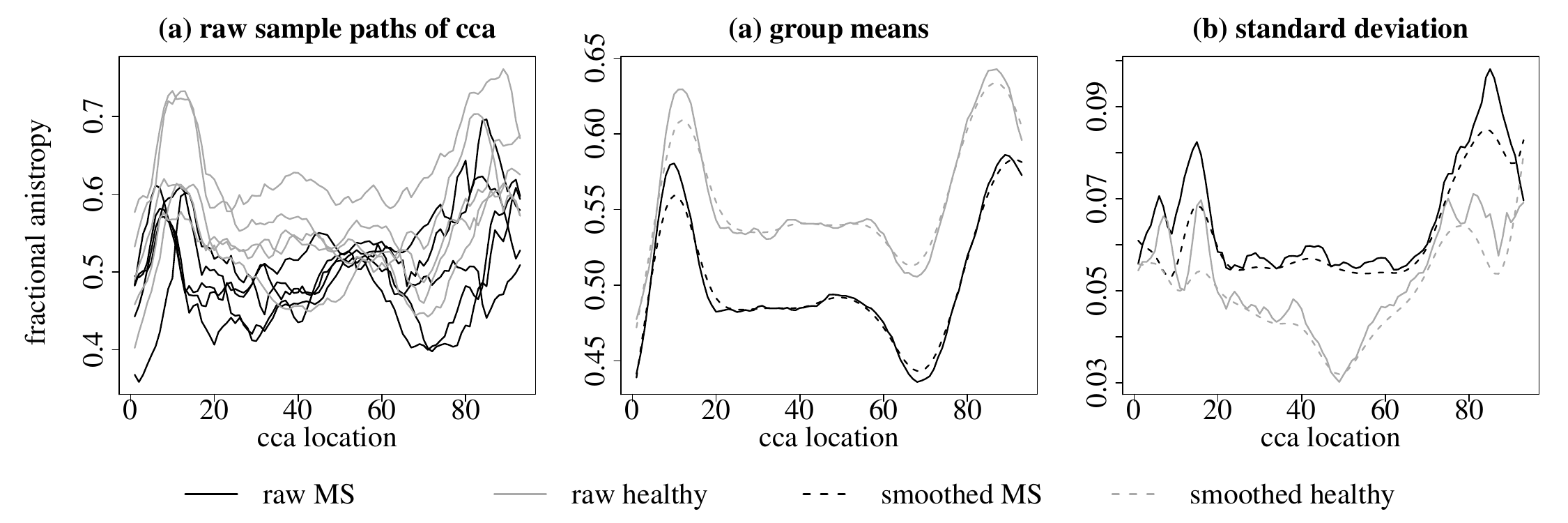}
  \caption{Panel (a) shows profiles of FA, five each of cases and controls, and panels (b) and (c) show the group means and standard deviations. Compared to the controls, the MS group has a lower mean and a higher standard deviation.}
  \label{ccamusd}
\end{figure}

Both parametric and nonparametric methods have drawbacks in classifying functional data.  Parametric models, such as linear and quadratic discriminant analysis, are popular in functional classification, especially because nonparametric methods are likely to encounter the curse of dimensionality. However, parametric methods can cast rigid assumptions on the class boundaries (\cite{li}).  Our interest is in methods that avoid stringent assumptions on the data.   \cite{DMY2017} proposed a nonparametric Bayes classifier, assuming that the subgroups share the same sets of eigenfunctions, and that the scores projected on them are independent. With these assumptions and the definition of the density of random functions proposed by \cite{DH2010}, the joint densities of the truncated functional data  can be estimated using a \textit{univariate} kernel density estimation (KDE).   The Bayes rules estimated this way avoid the curse of dimensionality, but require that the groups have equal sets of eigenfunctions and independent scores.

We propose new semiparametric Bayes classifiers.  We project the functions onto the eigenfunctions of the pooled covariance function, that is, the covariance function marginalized over groups.  
These eigenfunctions can be estimated by applying a functional principal components analysis (fPCA) to the combined groups.
The projections will not be independent or even uncorrelated, unless these common eigenfunctions are also the eigenfunctions of the group-specific covariance functions, an assumption not likely to hold in many situations. For instance, in Section \ref{realdata} we discuss two real-data examples, and include a comparison of their group eigenfunctions in the Supplementary Material (Fig.\ \ref{ccagroup} and Fig.\ \ref{truckgroup}). Both cases appear to violate the equal eigenfunction assumption. We estimate the marginal density of the projected scores using a univariate KDE, as in  \cite{DMY2017}, and model the association between the scores using a parametric copula. Our semiparametric methodology avoids the restricted range of applications imposed by the assumption of equal group-specific eigenfunctions.  It also avoids the curse of dimensionality that a multivariate nonparametric density estimation would entail.  

In addition to the principal components (PC) basis,
we also consider a partial least squares (PLS) projection basis.  
PLS has attracted recent attention owing to its effectiveness in prediction and classification problems with high-dimensional and functional data. \cite{P2007} discuss a functional LDA combined with \mbox{PLS}.  \cite{DH2012} mention the potential advantage of PLS scores in their functional centroid classifier, when the difference between the group means does not lie primarily in the space spanned by the first few eigenfunctions. We find that PLS scores can be more efficient than PC scores in capturing group mean differences.  

This study contributes to the literature in two ways. In our numerical results, the new method shows improved prediction accuracy and strength in dimension reduction, and extends the functional Bayes classification to multiclass classification. In the theoretical analysis, several new conditions are added for the functional data to achieve asymptotic optimality.  These conditions are required because of the unequal group-specific eigenfunctions. Moreover, we propose asymptotic sparsity assumptions on the inverse of the copula correlations in our new method, following the design of \cite{yuan10} and \cite{Liu} for high-dimensional data. We also build a new theorem that uses the special copula structure to achieve asymptotic perfect classification.

In Section \ref{model}, we introduce our model and the copula-based functional Bayes classifiers.  Section \ref{sim} contains a comprehensive simulation study comparing our methods with existing classifiers on both binary and multiclass problems.
Section \ref{realdata} uses two real-data examples to show the strength of our classifiers in terms of accuracy and dimension reduction with respect to data size.
In Section \ref{theory}, we discuss the asymptotic properties of our classifiers.  We also establish conditions for our classifiers to achieve perfect classification on data generated by Gaussian and non-Gaussian processes.  Finally, in Section \ref{discussion}, we discuss future work, including extending the classification to the case where there are multiple functional predictors. Additional results and detailed proofs are provided in the Supplementary Material.
\setcounter{equation}{0}
\section{Model Setup \& Functional Bayes Classifiers with Copulas} \label{model}

\subsection{Methodology}
Suppose $(X_{i \cdot \cdot}, Y_i), \  i=1, \ldots, n$ are independent and identically distributed (i.i.d.) from the joint distribution of $(X, Y)$, where  $X$ is a square integrable function over some compact interval $\mathcal{T}$, that is, $X \in \mathcal{L}^2(\mathcal{T})$. Here $Y=0, 1$ is an indicator of groups $\Pi_0$ and $\Pi_1$, respectively, and $\pi_k=P(Y=k)$.   In addition, $X_{i\cdot k}$, for $i=1, \ldots, n_k$ and $k=0,1$, denotes the $i$th sample curve of  $X_{\cdot \cdot k}=\left(X|Y=k\right)$, and $n=\sum_{k=0, 1} n_k$. Our goal is to classify a new observation, $x$. 

Note that throughout the paper, we order the index of $X$ by observation counts ($i$), joint basis ($j$), and group labels ($k$): for curves, $X_{i\cdot \cdot}$ denotes the $i$th observation of the random function $X$, and $X_{\cdot \cdot k}$ is the random function $X|Y=k$. Therefore, $X_{i\cdot k}$ is the $i$th sample curve of $X_{\cdot \cdot k}$. Furthermore, $X_{\cdot j \cdot}$ and $X_{\cdot jk}$ are random variables from projecting $X$ and $X_{\cdot \cdot k}$, respectively, onto the $j$th joint basis function $\psi_j$, with $X_{ijk}$ the $i$th observation of $X_{\cdot jk}$. 

\cite{DMY2017} extended the Bayes classification from multivariate  to functional data: a new curve $x$ is classified into $\Pi_1$ if
\begin{equation} \label{eq:1}
Q(x)=\dfrac{P(Y=1|X=x)}{P(Y=0|X=x)}=\dfrac{\fbar_1(x) \pi_1}{\fbar_0(x) \pi_0} \approx \dfrac{f_1(x_1, \ldots, x_J) \pi_1}{f_0(x_1, \ldots, x_J) \pi_0} > 1,
\end{equation}
where $\fbar_k$ is the density of $X_{\cdot \cdot k}$ and $f_k$ is the joint density of the scores $X_{\cdot jk}$ on the basis $\psi_j$, for $1 \le j \le J$.

A key feature of the Bayes classification on functional data is that the classifiers vary with the choice of basis functions $\psi_j$ and with the estimation of $f_0, f_1$.  \cite{DMY2017} built the original functional Bayes classifier (BC), upon two important assumptions.  First, the sets of the first $J$ eigenfunctions, $\{\phi_1, \ldots, \phi_J\}$, of the covariance operators $G_1$ and $G_0$ of the two groups are equal.  Here, $G_k (\phi_j)(t)=\int_{\mathcal{T}} G_k (s,t) \phi_j (s) ds=\lambda_{jk} \phi_j (t)$,  $G_k(s,t)=\text{cov} \{X_{\cdot \cdot k}(s), X_{\cdot \cdot k}(t) \}=\displaystyle \sum_{j=1}^{\infty} \lambda_{jk} \phi_j (s) \phi_j (t)$, and $\lambda_{jk}$ is the $j$th eigenvalue in group $k$.
Second, letting $\psi_j=\phi_j$, for $1 \le j \le J$, the $J$ projected scores $X_{\cdot jk}=\langle X_{\cdot \cdot k}, \phi_{j} \rangle$ are independent. Then, with $f_{jk}$ as the marginal density of $X_{\cdot jk}$, the log ratio of $Q(x)$ in Eq.(\ref{eq:1}) becomes
\begin{equation} \label{eq:2}
\log Q(x) \approx \log Q_J(x)=\log \left (\dfrac{\pi_1}{\pi_0} \right)+\displaystyle \sum_{j=1}^J \log \left \{ \dfrac{f_{j1}(x_j)}{f_{j0}(x_j)} \right \}.
\end{equation}

A classifier that uses Eq.(\ref{eq:2}) avoids the curse of dimensionality and only needs to estimate the marginal densities, $f_{jk}$. However, as later simulations and examples show, its performance can degrade if the two aforementioned assumptions are not met. We propose new semiparametric Bayes classifiers based on copulas that do not require these two assumptions, and yet are free from the curse of dimensionality. The theoretical work in Section~\ref{theory} proves that these classifiers maintain the advantages of BC over a wider range of data distributions, and are capable of perfect classification when $n \to \infty$ and $J\to\infty$.

\subsection{Copula-Based Bayes Classifier with PC} \label{section: 2.2}
Allowing for possibly unequal group eigenfunctions, the covariance function of  group $k$ is
 \[
 G_k (s,t)= \text{cov}\left(X_{\cdot \cdot k} (s), X_{\cdot \cdot k}(t) \right)=\displaystyle \sum_{j=1}^\infty \lambda_{jk} \phi_{jk}(s) \phi_{jk}(t), \ k=0,1,
 \]
 with $\phi_{1k}, \ldots, \phi_{Jk}$ as the eigenfunctions. For simplicity, we assume the group means are $E(X|Y=0)=0$ and $E(X|Y=1)=\mu_d$. The joint covariance operator $G$ then has the kernel $G (s,t)= \pi_1 G_1(s,t)+\pi_0 G_0(s,t)+\pi_1\pi_0\mu_d(s)\mu_d(t)$. 
 
 As later examples suggest, the unequal group eigenfunction case is common.
 To accommodate this case, we  can project data from both groups onto the same basis functions.
Therefore, we use the eigenfunctions $\phi_1, \ldots, \phi_J$ of $G$ as the basis $\psi_1, \ldots, \psi_J$. 

The joint density $f_k$,for $k=0,1$, in Eq.(\ref{eq:1}) allows for potential score correlation and tail dependency, which we use copulas to model. A copula is a multivariate cumulative distribution function (CDF) with univariate marginal distributions that are all uniform, and it characterizes only the dependency between the components; see, for example, \cite{RM2015}. Here, we extend its use to truncated scores of functional data. 

Let $x_j=\langle x, \phi_j\rangle=\int_{\mathcal{T}} x(t)\phi_j(t) dt$ be the $j$th projected score of $x$. The copula function $C_k$ describes the distribution of the first $J$ scores in $\Pi_k$ by
\begin{align} 
F_k \left (x_1, \ldots, x_J \right )&=C_k\left \{F_{1k}(x_1), \ldots, F_{Jk}(x_J) \right \}, \label{eq:3}\\
f_k \left (x_1, \ldots, x_J \right )&=c_k\left \{F_{1k}(x_1), \ldots, F_{Jk}(x_J) \right \} f_{1k}(x_1)\cdots f_{Jk}(x_J). \label{eq:4}
\end{align} 
$F_k$ in Eq.(\ref{eq:3}) is the joint CDF of $X_{\cdot 1k}, \ldots, X_{\cdot Jk}$, and $C_k$ is the CDF of the uniformly distributed variables $F_{1k}(X_{\cdot 1k}), \ldots, F_{Jk}(X_{\cdot Jk})$, where $F_{jk}$ is the univariate CDF of $X_{\cdot jk}$. In Eq.(\ref{eq:4}), the joint density $f_k$ is decomposed into score marginal densities $f_{jk}$ and the copula density $c_k$ for the dependency between the projected scores. Our revised classifier is $\mathds{1}\left\{\log Q_J^*(x) >0\right\}$; that is, the new curve $x$ belongs to $\Pi_1$ if
\begin{equation} \label{eq:5}
\log Q^*_J \left(x\right)=\log \left (\dfrac{\pi_1}{\pi_0} \right )+\displaystyle \sum_{j=1}^J \log \left \{ \dfrac{f_{j1}(x_j)}{f_{j0}(x_j)} \right \} + \log \left \{\dfrac{c_1\{ F_{11}(x_1), \ldots, F_{J1}(x_J)\}}{c_0\{ F_{10}(x_1), \ldots, F_{J0}(x_J)\}} \right \} > 0.
\end{equation}

We also consider situations in which $Y$ has more than two classes. A more general procedure for multiclass classification is described in the Supplementary Material Section \ref{sp:multiclass}.

\subsection{Choice of Copula and Correlation Estimator} \label{section:copula}
There are a number of approaches to copula estimation. \cite{G1995} studied the asymptotic properties of semiparametric estimation in copula models. \cite{CHEN2006} discussed semiparametric copula estimation to characterize the temporal dependence in time series data. \cite{K2013} estimated the copula density nonparametrically using penalized splines, and \cite{GOV2012} applied multivariate kernel density estimation to copulas. 

To address the high dimensionality of functional data, we model the copula densities $c_1$ and $c_0$ parametrically, and use a kernel estimation for the univariate densities $f_{1k}, \ldots, f_{Jk}$, for $k=0, 1$.
We study the properties of Bayes classification using both Gaussian copulas and t-copulas, denoted by BCG and BCt, respectively. When $c_k$ is modeled by a Gaussian copula in Eq.(\ref{eq:4}),
$
c_k (\cdot)  = c_{ G,k }( \cdot | \mathbf{\Omega}_{G,k}),
$
where $c_{G,k}$ is the Gaussian copula density with $J \times J$ correlation matrix $\mathbf{\Omega}_{G,k}$. When there is tail dependency between the scores, a t-copula is used:
$
c_k( \cdot ) = c_{t,k}(  \cdot  | \mathbf{\Omega}_{t,k}, \nu_k  ),
$
with $c_{t,k}$ the t-copula density, $\mathbf{\Omega}_{t,k}$ the correlation matrix, and $\nu_k$ the tail index. 

There are several ways to estimate the correlation matrices $\mathbf{\Omega}_{G,k}$ or $\mathbf{\Omega}_{t,k}$.  We use rank correlations, and specifically, Kendall's $\tau$. Kendall's $\tau$ between the projected scores of $X_{\cdot \cdot k}$ on the $j$th and $j'$th basis is 
$
\rho_{\tau} \left(X_{\cdot j k}, X_{\cdot j' k} \right)= E \left[\text{sign} \left\{ \left(X^{(1)}_{\cdot j k}-X^{(2)}_{\cdot j k} \right)\left(X^{(1)}_{\cdot j' k}-X^{(2)}_{\cdot j' k} \right) \right \} \right], 
$
sign$\left(x\right)=\mathds{1} \left\{x>0\right\}-\mathds{1} \left\{x<0\right\}$, and $X^{(1)}_{\cdot \cdot k}$, $X^{(2)}_{\cdot \cdot k}$ are i.i.d.\ samples of $X_{\cdot\cdot k}$. The robustness of the rank correlation and its optimal asymptotic error rate are studied by  \cite{Liu}.

A relationship exists between the $(j, j')$th entry of the copula correlation $\mathbf{\Omega}_k$ and Kendall's $\tau$: $\mathbf{\Omega}^{jj'}_{k}=\sin \left(\dfrac{\pi}{2} \rho_{\tau}  \left(X_{\cdot j k}, X_{\cdot j' k} \right) \right)$ for both Gaussian copulas and $t$-copulas (\cite{Kendall}; \cite{Kruskal}; \cite{RM2015}). Then, $\mathbf{\Omega}^{jj'}_{k}$ is estimated by Kendall's $\tau$ as
$\hat{\mathbf{\Omega}}_k^{jj'}=\sin \left(\dfrac{\pi}{2} \hat{\rho}_{\tau, k}^{jj'}\right)$, where
\begin{align} 
&\hat{\rho}_{\tau, k}^{jj'}=\dfrac{2}{n_k\left(n_k-1\right)}\sum_{1 \le i \le i' \le n_k} \text{sign} \left\{\langle X_{i \cdot k}-X_{i' \cdot k}, \hat{\phi}_j\rangle \langle X_{i \cdot k}-X_{i' \cdot k}, \hat{\phi}_{j'}\rangle\right\}. \nonumber
\end{align} 
It is possible that $\hat{\mathbf{\Omega}}_k$ is not positive definite, but this problem is easily remedied (\cite{RM2015}).
Another rank correlation, Spearman's $\rho$, is similar and is omitted here. In the Supplementary Material \ref{sup:correst}, we show that for Gaussian copulas, the difference between the log determinant of $\hat{\mathbf{\Omega}}_k$, as estimated, and that of $\mathbf{\Omega}_k$ is $Op \left(J\sqrt{(\log J)/n}\right)$.

Additionally for t-copulas with $\hat{\mathbf{\Omega}}_{t,k}$, we apply a pseudo-maximum likelihood to estimate the tail parameter $\nu_k > 0$ by maximizing the log copula density
\newline
$
\displaystyle \sum_{i=1}^{n_k} \log \left[c_{t,k} \left\{\hat{F}_{1k} \left(X_{i1k} \right), \ldots, \hat{F}_{Jk}\left(X_{iJk}\right)|\hat{\mathbf{\Omega}}_{t,k}, \nu_k \right\} \right],
$
with $\hat{F}_{jk} \left(x\right)=\sum_{i=1}^{n_k} \mathds{1}\left\{X_{ijk} \le x\right\}/\left(n_k+1\right)$. \cite{MZ2002} discuss the maximum pseudo-likelihood estimation of t-copulas, and apply it to model extreme co-movements of financial assets. 

\subsection{Marginal Density $f_{jk}$ Estimation}
We estimate the marginal density $f_{jk}$ of the projected scores $X_{\cdot j k}$ using a kernel density estimation:
$
\hat{f}_{jk} \left(\hat{x}_j\right)=\dfrac{1}{n_k h_{jk}} \displaystyle \sum_{i=1}^{n_k} K\left(\dfrac{\langle x-X_{i \cdot k}, \hat{\phi}_j \rangle}{h_{jk}}\right),
$
with $K$ the standard Gaussian kernel, $\hat{\phi}_j$ the estimated $j$th joint eigenfunction, $h_{jk}=\hat{\sigma}_{jk} h$ the bandwidth for scores projected on $\hat{\phi}_j$ in group $k$, $\hat{\sigma}_{jk}$ as the estimated standard deviation of $\sigma_{jk}=\sqrt{\text{Var }(X_{\cdot j k})}$, and $\hat{x}_j=\langle x, \hat{\phi}_j\rangle$. Then, $\log Q_J^*\left(x\right)$ in Eq.(\ref{eq:5})  is estimated by
\begin{equation*} 
\log \hat{Q}^*_J \left(x\right)=\log \left (\dfrac{\hat{\pi}_1}{\hat{\pi}_0} \right )+\displaystyle \sum_{j=1}^J \log \left \{ \dfrac{\hat{f}_{j1}(\hat{x}_j)}{\hat{f}_{j0}(\hat{x}_j)} \right \} + \log \left \{\dfrac{\hat{c}_1\{ \hat{F}_{11}(\hat{x}_1), \ldots, \hat{F}_{J1}(\hat{x}_J)\}}{\hat{c}_0\{ \hat{F}_{10}(\hat{x}_1), \ldots, \hat{F}_{J0}(\hat{x}_J)\}} \right \},
\end{equation*}
where $\hat{c}_k$ is the Gaussian copula or t-copula density with the estimated parameters, and $\hat{\pi}_k=n_k/n$. Proposition~\ref{prop1} in Section~\ref{theory} shows that with an additional mild assumption, when the group eigenfunctions are unequal, $|\hat{f}_{jk}(\hat{x}_j)-f_{jk}(x_j)|$ is asymptotically bounded at the same rate as when the eigenfunctions are equal. Detailed proofs are included in Supplementary Material.

\subsection{Copula-Based Bayes Classifier with Partial Least Squares} \label{section:PLS}
An interesting alternative to using PCs is to use functional partial least squares (FPLS). FPLS finds directions that maximize the covariance between the projected $X$ and $Y$ scores, rather than focusing on the variation in  $X$ alone, as with \mbox{PCA}.  As the algorithm in the Supplementary Material \ref{sup:pls} describes, FPLS iteratively generates a weight function $w_j$ at each step $j$, for $1 \le j \le J$, which solves
$
\max_{w_j \in \mathcal{L}^2(\mathcal{T})} \text{cov}^2 \left\{Y^{j-1}, \langle X^{j-1}, w_j \rangle \right\},
$
such that $\|w_j\|=1$ and $ \langle w_j, G(w_j') \rangle=0$, for all $ 1 \le j' \le j-1$.  Recall that $G$ is the joint covariance operator of the random function $X$.   Here, $Y^{j-1}$ and $X^{j-1}$ are the updated function $X$ and the indicator $Y$ at step $j-1$ (see \ref{sup:pls}), respectively, and their corresponding sample values are denoted as $Y^{j-1}_i$ and $X^{j-1}_{i\cdot \cdot}$, for $i=1, \ldots, n$.

The algorithm gives the decomposition $X_{i \cdot \cdot}(t)=\sum_{j=1}^J s_{ij} P_j(t)+E_{i}(t)$, for $t \in \mathcal{T}$, where $\mathbf{s}_i=\left(s_{i1}, \ldots, s_{iJ}\right)^T$ is the length $J$ score vector, $P_j \in \mathcal{L}^2(\mathcal{T})$, for $1 \le j \le J$, are loading functions, and $E_{i}$ is the residual.  \cite{P2007} investigated PLS in linear discriminant analysis (LDA), and defined score vectors $\mathbf{S}_j$ as eigenvectors of the product of the Escoufier's operators of $X$ and $Y$ (\cite{E1970}).  For our case, the classifiers BCG and BCt now act on the PLS scores $\mathbf{s}_i=\left(s_{i1}, \ldots, s_{iJ}\right)^T$ of each observation $X_{i \cdot \cdot}$. We refer to these classifiers as BCG-PLS and BCt-PLS, respectively.  

The dominant PCA directions might only have large within-group variances and small between-group differences in means.  Such directions will have little power to discriminate between groups.
This problem can be fixed by FPLS.   The advantages of FPLS have been discussed, for example, by \cite{P2007} and  \cite{DH2012}. The latter found that when the difference between the group means projected on the $j$th PC direction is large only for large  $j$, their functional centroid classifier with PLS scores has lower misclassification rates than when using PCA scores.  As later examples show, FPLS is especially effective  in such situations. 

\setcounter{equation}{0}
\section{Comparison of Classifiers using Simulated Data} \label{sim}


\subsection{Data Design} \label{design}
To set up the simulation, for simplicity, we use $\pi_1=\pi_0=0.5$.  By  Karhunen\textendash Lo\`eve expansions, the functions $X_{i\cdot k}$, for $i=1, \ldots, n_k$, of group $k=0, 1$ can be decomposed as $X_{i\cdot k}=\mu_k + \sum_{j=1}^J \sqrt{\lambda_{jk}} \xi_{ijk} \phi_{jk}$, where $\mu_k$ is the group mean, $\lambda_{jk}$ is the $j$th eigenvalue in group $k$ corresponding to eigenfunction $\phi_{jk}$, and $\lambda_{1k} > \dots > \lambda_{Jk}$. The variables $\xi_{ijk}$ are distributed with $E(\xi_{ijk})=0$, var$(\xi_{ijk})=1$, and cov$(\xi_{i j k}, \xi_{i j' k})=0$, for $\forall j \ne j'$. The compact interval $\mathcal{T}$ is $\left[0, 1\right]$, and the functions $X_{i\cdot k}$ are observed at the equally spaced grid $t_1=0, t_2=1/50, \ldots, t_{51}=1$, with i.i.d.\ Gaussian noise $\epsilon_{ik}(t)$ centered at zero and with standard deviation $0.5$.  The classifiers are implemented both with and without pre-smoothing the data. Because they have similar performance, we report only the results using pre-smoothing. The total sample size is $n=250$, with $100$ training and $150$ test cases. The number of eigenfunctions for curve generation is $J=201$, double the size of the training data set, to imitate the infinite dimensions of the functional data. For each $j$, the bandwidth $h_{jk}$ for KDE is selected by the direct plug-in method (\cite{S1991}). Simulations are repeated $N=1000$ times. The Supplementary Material \ref{sp:v} includes additional results with increased training size.

The distribution of $(X,Y)$ is determined by four factors: the eigenfunctions (whether common or group-specific), difference between group means, eigenvalues, and score distributions. 
The factors are varied according to a $2 \times 2 \times 2 \times 3$ full factorial design, described below. We adopt a four-letter system to label the 24 factor-level combinations, which we call  ``scenarios.''

\bigskip\noindent
{\bf Factor 1: Eigenfunctions $\phi_{1k}, \ldots, \phi_{Jk}$ of group $k$:} \label{section3.1.1}
The first factor specifies the eigenfunctions of the covariance operators $G_1$ and $G_0$. When the two sets $\phi_{1k}, \ldots, \phi_{Jk}$, for $k=0, 1$, are the same, let the common eigenfunctions be the Fourier basis on $\mathcal{T}=\left[0, 1\right]$, where $\phi_{1k}(t)=1, \phi_{jk}(t)=\sqrt{2} \cos(j\pi t)$ or $\sqrt{2} \sin \left(\left(j-1\right)\pi t\right)$, for $1< j \le 201$ even or odd.

When the two groups have unequal eigenfunctions, the group $k=0$  uses the Fourier basis $\phi_{10}, \ldots, \phi_{J0}$ as above, but the group $k=1$ has a Fourier basis rotated by iterative updating: 
\begin{enumerate}[i)]
\item
let the starting value of $\phi_{11}, \ldots, \phi_{J1}$ be the original Fourier basis functions, as above;
\item
at step $(j, j')$, where $1 \le j \le J-1$, $j'=j+1, \ldots, J$, the pair of functions $(\phi^*_{j1},  \phi^*_{j'1})$ is generated by a Givens rotation of angle $\theta_{jj'}$ of the current pair $\left(\phi_{j1}, \phi_{j'1}\right)$ such that $\phi^*_{j1}(t)=\cos \left(\theta_{j j'}\right) \phi_{j1}(t)-\sin \left(\theta_{j j'}\right) \phi_{j'1}(t)$, $\phi^*_{j'1}(t)=\sin\left(\theta_{j j'}\right) \phi_{j1}(t)+\cos\left(\theta_{j j'}\right) \phi_{j'1}(t)$.

\item
the rotation angle for each pair of $(j, j')$ is $\theta_{j j'}=\dfrac{\pi}{3} \left(\lambda_{j0}+\lambda_{j'0}\right)$, with $\lambda_{j0},\lambda_{j'0}$ the $j$th and $j'$th eigenvalues, respectively, of group $k=0$. Hence, the major eigenfunctions receive greater rotations, with the angles proportional to their eigenvalues;
\item
then, we update $\phi_{j1}, \phi_{j'1}$ with the new $\phi^*_{j1}, \phi^*_{j'1}$ and continue the rotations until each pair of $(j, j')$, with $1 \le j \le J-1$, $j'=j+1, \ldots, J$, is rotated.
\end{enumerate}

The rotated Fourier basis of group $k=1$ guarantees that both groups $\Pi_1$ and $\Pi_0$ span the same eigenspace and satisfy the null hypothesis of the test of equal eigenspaces developed by \cite{BHK2009}.   This test was used by \cite{DMY2017} to check whether the two groups have the same eigenfunctions.    However, having equal eigenspaces is a necessary, but not sufficient condition for having equal sets of eigenfunctions, as proved by the rotated basis.
Because of the unequal eigenfunctions of the operators $G_1$ and $G_0$, the scores $X_{ijk}$ are correlated, which can be modeled by the new copula-based classifiers.  

We also tested other choices of the second set of eigenfunctions, including the Haar wavelet system on $\mathcal{L}^2([0,1])$. However, the results are similar, and so are omitted. We denote the scenario where $\Pi_1$ and $ \Pi_0$ have equal eigenfunctions as S (same), and otherwise as R (rotated).

\bigskip\noindent
{\bf Factor 2: Difference, $\mu_d$, Between the Group Means:}
The second factor, which is at two levels, S (same) and D (different), is the difference between the group means, $\mu_d=\mu_1-\mu_0$. For simplicity, we let $\mu_0=0$, $\mu_1=\mu_d$. Here, $\mu_{d}(t)=t$.

\bigskip\noindent
{\bf Factor 3: Eigenvalues $\lambda_{1k}, \ldots, \lambda_{Jk}$ of Group $k$:}
The third factor, at two levels labeled S and D, is whether the eigenvalues $\lambda_{1k}, \ldots, \lambda_{Jk}$ depend on $k$. 
We label the level where $\lambda_{j1}=\lambda_{j0}=1/j^2$ as S,  and that when $\lambda_{j1}=1/j^3$ and $ \lambda_{j0}=1/j^2$ as D, for $1 \le j \le J$.

\bigskip\noindent
{\bf Factor 4: Distribution of the standardized scores $\xi_{ijk}$:}
The fourth factor, at three levels N (normal), T (tail dependence and skewness), and V (varied), is the distribution of $\xi_{ijk}$. 

{\it N}: $\xi_{i1k}, \ldots, \xi_{iJk}$ have a Gaussian distribution $N\left(0, 1\right)$ for both $k=0$ and $1$.  

{\it T:} This level includes tail dependency by setting $\xi_{ijk}=\left(\delta_{ijk}-b\right)/\eta_{ik}$, where $\delta_{ijk} \sim \text{Exp}(\lambda^*), \lambda^*=5\sqrt{3}/3, b=1/\lambda^*$, and $\eta_{ik} \sim \chi^2 (5)/5$, for all $j=1, \ldots, J$. All $\delta_{ijk}$ and $\eta_{ik}$ are mutually independent, whereas the scores $\xi_{ijk}$ on each basis $j$ are  uncorrelated, but dependent, because they share the same denominator, $\eta_{ik}$.   The scores are skewed in both groups.

{\it V:}  In this level, the scores in the two groups have different types of distributions, with $\xi_{ij1} \sim N\left(0, 1\right)$, and $\xi_{ij0} \sim \text{Exp}(1)-1$. Simulation results of a different choice of the varied distributions of $\xi_{ij1}$ and $\xi_{ij0}$ are included in Supplementary Material Section \ref{sp:v} Table \ref{sp:binaryV500}.

Table \ref{table:1} lists all $24$ scenarios used in the simulations:
\begin{table}[h!]
\centering
\begin{tabular}{|P{4cm}|P{2cm}|P{2cm}|P{2cm}|} \hline 
   & $\xi_{ijk} \sim$ N  & $\xi_{ijk} \sim$ T & $\xi_{ijk} \sim$ V \\ 
 \hline
 $ \mu_d=0, \  \lambda_{j1} = \lambda_{j0}$ &  (R/S)SSN & (R/S)SST & (R/S)SSV \\
 \hline
 $ \mu_d=0, \ \lambda_{j1} \ne \lambda_{j0}$ &  (R/S)SDN & (R/S)SDT & (R/S)SDV \\ 
 \hline
 $\mu_d \ne 0,\ \lambda_{j1}=\lambda_{j0}$ &  (R/S)DSN & (R/S)DST  & (R/S)DSV\\ 
 \hline
 $ \mu_d \ne 0, \ \lambda_{j1} \ne \lambda_{j0}$ &  (R/S)DDN & (R/S)DDT & (R/S)DDV \\
 \hline
\end{tabular}
\caption{Simulation scenarios. The labels are ordered: eigenfunctions (R/S), group mean (S, D), eigenvalues (S, D), and $\xi_{ijk}$ distributions (N, T, V). Note that in SSSN and SSST, functions from both groups have the same distribution. We simply include them to have a full factorial design.}
\label{table:1}
\end{table}

\subsection{Functional Classifiers} \label{cltypes}
The classifiers used in this study are listed below. Five of them are  Bayes classifiers, and the last three are non-Bayes.  The methods proposed in this paper are described in (ii) - (iii).
\begin{enumerate}[(i)]
\item
BC: the original Bayes classifier of \cite{DMY2017}, with the log density ratio given by Eq.(\ref{eq:2}). The scores are by projection onto PCs;
\item
BCG, BCG-PLS: Bayes classifiers with a Gaussian copula to model correlation, using PC and PLS scores, respectively. The rank correlation used is Kendall's $\tau$. Both the Gaussian copula and the t-copula densities can be implemented using the R package {\tt copula} (\cite{copula});

\item
BCt, BCt-PLS: Bayes classifiers similar to (ii), but using a t-copula instead;

\item
CEN: functional centroid classifier in \cite{DH2012}, where observation $x$ is classified to group $k=1$ if $T(x)=\left(\langle x, \psi \rangle-\langle \mu_1, \psi \rangle \right)^2-\left(\langle x, \psi \rangle-\langle \mu_0, \psi \rangle \right)^2 \le 0$, with $\mu_1$ and $\mu_0$ the group means. Here, $\psi=\sum_{j=1}^{J^*} \lambda_j^{-1}\mu_j \phi_j$ is a function of the first $J^*$ joint eigenfunctions $\phi_j$, the corresponding eigenvalues $\lambda_j$, and  $\mu_j=\langle \mu_1-\mu_0, \phi_j \rangle$;
\item
PLSDA (PLS discriminant analysis): binary classifier using Fisher's linear discriminant rule, with FPLS as a dimension-reduction method. It is implemented in the R package {\tt pls} (\cite{plspkg}); 
\item
Logistic regression: logistic regression on functional PCs, implemented by the R function {\tt glm}. It is one of the functional generalized regressions discussed in \cite{muller2005generalized}.
\end{enumerate}  

In each simulation, $J^*$  is selected using $10$-fold cross validation on the training data. The candidate $J$ values range from $1$ to $30$ ($2$ to $30$ for classifiers using copulas). The estimation of the joint eigenfunctions $\phi_j$ follows the discretization approach of the fPCA, as described in Chapter 8.4 of \cite{FDA2005}. A similar discretization strategy is used for the PLS basis.

\subsection{Classifier Performance} \label{simsection}
\begin{table}[h!] 
\centering
\resizebox{\columnwidth}{!}{
\begin{tabular}{r|r|rrrr|rrr|rr}
  \hline
 & BC & BCG & BCGPLS & BCt & BCtPLS & CEN & PLSDA & logistic & CV & Ratio (CV)\\ 
  \hline
  \rowcolor{Gray0}
SSSN & 0.502 & 0.502 & 0.500 & 0.500 & 0.501 & 0.502 & 0.501 & 0.500 & 0.501 & 0.23\%\\ 
  SSDN &\textbf{ 0.227} & 0.244 & 0.345 & 0.258 & 0.443 & 0.464 & 0.495 & 0.466 & 0.232 & 2.43\% \\ 
  SDSN & 0.347 & 0.351 & 0.361 & 0.351 & 0.363 &\textbf{ 0.275} & 0.304 & 0.279 & 0.291 & 5.88\% \\ 
  SDDN &\textbf{ 0.169} & 0.173 & 0.303 & 0.175 & 0.327 & 0.231 & 0.262 & 0.234 & 0.173 & 2.64\%\\ 
  \hline
  \rowcolor{Gray0}
  SSST & 0.507 & 0.502 & 0.500 & 0.505 & 0.499 & 0.499 & 0.499 & 0.499 & 0.502 & 0.69\%\\ 
  SSDT &\textbf{ 0.438} &\textit{ 0.441} & 0.454 & 0.456 & 0.471 & 0.488 & 0.497 & 0.490 & 0.452 & 3.19\%\\ 
  SDST & 0.188 & 0.183 & 0.270 & 0.184 & 0.311 &\textbf{ 0.167} & 0.234 &\textit{ 0.169} & 0.170 & 1.96\% \\ 
  SDDT & 0.166 & 0.161 & 0.237 & 0.160 & 0.296 &\textbf{ 0.148} & 0.233 &\textit{ 0.150} & 0.152 & 2.59\% \\ 
  \hline
  SSSV &\textbf{ 0.355} & 0.361 & 0.484 & 0.363 & 0.493 & 0.476 & 0.481 & 0.489 & 0.363 & 2.20\%\\ 
  SSDV &\textbf{ 0.253} & 0.270 & 0.373 & 0.276 & 0.430 & 0.455 & 0.477 & 0.462 & 0.257 & 1.78\%\\ 
  SDSV &\textbf{ 0.264} & 0.275 & 0.401 & 0.276 & 0.408 & 0.279 & 0.315 & 0.283 & 0.273 & 3.27\%\\ 
  SDDV &\textbf{ 0.202} & 0.209 & 0.309 & 0.207 & 0.313 & 0.236 & 0.280 & 0.238 & 0.210 & 3.95\%\\ 
  \hline
  \hline
  RSSN & 0.327 &\textbf{ 0.147} & 0.183 & \textbf{0.147} & 0.180 & 0.494 & 0.497 & 0.485 & 0.151 & 2.67\%\\ 
  RSDN & 0.252 &\textbf{ 0.090} & 0.140 & 0.093 & 0.164 & 0.489 & 0.500 & 0.482 & 0.093 & 2.93\%\\ 
  RDSN & 0.287 &\textbf{ 0.128} & 0.154 & \textbf{0.128} & 0.152 & 0.327 & 0.333 & 0.329 & 0.131 & 2.71\%\\ 
  RDDN & 0.208 &\textbf{ 0.077} & 0.112 &\textit{ 0.079} & 0.128 & 0.287 & 0.300 & 0.288 & 0.080 & 3.44\%\\ 
  \hline
  RSST & 0.435 &\textbf{ 0.354} & 0.373 &\textit{ 0.357} & 0.372 & 0.486 & 0.490 & 0.489 & 0.361 & 1.95\%\\ 
  RSDT & 0.400 &\textbf{ 0.326} & 0.348 & 0.336 & 0.365 & 0.486 & 0.491 & 0.485 & 0.339 & 3.87\%\\ 
  RDST & 0.178 &\textbf{ 0.148} & 0.248 & 0.154 & 0.261 & 0.174 & 0.252 & 0.175 & 0.156 & 5.80\%\\ 
  RDDT & 0.166 &\textbf{ 0.137} & 0.217 & 0.142 & 0.255 & 0.159 & 0.249 & 0.158 & 0.147 & 7.68\%\\ 
  \hline
  RSSV & 0.266 &\textbf{ 0.147} & 0.202 &\textit{ 0.149} & 0.204 & 0.472 & 0.481 & 0.475 & 0.150 & 1.71\%\\ 
  RSDV & 0.233 &\textbf{ 0.100} & 0.143 & 0.105 & 0.157 & 0.465 & 0.475 & 0.469 & 0.104 & 3.85\%\\ 
  RDSV & 0.241 &\textbf{ 0.145} & 0.183 &\textit{ 0.146} & 0.191 & 0.332 & 0.349 & 0.337 & 0.148 & 2.28\%\\ 
  RDDV & 0.238 &\textbf{ 0.116} & 0.157 & 0.120 & 0.167 & 0.299 & 0.325 & 0.300 & 0.121 & 3.97\%\\ 
   \hline
\end{tabular}}
\caption{Misclassification rates of eight classifiers on $24$ scenarios, each an average  from $1000$ simulations. Lowest rates of each data case are in bold, and cases within margin of error (see text) of the lowest are in italics. The column labeled CV contains error rates of the classifier selected by cross-validation. Ratio(CV) is the percent difference from the best of the eight classifiers for that scenario. CV error rates are not included in the rankings that determine coloring. SSSN and SSST are in gray because there is actually no difference between groups in these scenarios, and, because $\pi_0 = \pi_1 = 1/2$, the true misclassification rate is 0.5.}
\label{rates}
\end{table}

Table \ref{rates} contains the average misclassification rates over $1000$ simulations by each method on each scenario.  In addition, for each simulation, we use $10$-fold cross-validation to select the classifier with the best performance on the training data among the eight classifiers in Section~\ref{cltypes}. The average misclassification rates of the CV-selected classifier are listed in the CV column.  The column Ratio(CV) contains the percentage  difference between the CV-selected (CV) and the best (opt) classifier: $\text{Ratio(CV)} =\left \{\text{err(CV)}-\text{err(opt)} \right \} / \text{err(opt)} \times 100\%$. For each scenario, the lowest error rates of the eight classifiers are in bold. We label those within the optimal case's margin of error (MOE) for each data scenario $\gamma$ in italics: $\text{MOE}_{\gamma}=1.96\times\sigma^*_{\gamma}/\sqrt{1000}$, where $\sigma^*_{\gamma}$ is the sample standard deviation of the best classifier's (at scenario $\gamma$) error rates from $1000$ simulations. The simulations enable a comprehensive understanding of the classifiers' behaviors, which we now discuss.

\begin{itemize}
\item[\textendash]
\textit{Equal versus Unequal Eigenfunctions}. A comparison between the top and bottom half of Table \ref{rates} demonstrates the strength of our copula-based classifiers, especially on unequal eigenfunctions (bottom half). By its nature, BC has strong performance when the two groups have the same set of eigenfunctions and the scores $\xi_{ijk}$ are mutually independent, for example, in SSDN and SSDV. However, when the data have a more complicated structure, such as score tail dependency and location difference, CEN and logistic obtain better results (SDST, SDDT). Note that in every case with equal eigenfunctions, BCG/BCt are always the ones with rates closest to those of BC.  

\begin{figure}[h!] 
\centering
      \includegraphics[scale=0.43]{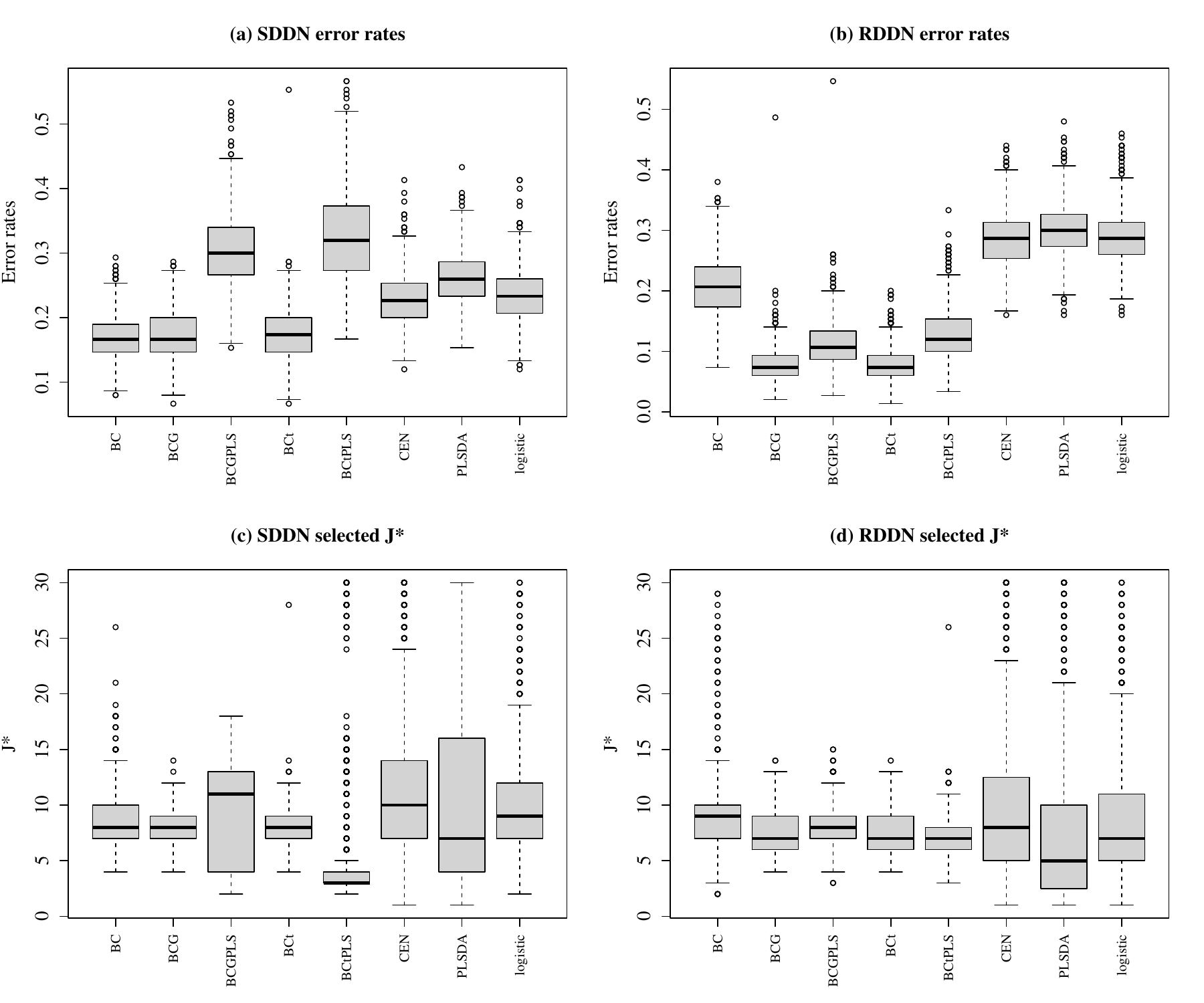}
  \caption{Part (a) and (b) are box plots of the error rates by the eight classifiers in scenarios SDDN and RDDN. The bottom two plots (c) and (d) are box plots of cross-validated $J^*$ in each simulation. \label{simbox}}
\end{figure}

On the other hand, when the group eigenfunctions are different, BC and the three non-Bayes classifiers fail to outperform BCG/BCt in any scenario, even though the group eigenspaces remain equal. BCG maintains its robust performance of lowest error rates throughout all cases. BCt is not far behind, and falls into BCG's MOE $50\%$ of the time as labeled. 

Fig.\ \ref{simbox} compares the misclassification rates and the corresponding $J^*$ selected in each of the $1000$ simulations at two scenarios, SDDN and RDDN.  These two scenarios differ only in their eigenfunction setting. In Plot (a), where the groups have equal eigenfunctions, BC, BCG, and BCt show similar behaviors in classification. In Plot (b), where the group eigenfunctions differ, BCG and BCt have the lowest error rates and variation, followed by BCG-PLS and BCt-PLS. In Plots (c) and (d), we find that BCG and BCt are the only classifiers that have a stable choice of optimal $J^*$: both methods choose $J^* < 10$ more than $75\%$ of the time with few outliers, regardless of whether the group eigenfunctions are equal or not.

\item[\textendash]
\textit{Difference between the group means}.
Under the equal eigenfunction setting, non-Bayes classifiers such as CEN and the logistic regression are naturally sensitive to a location difference, especially when other factors are kept the same; see for example, SDSN, SDST. 
However, in the bottom half of Table \ref{rates}, where the group eigenfunctions differ, BCG shows the strongest performance in all cases, with BCt a close second. 

In this table, the PC-based methods BCG and BCt show an advantage over their PLS counterparts in scenarios with a location difference. That is because $\mu_d$ is effectively captured by PCs. In Section \ref{sec:multiclass}, when the new $\mu_d$ has nonzero projections only on the last several bases, PLS-based classifiers can do a better job than other methods in distinguishing such a difference, as mentioned in \cite{DH2012}. This phenomenon is also discussed in Section \ref{realdata}.

\item[\textendash]
\textit{Difference in group eigenvalues and score distributions}. 
In general, we find that the marginal densities of the scores and their eigenvalues have similar effects on the classifiers' performance. They contribute to the difference of the functional distributions in each group, which the three non-Bayes methods (CEN, PLSDA, logistic) fail to detect. For all scenarios in Table \ref{rates} without a location difference, CEN, PLSDA, and the logistic regression all show very poor performance, with error rates close to $50\%$.

\end{itemize}

The two right-most columns in Table \ref{rates} show that the CV-selected method achieves comparable performance to the optimal result of each scenario. This demonstrates the stability and strength of our copula-based Bayes classifiers, especially under the unequal eigenfunction setting. Sections \ref{sec:corrRSDN} and \ref{sec:corrRSDT} in the Supplementary Material report the correlations between the first $10$ scores in the scenarios RSDN and RSDT, respectively.  These high correlations are consistent with the strong performance of the copula-based classifiers in the scenarios where the two groups have different eigenfunctions.

\subsection{Multiclass Classification Performance} \label{sec:multiclass}

We also investigate the performance of the aforementioned methods in terms of classifying data into more than two labels, because the group eigenfunctions from multiple different classes are more likely to be unequal, making it increasingly necessary to consider the dependency of the scores on the joint basis.

We now denote the group labels as $Y=k$, for $k=0, 1 , 2$, and set up the multiclass scenarios following the design in Section \ref{design}. The first column in Table \ref{table:multiclass} lists the $12$ scenarios considered. The first letter $M$ labels unequal group eigenfunctions: when $Y=0$ and $1$, the group eigenfunctions are the Fourier basis and its rotated counterpart, respectively, as described in type R of Factor 1 for binary data; when $Y=2$, the group basis is again the rotated Fourier functions on $\mathcal{T}=[0,1]$, but the rotation angle factor used in iii) of Factor 1 in Section \ref{design} is now $\pi/4$ instead of $\pi/3$. We omit cases of equal group eigenfunctions, because similar results can be found in the binary setup, and the likelihood of an unequal basis increases as the levels of $Y$ increase.

The second letter S or D again denotes equal group means or not, respectively. When the group means $\mu_k$ are unequal (labeled D), we set $\mu_0=0$, $\mu_1$ is the identity function used previously, and $\mu_2 = \sum_{j=192}^{201}\phi_{j0}$. The function $\mu_2$ follows a similar design to that of \cite{DH2012}, where the group mean only has nonzero weights on the last three of $40$ eigenfunctions. We assign the nonzero weights to the last $10$ of the $201$ bases.

Similarly, S or D in the third position represents the same or different group eigenvalues, respectively. When the group eigenvalues are equal, $\lambda_{jk}=10/j^2$ for all $k$; otherwise, $\lambda_{jk}=10/j^2, 10/j^3, 10/j$, respectively, for $k=0, 1, 2$, for $j\ge 1$. Finally,  the last letter inherits the design from Factor 4 of Section \ref{design} to describe the standardized score distribution patterns: similarly to the binary case, N and T denote the Gaussian and skewed distributions, respectively, for all three levels, while for V, we define the scores $\epsilon_{ijk}$ to follow a standard Gaussian, centered exponential with rate one, or skewed distribution in T, for $k=0, 1, 2$ respectively.

The other setup details of the noise, data pre-smoothing, and bandwidth selection are all similar to Section~\ref{design} for binary data. For each simulation, we have $100$ training and $150$ test cases. The optimal cut-off $J^*$ is selected using cross-validation from $J\le 10$. Table~\ref{table:multiclass} presents the misclassification rates from $1000$ Monte Carlo repetitions by seven of the eight classifiers in Section \ref{cltypes}. Note that functional centroid classifier is not applicable to multiclass data, and thus is excluded here. As in the binary case, the Supplementary Material Table \ref{sp:multilevelV500} includes additional results with an increased training size and a different set of score distributions (V).

\begin{table}[h!]
\centering
\resizebox{\columnwidth}{!}{
\begin{tabular}{r|r|rrrr|rr|rr}
  \hline
 & BC & BCG & BCGPLS & BCt & BCtPLS & PLSDA & logistic & CV & Ratio(CV) \\ 
  \hline
MSSN & 0.520 &\textbf{ 0.325} & 0.392 &\textit{ 0.327} & 0.392 & 0.641 & 0.637 & 0.328 & 0.89\%\\ 
  MDSN & 0.356 & 0.247 &\textit{ 0.237} & 0.245 &\textbf{ 0.235} & 0.446 & 0.427 & 0.226 & -3.88\%\\ 
  MSDN & 0.213 &\textit{ 0.169} & 0.281 &\textbf{ 0.168} & 0.310 & 0.636 & 0.618 & 0.173 & 3.00\%\\ 
  MDDN & 0.194 &\textit{ 0.156} & 0.272 &\textbf{ 0.156} & 0.295 & 0.540 & 0.509 & 0.157 & 1.11\%\\ 
  \hline
  MSST & 0.560 &\textbf{ 0.450} & 0.503 &\textbf{ 0.450} & 0.492 & 0.635 & 0.638 & 0.456 & 1.25\%\\ 
  MDST & 0.343 &\textbf{ 0.286} & 0.303 &\textbf{ 0.286} & 0.333 & 0.424 & 0.364 & 0.284 & -0.72\%\\ 
  MSDT & 0.449 &\textit{ 0.399} & 0.444 &\textbf{ 0.397} & 0.467 & 0.624 & 0.616 & 0.401 & 0.95\%\\ 
  MDDT & 0.342 & 0.297 & 0.355 &\textbf{ 0.287} & 0.403 & 0.483 & 0.401 & 0.293 & 2.38\%\\ 
  \hline
  MSSV & 0.325 &\textbf{ 0.259} & 0.394 &\textit{ 0.261} & 0.475 & 0.633 & 0.615 & 0.264 & 2.23\%\\ 
  MDSV & 0.288 &\textit{ 0.237} & 0.356 &\textbf{ 0.234} & 0.433 & 0.436 & 0.399 & 0.241 & 2.93\%\\ 
  MSDV & 0.385 & 0.314 & 0.427 &\textbf{ 0.302} & 0.435 & 0.631 & 0.627 & 0.311 & 3.00\%\\ 
  MDDV & 0.272 &\textit{ 0.223} & 0.322 &\textbf{ 0.219} & 0.340 & 0.475 & 0.434 & 0.224 & 2.18\%\\ 
   \hline
\end{tabular}}
\caption{Misclassification rates averaged over $1000$ simulations of the seven classifiers on $12$ multiclass data scenarios. Best case in each scenario is in bold, and cases within margin of error of the lowest are in italic. $P(Y=k)=1/3$, for $k=0,1,2$, so the true misclassification rate of any method is approximately $0.667$.}
\label{table:multiclass}
\end{table}

Table \ref{table:multiclass} indicates that for data of multiple labels, the behaviors of the seven classifiers follow a similar pattern to that of the binary case when the group eigenfunctions are unequal. In particular, BCt shows strength under increased data complexity, followed closely by BCG. BCG-PLS/BCt-PLS also prove their advantage in detecting location differences on minor basis functions in MDSN. Although they fail to outperform their PC-based counterparts under more complicated scenarios such as MDST and MDSV, we believe this is because the group means are not the only dominant difference in these two data cases.

Tables \ref{rates} and \ref{table:multiclass} give us clear guidelines that deciding whether or not to use copulas in a classification makes a more significant impact on the outcome than the type of copulas, because both BCG and BCt present competitive performance. The tables also reveal the strength of copula-based methods in dimension reduction. Classifiers using copulas are able to achieve high accuracy with small cut-off $J^*$, which indicates their advantage in small samples. 
In addition, in general, PCs are preferable to PLS, owing to their robustness and simplicity of implementation. BCG-PLS and BCt-PLS should be considered when the group mean difference is significant and located at minor eigenfunctions, which we discuss further in the real-data examples.

\setcounter{equation}{0}
\section{Real-Data Examples} \label{realdata}
In this section, we use two real-data examples to illustrate the strength of our new method in terms of classification and dimension reduction with respect to the data size $n$. 

\subsection{Classification of Multiple Sclerosis Patients}\label{DTIandMS}
Our first example explores the classification of multiple sclerosis (MS) cases based on FA profiles of the cca tract.  FA is the degree of anisotropy of water diffusion along a tract, and is measured by diffusion tensor imaging
(DTI).  Outside the brain, water diffusion is isotropic ( \cite{G2012}).
MS is an autoimmune disease  leading to lesions in white matter tracts such as the cca.  These lesions decrease \mbox{FA}.

The DTI data set in the R package {\tt refund} (\cite{refund}) contains FA profiles at $93$ locations on the cca of $142$ subjects. 
The data  were collected at Johns Hopkins University and the Kennedy\textendash Krieger Institute. 
The numbers of visits per subject range from one to eight, but we used the $142$ FA curves from the first visits only.  One subject with partially missing FA data was removed. Among the $141$ subjects, $42$ are healthy ($k=0$) and $99$ were diagnosed with MS ($k=1$).   We use local linear regression for data pre-smoothing. To determine the optimal number of dimensions $J^*$ for each method, we use cross-validation with maximal $J=30$. The misclassification rates from using 10-fold cross-validation were recorded for 1000 repetitions. 


As discussed in Section~\ref{sec:intro}, Panel (a) in Fig.\ \ref{ccamusd} plots $5$ FA profiles from each group, and panels (b) and (c) display the group means and standard deviations of the cases and controls, using raw and pre-smoothed data.  Compared with the controls, MS patients have lower mean FA values and greater variability. We see that smoothing removes some noise.

\begin{table}[h!]
\centering
\begin{tabular}{r|rrrrrrrr}
  \hline
 Method & BC & BCG & BCGPLS & BCt & BCtPLS & CEN & PLSDA & logistic \\ 
  \hline
Error Rate & 0.228 & 0.199 & 0.211 &\textbf{ 0.192} & 0.211 & 0.264 & 0.219 & 0.216 \\ 
\hline
\end{tabular}
\caption{ \label{tab:ccatable} Average misclassification rates of eight functional classifiers by 1000 repetitions of 10-fold CV. BCt has the best performance. The best case is in bold.}
\end{table}

\begin{figure}[h!]
\centering
      \includegraphics[scale=.43]{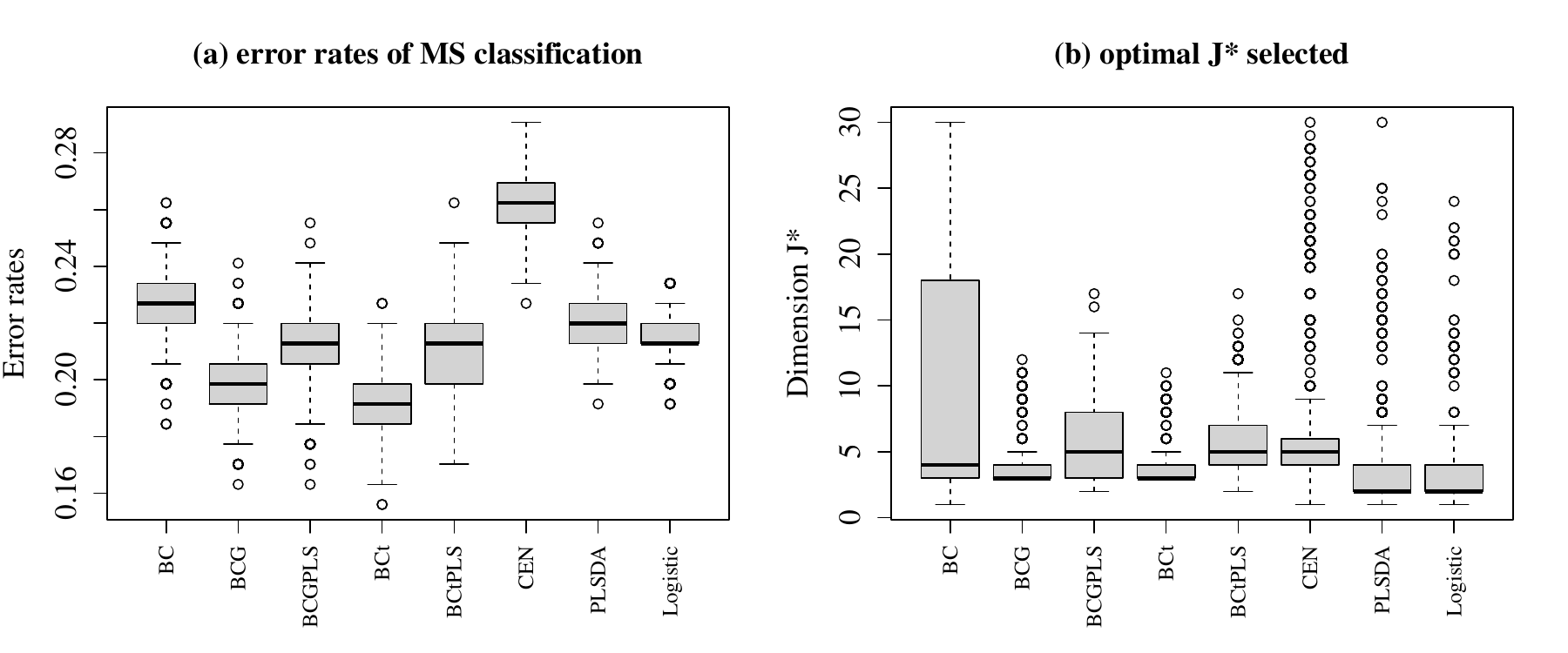}
  \caption{Box plots of misclassification rates and optimal number of components $J^*$ in the MS study over 1000 repetitions of 10-fold cross-validation. BCt achieves the lowest average error rate, while requiring a very small number of components ($J^* < 5$) with lowest variation.}
  \label{ccabox}
\end{figure}

As shown in Table \ref{tab:ccatable} and Part (a) of Fig.\ \ref{ccabox}, BCt achieves the lowest error rate at $0.192$, with a margin of error $0.0007$. The rates of the other methods fail to fall into this range, and are all significantly higher than that of BCt.
In fact, the third quartile for BCt is below the first quartile of all other methods, except BCG. Part (b) is a box plot of cross-validated $J^*$ during each simulation for all classifiers. Here, BCt and BCG achieve the lowest error rates, with a minimal number of dimensions. In addition, compared with methods such as CEN, PLSDA, or logistic regression, their choice of optimal $J^*$ is very stable, with the smallest variation and few outliers. In contrast, BC is prone to employing a large number of components in classification. This tendency can be found in other examples too.

In the Supplementary Material, we compare the loadings (\ref{ccadim}), score distributions (\ref{ccadensity}, and group eigenfunctions (\ref{ccagroup}) between using PC and PLS. The difference explains why PC is a better choice for this example. 
Note that it is not our intent to develop DTI as a technique for diagnosing \mbox{MS}.  
DTI is too expensive and time-consuming for that purpose.
Instead, we are looking for differences in FA between cases and controls, because these  could inform researchers about the nature of the disease.
We have found clear differences between cases and controls in the mean and variance of FA.  The strong positive correlation between the second and the third PC scores in the healthy cases (Spearman's $\rho$ at $0.525$ and an adjusted $p$-value $2\times 10^{-2}$) is diminished in the MS group. BCt and BCG are best able to use a compact model to capture subtle differences, such as correlations.

\subsection{Particulate Matter (PM) Emission of Heavy-Duty Trucks} \label{PMex}

As a second example, we investigate the relationship between the movement patterns of heavy-duty trucks and particulate matter (PM) emissions.  We use the data in \cite{M2015}, originally extracted from the Coordinating Research Council E55/59 emissions inventory program documentary (\cite{C2007}). The data set contains $108$ records of truck speed in miles/hour over $90$-second intervals, and the logarithms of their PM emission in grams (log PM), captured by $70$ mm filters. 

We dichotomize \mbox{log PM}.  The 41 of 108 cases with log PM above average are called high emission ($k=1$), and the other cases are low emission ($k=0$). We classify log PM level using the  $90$-second velocity profiles. The misclassification rates are estimated using $10$-fold cross-validation, repeated $1000$ times.

\begin{figure}[h!] 
\centering
      \includegraphics[scale=0.35]{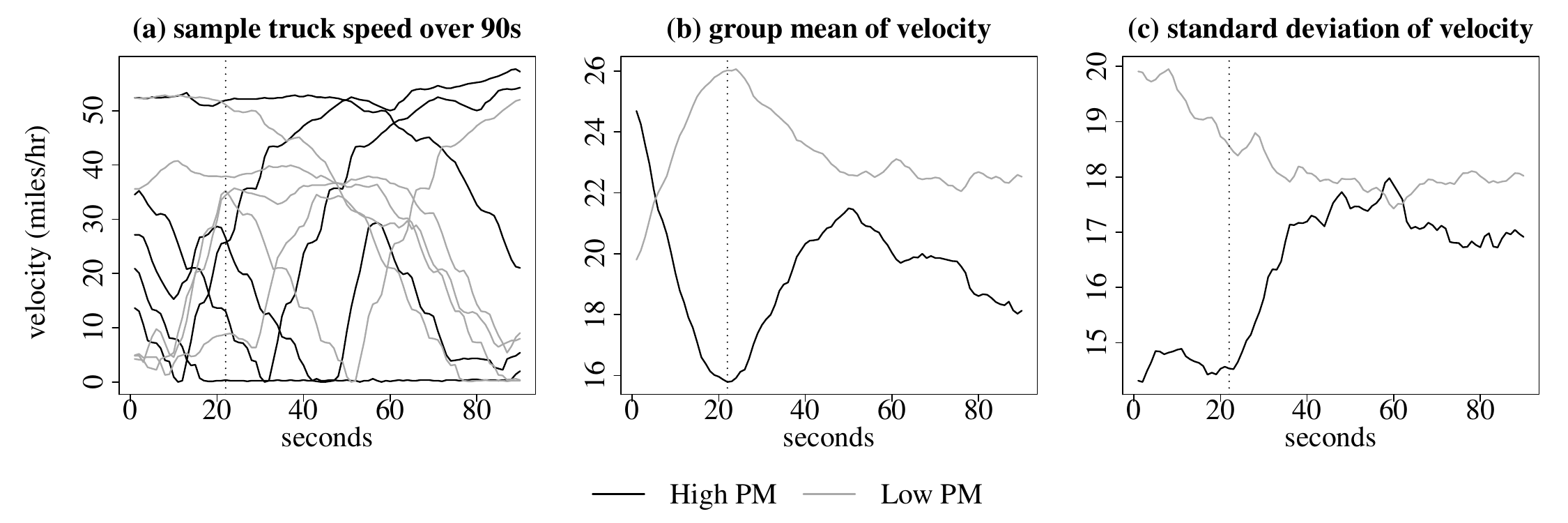}
  \caption{Plots of five sample paths in each PM group, as well as group mean and standard deviation of truck velocity data. On average, trucks in high PM group have lowest speed at 22 seconds,  marked with a dashed line on each plot.}
  \label{truckmusd}
\end{figure}

As Fig.\ \ref{truckmusd} shows, during the first $20$ seconds, vehicles in the high PM group, on average, decelerate to a minimum speed, whereas the low PM group tends to speed up. The high PM group also has much lower variation than the low PM group. 

\begin{table}[ht]
\centering
\begin{tabular}{r|rrrrrrrr}
  \hline
 & BC & BCG & BCGPLS & BCt & BCtPLS & CEN & PLSDA & logistic \\ 
  \hline
Error rate & 0.285 & 0.280 &\textbf{ 0.207} & 0.280 & \textbf{ 0.207} & 0.278 & 0.256 & 0.228 \\ 
   \hline
\end{tabular} 
\caption{ \label{tab:trucktable} Average misclassification rates of eight functional classifiers by 1000 repetitions of 10-fold cross-validation. BCt-PLS and BCG-PLS have the best performance. The best cases are in bold.}
\end{table}

\begin{figure}[h!]
\centering
      \includegraphics[scale=.43]{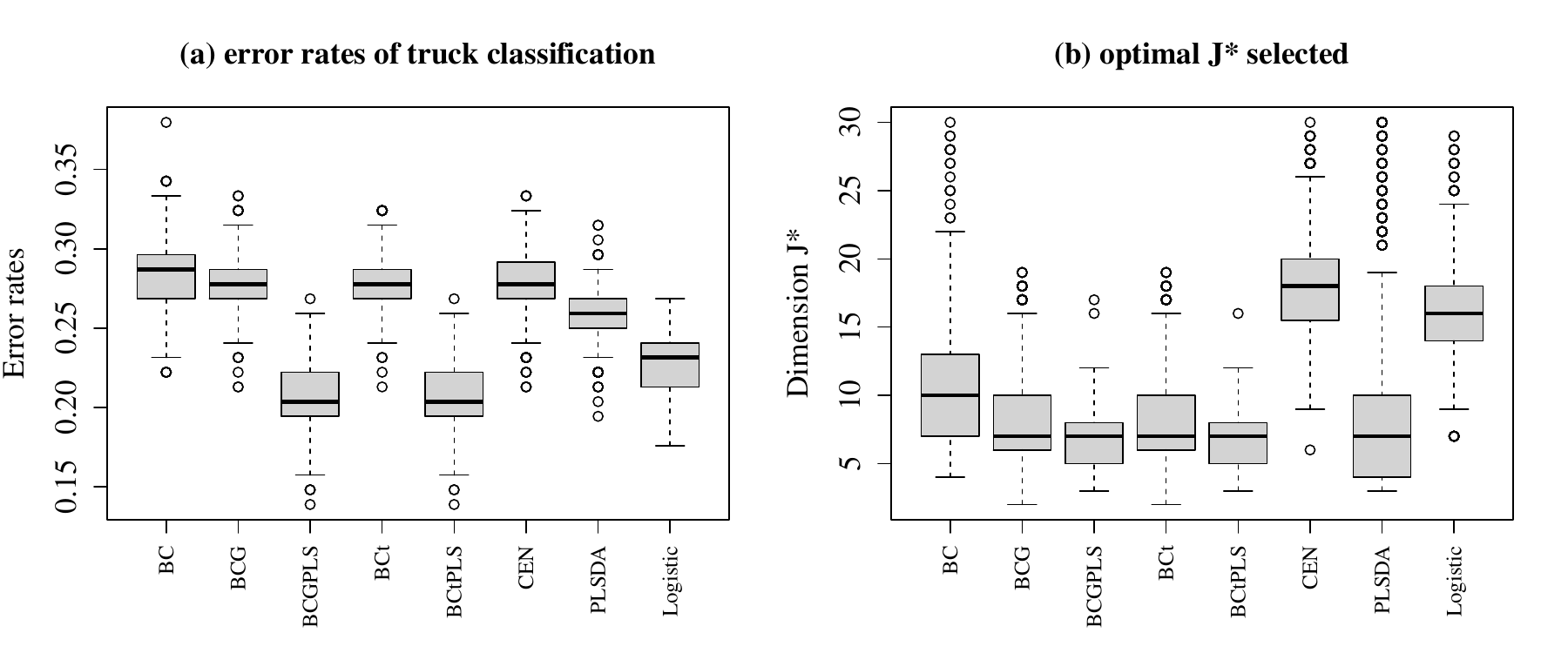}
  \caption{Box plots of misclassification rates and optimal number of components $J^*$ in the truck emission case over 1000 repetitions of 10-fold cross-validation. BCt-PLS and BCG-PLS achieve the lowest average error rate with $J^*$ concentrated around $7$.}
\label{truckboxplot}
\end{figure}

From Fig.\ \ref{truckboxplot} and Table \ref{tab:trucktable}, BCG-PLS and BCt-PLS have the lowest misclassification rates.  The third quartiles of their error rates are below the first quartiles of the other classifiers, except for the logistic regression. In addition, both methods keep the classification model compact by requiring small $J^*$ with low variation. BC and the three methods on the right of plot (b) of Fig.\ \ref{truckboxplot} again demand more components with bigger variation in classification. In Section \ref{sup:dataexample} of the Supplementary Material, we include additional results for both data examples to validate their different choices of PC- and PLS-based classifiers.

\setcounter{equation}{0}
\section{Theoretical Asymptotic Properties} \label{theory}
An interesting feature of functional classifiers is  asymptotic perfect classification. That is, under certain conditions, the error rate goes to zero as $J \to \infty$, owing to the infinite-dimensional nature of functional data (\cite{DH2012}). \cite{DMY2017} discussed the perfect classification by BC under equal group eigenfunctions. In this section, we prove that when the group eigenfunctions differ, perfect classification is retained by our classifier $\mathds{1}\{\log Q^*_J(X)>0\}$ for both Gaussian and non-Gaussian processes. The  scores $X_{\cdot j k}$, for $1 \le j \le J$, in this section are all projected onto joint eigenfunctions $\phi_1, \ldots, \phi_J$.

We first show that $\log Q^*_J\left(X\right)$ and the estimated $\log \hat{Q}^*_J\left(X\right)$ are asymptotically equivalent under mild conditions. Then, the behavior of the Bayes classifier $\mathds{1}\{\log Q^*_J(X)>0\}$ is studied in two settings: first, when the random function $X_{\cdot \cdot k}$ is a Gaussian process for both $k=0, 1$; and second, the more general case, when $X$ is non-Gaussian, but its projected scores are meta-Gaussian distributed in each group. For simplicity, we assume here that $\pi_1=\pi_0$. 

\subsection{Asymptotic equivalence of $\log \hat{Q}^*_J\left(X\right)$ and $\log Q^*_J\left(X\right)$}
We first list several assumptions, which help establish the asymptotic equivalence of both the marginal and the copula density components of $\log \hat{Q}^*_J\left(X\right)$ and $\log Q^*_J\left(X\right)$. 

\begin{assump} \label{kdeA1} For all 
$\, C >0$ and some $\delta > 0$: $\sup_{t \in \mathcal{T}} E\{|X(t)|^C\} < \infty$, 
\newline $\sup_{s,t \in \mathcal{T}: s \ne t} E[\{|s-t|^{-\delta} |X(s)-X(t)|\}^C] < \infty$.
\end{assump}

\begin{assump} \label{kdeA2}
For integers $r \ge 1$, $\lambda_j^{-r}E[\int_{\mathcal{T}}\{X-E(X)\}\phi_j]^{2r}$ is bounded uniformly in $j$.
\end{assump}

\begin{assump} \label{kdeA3}
There are no ties among the eigenvalues $\{\lambda_j\}_{j=1}^{\infty}$.
\end{assump}

\begin{assump} \label{kdeA4}
The density $g_j$ of the $j$th standardized score $\langle X-E(X), \phi_j \rangle/\sqrt{\lambda_j}$ is bounded and has a bounded derivative; for some $\delta>0$, $h=h(n)=O(n^{-\delta})$ and $n^{1-\delta}h^3$ is bounded away from zero as $n \to \infty$. The ratio $f_{j1}(X_{\cdot j \cdot})/f_{j0}(X_{\cdot j \cdot})$ is atomless for all $j \ge 1$.
\end{assump}

For all $c>0$, let $\mathcal{S}(c)=\{x \in \mathcal{L}^2(\mathcal{T}): \|x\| \le c\}$. Assumptions \ref{kdeA1}\textendash\ref{kdeA4} are from \cite{DH2010}, adapted here to bound the difference $D_{jk}\left(x_j\right)=\hat{g}_{jk}\left(\hat{x}_j\right)-\bar{g}_{jk}\left(x_j\right)$ s.t.\ $\sup_{x \in \mathcal{S}(c)} |D_{jk}\left(x_j\right)|=op\{(nh)^{-1/2}\}$.
We let $\hat{g}_{jk}\left(\hat{x}_j\right)=1/\left(n_kh\right)\sum_{i=1}^{n_k} K\left \{\langle X_{i \cdot k}-x, \hat{\phi}_j\rangle/\left(\hat{\sigma}_{jk}h\right)\right\}$ be the estimated density of the standardized scores of group $k$ on basis $\hat{\phi}_j$, with $\bar{g}_{jk}\left(x_j\right)$ using $\phi_j$ and $\sigma_{jk}$. In addition, the following assumption is added for $D_{jk}\left(x_j\right)$, for both $k=0, 1$:
\begin{assump} \label{kdeA5}
$\sup_{x \in \mathcal{S}(c)} \left|\hat{\pi}_k D_{jk}\left(x_j\right)/\left(\hat{\pi}_0 D_{j0}\left(x_j\right)+\hat{\pi}_1 D_{j1}\left(x_j\right)\right)\right|=Op\left(1+\sqrt{\dfrac{\log n}{nh^3}}\right).$
\end{assump}

We use \ref{kdeA5} to give a mild bound simply to avoid the case where the magnitudes of both $D_{jk}\left(x_j\right)$, for $k=0, 1$, are too large and close, but with opposite signs. \ref{kdeA5} guarantees that the difference between the estimated marginal density $\hat{f}_{jk}\left(\hat{x}_j\right)$ and $f_{jk}\left(x_j\right)$ is able to be bounded by the same rate as when the group eigenfunctions are equal. However, this is not a necessary condition for the asymptotic equivalence of $\log \hat{Q}^*_J (X)$ and $\log Q_J^*(X)$, and we can certainly relax its bound for Theorem \ref{SJ} below.

Then, $\hat{f}_{jk}\left(\hat{x}_j\right)=\left(1/\hat{\sigma}_{jk}\right) \hat{g}_{jk}\left(\hat{x}_j\right)$, and we have Proposition \ref{prop1} (see the Supplementary Material for the proof):
\begin{proposition} \label{prop1}
Under Assumptions \ref{kdeA1}\textendash\ref{kdeA5}, when the group eigenfunctions are unequal, the estimated marginal density $\hat{f}_{jk}$ using scores $\langle X_{i \cdot k}, \hat{\phi}_j \rangle$ achieves an asymptotic error bound:
$
\sup_{x \in \mathcal{S}(c)} |\hat{f}_{jk}(\hat{x}_j)-f_{jk}(x_j)| = Op\left \{h+\sqrt{\dfrac{\log n}{nh}} \right \},
$
where the rate is the same as in \cite{DMY2017}, where the group eigenfunctions are equal.
\end{proposition}

\begin{assump} \label{A2}
The CDFs $F_{jk}$ of scores $X_{\cdot j k}$ are continuous and strictly increasing, with correspondent marginal densities $f_{jk}$ continuous and bounded. In addition, $f_{jk}$ are bounded away from zero on any compact interval within their supports.
\end{assump}
\ref{A2} ensures that the scores $X_{\cdot j k}$ and their monotonic transformations are atomless; this also follows Condition 5 in \cite{DMY2017}.

Then, in addition to the marginal densities, we establish the equivalence of $\mathbf{\Omega}_k^{-1}$ and $\hat{\mathbf{\Omega}}_k^{-1}$ in $\log Q^*_J\left(X\right)$ and $\log \hat{Q}^*_J\left(X\right)$, respectively, as $n \to \infty$. As mentioned in Section \ref{section:copula}, we calculate $\hat{\mathbf{\Omega}}_k$ using rank correlations. In addition, when $J$ is large, the inverse of $\hat{\mathbf{\Omega}}_k$ can be estimated using the graphical Dantzig selector (\cite{yuan10}), which solves the matrix inverse by connecting the entries of the inverse correlation matrix to a multivariate linear regression, and exploits the sparsity of the inverse matrices (\cite{yuan10}). \cite{Liu} provided a $q$-norm $Op$ bound of the difference between the inverse Gaussian copula matrix and its estimation by the Dantzig estimator for high-dimensional problems, and is extended here for the difference between  $\mathbf{\Omega}_k^{-1}$ and $\hat{\mathbf{\Omega}}_k^{-1}$. 

Our sparsity assumptions on the inverse correlation matrices  follow the design of \cite{yuan10} and \cite{Liu}: let $\mathbf{\Omega}_k$ belong to the class of matrices $\mathcal{C}\left(\kappa, \tau, M, J\right):=\{\mathbf{\Omega}^{J \times J}: \mathbf{\Omega} \succ \mathbf{0}, \text{diag} (\mathbf{\Omega})=\mathbf{1}, \|\mathbf{\Omega}^{-1}\|_1 \le \kappa, \dfrac{1}{\tau} \le \lambda_{\min} (\mathbf{\Omega}) \le \lambda_{\max} (\mathbf{\Omega}) \le \tau, \text{deg} (\mathbf{\Omega}^{-1}) \le M \}$, where $\kappa, \tau \ge 1$ are constants determining the tuning parameter in the graphical Dantzig selector, and the parameter $M$ bounding deg$(\mathbf{\Omega}^{-1})=\max_{1\le j \le J} \sum_{j'=1}^{J} I(\mathbf{\Omega}^{-1}_{jj'} \ne 0)$ is dependent on $J$.  Assuming these sparsity conditions, we have the following theorem.

\begin{theorem} \label{SJ}
Under \ref{kdeA1}\textendash\ref{A2}, $\forall \epsilon > 0$, as $n \to \infty$, there exists a sequence $J\left(n, \epsilon, M\right) \to \infty$, and a set $S$ dependent on $J\left(n, \epsilon, M\right)$, $P\left(S\right) \ge 1-\epsilon$, such that 
\[
P\left(S \cap \left\{\mathds{1}\left\{\log \hat{Q}^*_J\left(X\right) \ge 0 \right\} \ne \mathds{1}\left\{\log Q^*_J\left(X\right) \ge 0\right\}\right\}\right) \to 0,
\] 
provided that $MJ\sqrt{\log J}=o\left(\sqrt{n}\right)$. 
\end{theorem}

Theorem \ref{SJ} proves that under unequal group eigenfunctions, $\log \hat{Q}_J^*\left(X\right)$ using copulas retains the property in Theorem A1 of \cite{DMY2017} for the estimated Bayes classifiers with equal group eigenfunctions and independent scores: as $n \to \infty$, $\log \hat{Q}_J^*\left(X\right)$ gets arbitrarily close to the true Bayes classifier $\log Q_J^*\left(X\right)$, which enables us to discuss the performance of our method using the properties of the true Bayes classifier.

\subsection{Perfect classification when $X$ is a Gaussian process in both groups}

Let $X_{\cdot \cdot k}$ be a centered Gaussian process such that $X_{\cdot \cdot k}=\sum_{q=1}^{\infty} \sqrt{\lambda_{qk}} \xi_{qk} \phi_{qk}$, with $\xi_{qk} \sim N(0,1)$, for $k=0, 1$. We denote the $J \times J$ covariance matrix of scores $X_{\cdot j k}$, for $1 \le j \le J$, as $\mathbf{R}_k$, where its $\left(j, j'\right)$th entry is equal to
$
\text{cov} \left(X_{\cdot j k}, X_{\cdot j' k}\right)=\sum_{q=1}^{\infty} \lambda_{qk} \langle \phi_{qk}, \phi_j\rangle \langle \phi_{qk}, \phi_{j'}\rangle, 
$
and its eigenvalues are $d_{1k}, \ldots, d_{Jk}$. Let $\vec{\mu}_J$ be a length-$J$ vector $\left(\mu_1, \ldots, \mu_J\right)^T$ by projecting $\mu_d$ on first $J$ bases, $\mu_j=\langle \mu_d, \phi_j\rangle$. By the law of total covariance and the result that the trace of a matrix is equal to the sum of its eigenvalues, we derive the following relationship between the two sets of eigenvalues  (i.e.\, $\lambda_j$, $\lambda_{jk}$, and $d_{jk})$:
$
\sum_{j=1}^J \lambda_j=\pi_1 \sum_{j=1}^J d_{j1}+\pi_0 \sum_{j=1}^J d_{j0}+\pi_1\pi_0\sum_{j=1}^J \mu_j^2,
$
and
$
\sum_{j=1}^J d_{jk}=\sum_{j=1}^J \sum_{q=1}^{\infty} \lambda_{qk}\langle \phi_{qk}, \phi_j \rangle^2.
$
The following assumption is standard in functional data for the distribution of $X$, and ensures that $d_{jk} >0$, for $1 \le j \le J$, $k=0, 1$:
\begin{assump}\label{A1}
Both the group covariance operators, $G_1$, $G_0$, and the covariance matrices $\mathbf{R}_0$, $\mathbf{R}_1$ are bounded and positive definite, and $\mu_d \in \mathcal{L}^2(\mathcal{T})$. 
\end{assump}

When $X$ is Gaussian in both groups, $\log Q_J^*(X)$ is a quadratic form in $\mathbf{X}_J$ ($\mathbf{X}_J$ is a length-$J$ vector with $j$th entry $\langle X, \phi_j\rangle$):
\begin{equation} \label{quad}
\log Q^*_J(X)=-\dfrac{1}{2}\left(\mathbf{X}_J-\vec{\mu}_J\right)^T \mathbf{R}_1^{-1}\left(\mathbf{X}_J-\vec{\mu}_J\right)+\dfrac{1}{2}\mathbf{X}_J^T \mathbf{R}_0^{-1}\mathbf{X}_J+\log \sqrt{\dfrac{|\mathbf{R}_0|}{|\mathbf{R}_1|}}.
\end{equation}
With potentially unequal group eigenfunctions, entries in $\mathbf{X}_J$ at $Y=k$ can be correlated, which complicates the distribution of $\log Q^*_J(X)$ in each group. 

Therefore, we implement a linear transformation of $\mathbf{X}_J$ in Steps i)\textendash iii):

\begin{enumerate}[i)]
\item
The eigendecomposition of the matrix product gives $\mathbf{R}_0^{1/2} \mathbf{R}_1^{-1} \mathbf{R}_0^{1/2}=\mathbf{P}^T \mathbf{\Delta} \mathbf{P}$, where $\mathbf{\Delta}=$ diag$\{\Delta_1, \ldots, \Delta_J\}$, $\Delta_j$ as eigenvalues of $\mathbf{R}_0^{1/2} \mathbf{R}_1^{-1} \mathbf{R}_0^{1/2}$. By the equivalence of the determinants, $\prod_{j=1}^J \dfrac{d_{j0}}{d_{j1}}=\prod_{j=1}^J\Delta_j$. In addition, $\Delta_j > 0$, for all $j$, under \ref{A1};
\item
Let $\mathbf{Z}=\mathbf{R}_0^{-1/2}\mathbf{X}_J$, $\mathbf{U}=\mathbf{P}\mathbf{Z}$;
\item
When $k=0$, the $j$th entry $U_j$ of the vector $\mathbf{U}$ has a standard Gaussian distribution; at $k=1$, $U_j \sim N(-b_j, 1/\Delta_j)$, with $b_j$ the $j$th entry of $\mathbf{b}=-\mathbf{P}\mathbf{R}_0^{-1/2}\vec{\mu}_J$. 
\end{enumerate}

Consequently, the entries of $\mathbf{U}$ are uncorrelated for both $k=0$ and $1$, Eq.(\ref{quad}) becomes
\begin{equation*} 
\log Q^*_J(X)=-\dfrac{1}{2}\sum_{j=1}^J \Delta_j \left(U_j+b_j\right)^2+\dfrac{1}{2}\sum_{j=1}^J U_j^2 + \dfrac{1}{2}\sum_{j=1}^J \log \Delta_j,
\end{equation*} 
and the asymptotic behaviors of the Bayes classifier for Gaussian processes are concluded.

\begin{theorem} \label{theorem1}
With \ref{A1},  when the random function $X$ is a Gaussian process at both $Y=0$ and $1$ and the group eigenfunctions of $G_0$, $G_1$ are unequal, the functional Bayes classifier $\mathds{1}\{\log Q^*_J(X)>0\}$ achieves perfect classification when either $\|\mathbf{R}_0^{-1/2}\vec{\mu}_J\|^2 \to \infty$, or $\sum_{j=1}^J (\Delta_j-1)^2 \to \infty$, as $J \to \infty$. Otherwise, its error rate err$\left(\mathds{1}\{\log Q^*_J(X)>0\}\right) \not \to 0$.
\end{theorem}

Theorem \ref{theorem1} is a natural extension of Theorem 2 in  \cite{DMY2017}. It again reveals that the error rate of the Bayes classifier approaches zero asymptotically when $\Pi_1$ and $\Pi_0$ are sufficiently different in terms of either the group means or the scores' variances. In addition, recognizing the different correlation patterns between group scores helps improve the classification accuracy. Instead of adopting $\mu_j/\sqrt{\lambda_{j0}}$ and $\lambda_{j0}/\lambda_{j1}$ to build conditions for perfect classification, as in \cite{DMY2017}, we use the transformed $\mathbf{R}_0^{-1/2}\vec{\mu}_J$ and $\Delta_j$ to accommodate the potentially unequal group eigenfunctions and the dependent scores. For the special case when the eigenfunctions are actually equal, the covariance matrices $\mathbf{R}_k=\text{diag}\{\lambda_{1k}, \ldots, \lambda_{Jk}\}$ with $\Delta_j=\lambda_{j0}/\lambda_{j1}$, and consequently the two conditions in Theorem \ref{theorem1} become the same as those proposed in \cite{DMY2017}. The proof of Theorem \ref{theorem1} is in Section~\ref{sec:suppTheorem2Proof} of the Supplementary Material.

\subsection{When $X$ is a non-Gaussian process}
For non-Gaussian processes, when the projected scores  $X_{\cdot j k}$, for $1 \le j \le J$, fit a Gaussian copula model, that is, they are meta-Gaussian distributed, we derive sufficient conditions in terms of the marginal densities $f_{jk}$ and the score correlations in order to achieve an asymptotically zero misclassification rate. 

First, we let $\mathbf{u}_k=\left(u_{1k}, \ldots, u_{Jk}\right)^T$ be a length-$J$ random vector with $u_{jk}=\Phi^{-1}\left(F_{jk}\left(X_{\cdot j \cdot}\right)\right)$, where $\Phi\left(\cdot\right)$ is the CDF of $N(0, 1)$. When $Y=k$, $\left (u_{jk}|Y=k \right) \sim N(0, 1)$, and var$\left(\mathbf{u}_k|Y=k\right)=\mathbf{\Omega}_k$, as denoted before. Let the eigendecomposition be $\mathbf{\Omega}_k=\mathbf{V}_k \mathbf{D}_k \mathbf{V}_k^T$, with $\mathbf{D}_k$ the diagonal matrix with eigenvalues $\omega_{jk}$, for $j=1, \ldots, J$. On the other hand, $u_{jk}|Y=k'$ follows a more complicated distribution when $k' \ne k$. We denote var$\left(\mathbf{u_k}|Y=k'\right)=\mathbf{M}_k$ with the eigendecomposition $\mathbf{M}_k=\mathbf{U}_k \tilde{\mathbf{D}}_k \mathbf{U}_k^T$, and the eigenvalues of $\mathbf{M}_k$ are $\upsilon_{jk}$, for $j=1, \ldots, J$. 

Therefore, the log density ratio $\log Q^*_J(X)$ in the Bayes classifier with a Gaussian copula can be represented as
\begin{align} 
\log Q^*_J(X)&=\sum_{j=1}^J \log \dfrac{f_{j1}\left(X_{\cdot j \cdot}\right)}{f_{j0}\left(X_{\cdot j \cdot}\right)}+\dfrac{1}{2}\log \dfrac{|\mathbf{\Omega}_0|}{|\mathbf{\Omega}_1|}-\dfrac{1}{2}\mathbf{u}_1^T\left(\mathbf{\Omega}_1^{-1}-\mathbf{I}\right)\mathbf{u}_1+\dfrac{1}{2}\mathbf{u}_0^T\left(\mathbf{\Omega}_0^{-1}-\mathbf{I}\right)\mathbf{u}_0 \nonumber\\
&=\sum_{j=1}^J \log \dfrac{f_{j1}\left(X_{\cdot j \cdot}\right)}{f_{j0}\left(X_{\cdot j \cdot}\right)}\Big /\dfrac{\sqrt{\omega_{j1}}}{\sqrt{\omega_{j0}}}-\dfrac{1}{2}\mathbf{u}_1^T\left(\mathbf{\Omega}_1^{-1}-\mathbf{I}\right)\mathbf{u}_1+\dfrac{1}{2}\mathbf{u}_0^T\left(\mathbf{\Omega}_0^{-1}-\mathbf{I}\right)\mathbf{u}_0. \label{newsum}
\end{align} 
Similarly to \ref{A1}, we make an assumption on the covariances of $\mathbf{u}_k$, conditional on  $Y$:
\begin{assump} \label{A4}
The matrices $\mathbf{\Omega}_k$ and $\mathbf{M}_k$, for $k=0, 1$, are bounded and positive definite.
\end{assump}

Next, we define a sequence of ratios $g_j$, for $j=1, 2, \ldots$, by
$
g_j=\dfrac{f_{j1}\left(X_{\cdot j \cdot}\right)}{f_{j0}\left(X_{\cdot j \cdot}\right)}\Big /\dfrac{\sqrt{\omega_{j1}}}{\sqrt{\omega_{j0}}},
$
where $g_j$ compares the ratio of the marginal densities to the ratio of the eigenvalues of the correlation matrices. In addition, let 
\begin{equation*} 
s_{jk}=\dfrac{\text{var}\left(\langle V_{jk}, \mathbf{u}_k \rangle|Y=k\right)}{\text{var}\left(\langle V_{jk}, \mathbf{u}_k \rangle|Y=k'\right)}=\dfrac{\mathbf{V}_{jk}^T \mathbf{\Omega}_k \mathbf{V}_{jk}}{\mathbf{V}_{jk}^T\mathbf{M}_k \mathbf{V}_{jk}}
=\dfrac{\omega_{jk}}{\sum_{q=1}^J C_{(j, q)k}^2 \upsilon_{qk}},
\end{equation*}
where $C_{(j, q)k}=\langle \mathbf{U}_{qk}, \mathbf{V}_{jk}\rangle$, $\sum_{q=1}^J C_{(j, q)k} =1$, and $\mathbf{U}_{qk}$ and $\mathbf{V}_{jk}$ are the $q$th and $j$th columns, respectively, of the eigenvector matrices $\mathbf{U}_{k}$ and $\mathbf{V}_{k}$. As a result, $s_{jk}$ compares the $j$th eigenvalue of $\mathbf{\Omega}_k$ against a convex combination of the eigenvalues of $\mathbf{M}_k$, the individual weights of which are determined by projecting $\mathbf{V}_{jk}$ onto the eigenvalues of $\mathbf{M}_k$, $\mathbf{U}_{qk}$.

In terms of the sequences $g_j$ and $s_{jk}$,
 for $j=1,2, \ldots$, we derive the following theorem for non-Gaussian processes; the proof is in Section~\ref{sec:proofTheorem3} of the Supplementary Material.

\begin{theorem} \label{theorem2} 
With Assumptions \ref{A2}, \ref{A1}, and \ref{A4}, when the projected scores $X_{\cdot j k}$, for $j=1, \ldots, J$, are meta-Gaussian distributed at each group $\Pi_k$, perfect classification by the Bayes classifier $\mathds{1}\{\log Q^*_J(X)>0\}$ is achieved asymptotically if a subsequence $g^*_r=g_{j_r}$ of $g_j$ exists, with corresponding $s_{j_rk}$, such that one of the following conditions is satisfied as $r \to \infty$:  
\begin{enumerate}[a)]
\item
$g_{j_r}=op (1)$, and $s_{j_r0} \to 0$;
\item
$1/g_{j_r}=op (1)$, and $s_{j_r1} \to 0$; 

or when $g_{j_r}$ has distinct behaviors in subgroups: 
\item 
$g_{j_r}=op (1)$ at $Y=1$, $1/g_{j_r}=op (1)$ at $Y=0$, with both $s_{j_r0}$ and $s_{j_r1}$ $\to 0$;
\item
$1/g_{j_r}=op (1)$ at $Y=1$, and $g_{j_r}=op (1)$ at $Y=0$.
\end{enumerate}
\end{theorem}

Based on the structure of the log density ratio described in Eq.(\ref{newsum}), Theorem \ref{theorem2} discusses the occurrence of perfect classification in two aspects: $g_j$, which mainly depicts the relative magnitude of the score marginal densities at each $k=0, 1$; and $s_{jk}$, which compares the correlation between the scores conditioned at each group. Either part showing enough disparity between groups results in perfect classification. 

For example, in Theorem \ref{theorem2} a), when there exists a subsequence $g_{j_r} \to 0$ in probability, indicating the dominance of the marginal densities by the group $Y=0$, the misclassification tends to occur at $Y=1$. However, as $s_{j_r0} \to 0$, the covariance of $\mathbf{u}_0$ conditioned at $Y=1$ becomes much larger than at $Y=0$. As a result, the nonnegative $\mathbf{u}_0^T \mathbf{\Omega}_0^{-1} \mathbf{u}_0^T$ in Eq.(\ref{newsum}) with large variation when $Y=1$ compensates to eventually avoid misclassifying $X$ to group $0$. When $g_{j_r}$ behaves perfectly, as in case d), where the corresponding group marginal densities are dominant in each subgroup $Y=k$, we do not need to impose requirements on the copula correlation to achieve perfect classification.

\begin{remark}
Theorem \ref{theorem2} provides sufficient, but not necessary conditions for the Bayes classifier to achieve asymptotic perfect classification under unequal group eigenfunctions. Owing to the optimality of the Bayes classifier in minimizing the zero-one loss, various conditions from other functional classifiers to achieve an asymptotically zero error also work here. For example, \cite{DH2012} proposed conditions in terms of group eigenvalues and the mean difference for the functional centroid classifier to reach perfect classification. These also work as sufficient conditions for $\mathds{1}\{\log Q_J^*(X)>0\}$ in our case. With a copula model, which is not found in previous work, Theorem \ref{theorem2} uses the relation between the scores' marginal densities and correlations to reduce the error rate to zero asymptotically.
\end{remark}

\section{Discussion} \label{discussion}
\subsection{Remarks}
Our copula-based Bayes classifiers remove the assumptions of equal group eigenfunctions and independent scores. As our two examples show, it is not uncommon  to have unequal group eigenfunctions (see Fig.\ \ref{ccagroup} and Fig.\ \ref{truckgroup}). The new methods also prove to have stronger performance in terms of dimension reduction than that of the original BC.
Our simulation results prove the strength of our method in distinguishing groups by the differences in their functional means and their covariance functions. 
We examined the two choices of projection directions, PC and \mbox{PLS}.  PLS can detect location differences on eigenfunctions corresponding to smaller eigenvalues. 
We discussed new conditions for the estimated classifier to be asymptotically equivalent to the true Bayes classifier, and for perfect classification to occur. These differ from those of previous works, owing to the unequal group eigenfunction setting. We also imposed sparsity conditions on the inverse of the copula correlations.

\subsection{Future Work}
In future work, we would like to extend the copula-based classification to the problem with multiple functional covariates. Some previous works discuss this situation in the framework of functional generalized models: \cite{crainiceanu2009generalized} proposed a generalized multilevel regression model where there are repeated curve measurements for each subject; \cite{zhu2010bayesian} discussed an FGLM approach for the classification of multilevel functions with Bayesian variable selection; and \cite{li2010generalized} present a generalized functional linear model where there are both functional and multivariate covariates, and use a semiparametric single-index function to model the interaction between them. We plan to approach the problem from a different angle, using functional Bayes classification again, owing to its strong performance in the single functional predictor case. Furthermore, because it is natural to assume that the response depends on the covariates and their interactions, it becomes more important for our method to model the dependency between the projected scores. Another aspect we would like to consider is how to choose a proper functional basis for multiple functional predictors.

\vskip 14pt
\noindent {\large\bf Supplementary Materials}

The Supplementary Materials for this document contain additional results for the simulations, for the fractional anisotropy (FA) example, and for the example using truck emissions.  They also contain proofs of Theorems 1, 2, and 3.

\par
\vskip 14pt
\noindent {\large\bf Acknowledgements}

The authors gratefully acknowledge the helpful feedback from the associate editor and referees.
The MRI/DTI data in the refund package were collected at Johns Hopkins University and the Kennedy\textendash Krieger Institute.

\par


\lhead[]{}\rhead[\fancyplain{}\leftmark\footnotesize]{\fancyplain{}\rightmark\footnotesize{} }
\bibhang=1.7pc
\bibsep=2pt
\fontsize{9}{14pt plus.8pt minus .6pt}\selectfont
\renewcommand\bibname{\large \bf References}
\expandafter\ifx\csname
natexlab\endcsname\relax\def\natexlab#1{#1}\fi
\expandafter\ifx\csname url\endcsname\relax
  \def\url#1{\texttt{#1}}\fi
\expandafter\ifx\csname urlprefix\endcsname\relax\def\urlprefix{URL}\fi

\bibliographystyle{apalike}
\bibliography{references} 

\vskip .65cm
\noindent
Department of Statistics and Data Science, Cornell University
\vskip 2pt
\noindent
E-mail: wh365@cornell.edu
\vskip 2pt

\noindent
School of Operations Research and Information Engineering, and Department of Statistics and Data Science, Cornell University
\vskip 2pt
\noindent
E-mail: dr24@cornell.edu

\setcounter{page}{1}
\setcounter{section}{0}
\setcounter{figure}{0}
\setcounter{table}{0}

\clearpage

\renewcommand{\baselinestretch}{2}
\renewcommand{\sectionmark}[1]{}
\renewcommand{\subsectionmark}[1]{}

\markright{ \hbox{\footnotesize\rm Statistica Sinica: Supplement
}\hfill\\[-13pt]
\hbox{\footnotesize\rm
}\hfill }

\markboth{\hfill{\footnotesize\rm WENTIAN HUANG AND DAVID RUPPERT} \hfill}
{\hfill {\footnotesize\rm COPULA-BASED FUNCTIONAL BAYES CLASSIFICATION} \hfill}

\renewcommand{\thefootnote}{}
$\ $\par \fontsize{12}{14pt plus.8pt minus .6pt}\selectfont


 \centerline{\large\bf Supplementary Materials for ``Copula-Based Functional}
\vspace{2pt}
 \centerline{\large\bf Bayes Classification with Principal Components}
\vspace{2pt}
 \centerline{\large\bf and Partial Least Squares''}
\vspace{.25cm}
 \centerline{WENTIAN HUANG AND DAVID RUPPERT}
\vspace{.4cm}
 \centerline{\it Department of Statistics and Data Science, Cornell University}
\vspace{.55cm}

\def\theequation{S\arabic{section}.\arabic{equation}}
\renewcommand{\thesection}{S\arabic{section}}   
\renewcommand{\thetable}{S\arabic{table}}   
\renewcommand{\thefigure}{S\arabic{figure}}

\fontsize{12}{14pt plus.8pt minus .6pt}\selectfont

\section{Algorithm of Functional Partial Least Squares} \label{sup:pls}
\setcounter{equation}{0}

FPLS consists of these steps:

\begin{enumerate}[(i)]
\item
Begin $\mathbf{X}^0=\left(X_{1\cdot \cdot}^0, \ldots, X_{n\cdot \cdot}^0 \right)^T$, $\mathbf{Y}^0=\left(Y_1^0, \ldots, Y_n^0 \right)^T$ centered at their marginal means; 
\item
At step $j$, $1 \le j \le J$, the $j$-th weight function $w_j$ solves \\$\max_{w_j \in \mathcal{L}^2(\mathcal{T})} \text{cov}^2 \left\{\mathbf{Y}^{j-1}, \langle \mathbf{X}^{j-1}, w_j \rangle \right\}$, such that $\|w_j\|=1$ and $\langle w_j, G(w_{j'}) \rangle=0$  for all $ 1 \le j' \le j-1$. Note that we use $\langle \mathbf{X}^{j-1}, w_j \rangle$ to represent an $n$-dimensional vector with elements $\langle X^{j-1}_{i\cdot \cdot}, w_j \rangle$, $i=1,\ldots, n$.
Optimal weight function $w_j$ here has the closed form $w_j=\dfrac{\sum_i Y_i^{j-1} X_{i \cdot \cdot}^{j-1}}{\|\sum_i Y_i^{j-1} X_{i \cdot \cdot}^{j-1}\|}$. It is a sample estimation of the theoretical weight function used in algorithms like \cite{aguilera2010};
\item
The $n$-vector $\mathbf{S}_{j}=\left(s_{1j}, \ldots, s_{nj}\right)^T$ contains  the  $j$-th scores: $\mathbf{S}_{j}=\langle \mathbf{X}^{j-1}, w_j \rangle$;

\item
The loading function $P_j \in \mathcal{L}^2(\mathcal{T})$ is generated by ordinary linear regression of $\mathbf{X}^{j-1}$ on scores $\mathbf{S}_j$: $P_j(t)=\mathbf{S}_j^T \mathbf{X}^{j-1}\left(t\right)/\|\mathbf{S}_j\|^2$, $t \in \mathcal{T}$. 
Similarly, ${\mathcal{D}}_j=\mathbf{S}_j^T\mathbf{Y}^{j-1}/\|\mathbf{S}_j\|^2$;
\item
Update $\mathbf{X}^{j}(t)=\mathbf{X}^{j-1}(t)-P_j(t) \mathbf{S}_j$, $t \in \mathcal{T}$ and $\mathbf{Y}^j=\mathbf{Y}^{j-1}-{\mathcal{D}}_j\mathbf{S}_j$;
\item
Return to (ii) and iterate for a total of $J$ steps.
\end{enumerate}

\section{A more general procedure for multiclass classification} \label{sp:multiclass}
\setcounter{equation}{0}
We describe a detailed procedure of using the copula-based Bayes classification on data with more than $2$ classes, which is complementary to Section \ref{section: 2.2}. 

Assume the response $Y$ has $K$ potential classes ($K > 2$), and the group mean for each subgroup $k$ is $E\left(X|Y=k)\right) = \mu_k$. $P(Y = k) = \pi_k$ for $k = 0, \ldots, K-1$. Then joint covariance operator $G$ has the kernel $G\left(s, t\right) = \sum_k \pi_k G_k + \sum_k \pi_k \mu_k(s)\mu_k(t) - \mu(s)\mu(t)$, where $\mu = E\left(X\right) = \sum_k\pi_k\mu_k$ is the overall mean. Let the truncated joint eigenfunctions again be $\phi_1, \ldots, \phi_J$. The copula densities $c_k$ and score marginal densities $f_{jk}$ are built similar to the binary case, for each class $k = 0, \ldots, K-1$. Then for a test curve $x$ with $x_j=\langle x, \phi_j\rangle$ as the $j$th projected score on the joint basis, we predict $x$'s class to be $k^*$ where 
\begin{equation}
    k^* = \text{argmax}_k f_k \left (x_1, \ldots, x_J \right ) \pi_k = \text{argmax}_k \pi_k c_k\left \{F_{1k}(x_1), \ldots, F_{Jk}(x_J) \right \} \Pi_{j=1}^J f_{jk}(x_j).
\end{equation}

\section{Additional Details and Outputs of Numerical Study in Section \ref{sim}}
\setcounter{equation}{0}

\subsection{Results with Different Score Distributions (V) and Increased Training Size} \label{sp:v}

\begin{table}[h!]
\centering
\resizebox{\columnwidth}{!}{
\begin{tabular}{r|r|rrrr|rrr|rr}
  \hline
 & BC & BCG & BCGPLS & BCt & BCtPLS & CEN & PLSDA & logistic & CV & Ratio (CV) \\ 
  \hline
  \rowcolor{Gray0}
SSSN & 0.495 & 0.500 & 0.503 & 0.492 & 0.504 & 0.502 & 0.500 & 0.500 & 0.505 & 2.49\% \\ 
  SSDN &\textbf{0.200} & 0.208 & 0.304 & 0.214 & 0.400 & 0.474 & 0.495 & 0.473 & 0.202 & 1.10\% \\ 
  SDSN & 0.276 & 0.272 & 0.274 & 0.273 & 0.275 &\textbf{ 0.237} & 0.279 &0.240 & 0.239 & 0.96\% \\ 
  SDDN & 0.142 &\textbf{ 0.137} & 0.270 &\textbf{ 0.137} & 0.272 & 0.202 & 0.245 & 0.206 & 0.138 & 0.88\% \\
  \hline
  \rowcolor{Gray0}
  SSST & 0.508 & 0.504 & 0.498 & 0.511 & 0.509 & 0.500 & 0.496 & 0.495 & 0.504 & 1.80\%\\ 
  SSDT &\textbf{ 0.414} &\textbf{ 0.414} & 0.426 & 0.421 & 0.454 & 0.492 & 0.498 & 0.496 & 0.415 & 0.24\% \\ 
  SDST & 0.161 &0.158 & 0.183 &\textbf{ 0.153} & 0.205 &\textit{ 0.155} & 0.221 &\textbf{ 0.153} & 0.150 & -1.66\%\\ 
  SDDT & 0.137 & 0.134 & 0.161 &\textbf{ 0.129} & 0.188 & 0.136 & 0.224 & 0.132 & 0.132 & 2.48\%\\ \hline
  SSSV &\textit{ 0.383} &\textbf{ 0.382} & 0.484 &\textbf{ 0.382} & 0.482 & 0.489 & 0.495 & 0.494 & 0.385 & 0.96\% \\ 
  SSDV &\textbf{ 0.187} & 0.195 & 0.326 & 0.199 & 0.402 & 0.468 & 0.498 & 0.476 & 0.189 & 0.71\% \\ 
  SDSV &\textbf{ 0.190} & 0.194 & 0.333 & \textit{0.192} & 0.309 & 0.234 & 0.281 & 0.233 & 0.191 & 0.60\% \\ 
  SDDV &\textbf{ 0.136} & 0.142 & 0.306 & 0.140 & 0.329 & 0.197 & 0.256 & 0.198 & 0.140 & 2.35\% \\ \hline
  \hline
  RSSN & 0.284 &\textbf{ 0.110} & 0.128 &\textbf{ 0.110} & 0.120 & 0.498 & 0.503 & 0.482 & 0.111 & 1.22\% \\ 
  RSDN & 0.251 &\textbf{ 0.050} & 0.097 & 0.053 & 0.123 & 0.490 & 0.494 & 0.474 & 0.051 & 3.08\% \\ 
  RDSN & 0.248 &\textit{ 0.090} & 0.099 &\textbf{ 0.089} & 0.096 & 0.292 & 0.298 & 0.291 & 0.092 & 2.92\% \\ 
  RDDN & 0.195 &\textbf{ 0.041} & 0.072 &\textbf{ 0.041} & 0.084 & 0.267 & 0.285 & 0.269 & 0.042 & 2.29\%\\ \hline
  RSST & 0.401 & 0.295 & 0.314 &\textbf{ 0.289} & 0.302 & 0.497 & 0.495 & 0.486 & 0.290 & 0.58\%\\ 
  RSDT & 0.358 &\textbf{ 0.260} & 0.296 & 0.271 & 0.291 & 0.490 & 0.487 & 0.477 & 0.265 & 1.95\%\\ 
  RDST & 0.156 &\textbf{ 0.113} & 0.177 & 0.117 & 0.176 & 0.152 & 0.239 & 0.153 & 0.114 & 1.54\%\\ 
  RDDT & 0.134 &\textbf{ 0.095} & 0.152 & 0.099 & 0.171 & 0.135 & 0.236 & 0.128 & 0.096 & 0.77\% \\ \hline
  RSSV & 0.215 & 0.125 & 0.174 &\textbf{ 0.120} & 0.173 & 0.480 & 0.479 & 0.478 & 0.122 & 1.83\% \\ 
  RSDV & 0.217 & \textbf{ 0.095} & 0.172 & 0.102 & 0.215 & 0.475 & 0.474 & 0.474 & 0.097 & 2.32\%\\ 
  RDSV & 0.159 &\textbf{ 0.086} & 0.141 &\textit{ 0.087} & 0.148 & 0.270 & 0.304 & 0.272 & 0.086 & -0.39\% \\ 
  RDDV & 0.181 & 0.084 & 0.188 & \textbf{0.081} & 0.221 & 0.231 & 0.289 & 0.231 & 0.081 & 0.50\% \\ 
   \hline
\end{tabular}}
\caption{Misclassification rates of eight classifiers on $24$ scenarios, each an average  from $100$ simulations. Training size 500, test size 150.}
\label{sp:binaryV500}
\end{table}

To check classification performance in the varied score (V) setup when distributions are non-normal and non-tail-dependent, we include simulation results Table \ref{sp:binaryV500} here with a different choice of V: when $k = 1$, scores are distributed as standardized $\chi^2(1)$; when $k = 0$, it is standardized gamma distribution with both rate and scale parameters to as 1.

Also, in Table \ref{sp:binaryV500} we increased the training size to 500 for classification performance check. The major findings are consistent with Section \ref{simsection}.

Similar process is applied to the multiclass classification and the results are included in Table \ref{sp:multilevelV500}. We again increased the training size for each data scenario to $500$, and used a different set of score distributions for the varied distribution setup (V): when $k = 0$, scores distribution is standardized $\chi^2(1)$; when $k = 1$, it is standardized gamma distribution with both rate and scale parameters as 1; when $k = 2$, scores have log-normal distribution with parameters $\mu = 0$ and $\sigma^2 = 1$.

\begin{table}[h!]
\centering
\begin{tabular}{r|r|rrrr|rr|rr}
  \hline
 & BC & BCG & BCGPLS & BCt & BCtPLS & PLSDA & logistic & CV.mean & ratio.cv \\ 
  \hline
MSSN & 0.469 &\textbf{ 0.199} & 0.223 & 0.200 & 0.223 & 0.636 & 0.632 & 0.200 & 0.43\% \\ 
  MDSN & 0.247 &\textbf{ 0.066} & 0.072 &\textbf{ 0.066} & 0.073 & 0.451 & 0.390 & 0.068 & 3.32\% \\ 
  MSDN & 0.167 &\textbf{ 0.052} & 0.108 &\textit{ 0.053} & 0.160 & 0.630 & 0.621 & 0.051 & -3.05\% \\ 
  MDDN & 0.147 &\textbf{ 0.047} & 0.097 &\textbf{ 0.047} & 0.127 & 0.506 & 0.475 & 0.047 & 0.27\% \\ 
  \hline
  MSST & 0.505 & 0.304 & 0.340 &\textbf{ 0.296} & 0.315 & 0.629 & 0.637 & 0.296 & 0.08\% \\ 
  MDST & 0.278 & 0.128 & 0.143 &\textbf{ 0.126} & 0.148 & 0.421 & 0.344 & 0.122 & -3.79\% \\ 
  MSDT & 0.409 & 0.247 & 0.288 &\textbf{ 0.214} & 0.335 & 0.622 & 0.623 & 0.207 & -2.91\% \\ 
  MDDT & 0.296 & 0.164 & 0.202 &\textbf{ 0.130} & 0.263 & 0.468 & 0.382 & 0.131 & 0.40\% \\ 
  \hline
  MSSV & 0.303 &\textbf{ 0.187} & 0.275 & 0.197 & 0.285 & 0.625 & 0.618 & 0.185 & -0.67\% \\ 
  MDSV & 0.196 &\textbf{ 0.097} & 0.248 &\textbf{ 0.097} & 0.264 & 0.465 & 0.391 & 0.100 & 3.20\% \\ 
  MSDV & 0.252 & 0.149 & 0.205 &\textbf{ 0.140} & 0.295 & 0.622 & 0.615 & 0.142 & 1.28\% \\ 
  MDDV & 0.206 & 0.115 & 0.162 &\textbf{ 0.109} & 0.238 & 0.523 & 0.462 & 0.108 & -0.79\% \\ 
   \hline
\end{tabular}
\caption{Misclassification rates averaged over $100$ simulations of the $7$ classifiers on $12$ multinomial data scenarios. Training sizes are again increased to $500$.}
\label{sp:multilevelV500}
\end{table}

\subsection{Correlation of Scores in RSDN}\label{sec:corrRSDN}
\begin{table}[h!]
\centering
\scalebox{0.8}{\begin{tabular}{r|rrrrrrrrrr}
  \hline
 & 1 & 2 & 3 & 4 & 5 & 6 & 7 & 8 & 9 & 10 \\ 
  \hline
1 & 1.000 &&&&&&&&& \\ 
  2 & -0.283 & 1.000 &&&&&&&& \\ 
  3 & 0.102 & -0.548 & 1.000 &&&&&&& \\ 
  4 & 0.292 & 0.384 & -0.253 & 1.000 & &&&&& \\ 
  5 & -0.119 & -0.346 & 0.210 & -0.668 & 1.000 &&&&& \\ 
  6 & -0.362 & -0.069 & -0.023 & -0.431 & 0.362 & 1.000 &&&& \\ 
  7 & 0.013 & -0.014 & 0.189 & 0.201 & -0.194 & -0.225 & 1.000 & && \\ 
  8 & 0.245 & 0.134 & -0.113 & 0.478 & -0.311 & -0.360 & 0.186 & 1.000 && \\ 
  9 & -0.159 & -0.042 & 0.180 & -0.085 & 0.045 & 0.204 & -0.070 & -0.039 & 1.000 &\\ 
  10 & -0.066 & 0.028 & 0.080 & 0.131 & -0.178 & -0.219 & 0.439 & 0.079 & 0.006 & 1.000 \\ 
   \hline
\end{tabular}}
\caption{Pearson correlations of scores on first $10$ joint basis at group $k=1$ in Scenario RSDN. Correlations are estimated from $500$ samples in total of both groups.}
\label{tab:corrRSDN1}
\end{table}

\begin{table}[h!]
\centering
\scalebox{0.8}{\begin{tabular}{r|rrrrrrrrrr}
  \hline
 & 1 & 2 & 3 & 4 & 5 & 6 & 7 & 8 & 9 & 10 \\ 
  \hline
1 &  & & & & & & & & &  \\ 
  2 & \textbf{0.000} & &&&&&&&& \\ 
  3 & 0.113 &\textbf{ 0.000 }&  &&&&&&& \\ 
  4 &\textbf{ 0.000} &\textbf{ 0.000} &\textbf{ 0.000} &  &&&&&& \\ 
  5 & 0.064 &\textbf{ 0.000} &\textbf{ 0.001} &\textbf{ 0.000} &  &&&&&  \\ 
  6 &\textbf{ 0.000} & 0.283 & 0.722 &\textbf{ 0.000} &\textbf{ 0.000} &  &&&& \\ 
  7 & 0.841 & 0.829 &\textbf{ 0.003} &\textbf{ 0.002} &\textbf{ 0.002} &\textbf{ 0.000} & &&& \\ 
  8 &\textbf{ 0.000} &\textbf{ 0.036} & 0.077 &\textbf{ 0.000} &\textbf{ 0.000} & \textbf{0.000} &\textbf{ 0.003} &  && \\ 
  9 &\textbf{ 0.013} & 0.518 &\textbf{ 0.005} & 0.188 & 0.480 &\textbf{ 0.001} & 0.275 & 0.545 &  & \\ 
  10 & 0.306 & 0.662 & 0.213 &\textbf{ 0.040} &\textbf{ 0.005} &\textbf{ 0.001} &\textbf{ 0.000} & 0.216 & 0.921 &  \\ 
   \hline
\end{tabular}}
\caption{P-values from significance test of correlations for scores in Group $k=1$ in Scenario RSDN. $P < 0.05$ is labeled green.}
\label{tab:pvalueRSDN1}
\end{table}

\begin{table}[h!]
\centering
\scalebox{0.8}{\begin{tabular}{r|rrrrrrrrrr}
  \hline
 & 1 & 2 & 3 & 4 & 5 & 6 & 7 & 8 & 9 & 10 \\ 
  \hline
1 & 1.000 &&&&&&&&& \\ 
  2 & 0.015 & 1.000 &&&&&&&& \\ 
  3 & -0.007 & 0.054 & 1.000 &&&&&&& \\ 
  4 & -0.082 & -0.158 & 0.135 & 1.000 &&&&&& \\ 
  5 & 0.011 & 0.046 & -0.036 & 0.460 & 1.000 &&&&& \\ 
  6 & 0.029 & 0.009 & 0.005 & 0.269 & -0.072 & 1.000 &&&& \\ 
  7 & -0.001 & 0.001 & -0.025 & -0.105 & 0.033 & 0.035 & 1.000 &&& \\ 
  8 & -0.017 & -0.012 & 0.017 & -0.254 & 0.053 & 0.054 & -0.023 & 1.000 && \\ 
  9 & 0.008 & 0.003 & -0.016 & 0.031 & -0.005 & -0.022 & 0.007 & 0.003 & 1.000 & \\ 
  10 & 0.005 & -0.005 & -0.014 & -0.072 & 0.031 & 0.037 & -0.061 & -0.009 & -0.000 & 1.000 \\ 
   \hline
\end{tabular}}
\caption{Pearson correlations of scores on first $10$ joint basis at group $k=0$ in Scenario RSDN. Correlations are estimated from $500$ samples in total of both groups.}
\label{tab:corrRSDN0}
\end{table}

\begin{table}[h!]
\centering
\scalebox{0.8}{\begin{tabular}{r|rrrrrrrrrr}
  \hline
 & 1 & 2 & 3 & 4 & 5 & 6 & 7 & 8 & 9 & 10 \\ 
  \hline
1 &  &&&&&&&&& \\ 
  2 & 0.805 &  &&&&&&&& \\ 
  3 & 0.917 & 0.392 &  &&&&&&& \\ 
  4 & 0.193 & \textbf{0.011} &\textbf{ 0.031} &  &&&&&& \\ 
  5 & 0.866 & 0.467 & 0.572 &\textbf{ 0.000} &  &&&&& \\ 
  6 & 0.642 & 0.884 & 0.940 &\textbf{ 0.000} & 0.249 &  &&&& \\ 
  7 & 0.991 & 0.990 & 0.688 & 0.093 & 0.603 & 0.579 &  &&& \\ 
  8 & 0.785 & 0.846 & 0.789 &\textbf{ 0.000} & 0.401 & 0.386 & 0.710 &  && \\ 
  9 & 0.903 & 0.960 & 0.797 & 0.616 & 0.931 & 0.722 & 0.918 & 0.957 &  & \\ 
  10 & 0.935 & 0.938 & 0.828 & 0.253 & 0.616 & 0.558 & 0.333 & 0.888 & 0.996 &  \\ 
   \hline
\end{tabular}}
\caption{P-values from significance test of correlations for scores in Group $k=0$ in Scenario RSDN. $P < 0.05$ is labeled green.}
\label{tab:pvalueRSDN0}
\end{table}

\newpage
\begin{figure}[h!]
\centering
      \includegraphics[scale=.6]{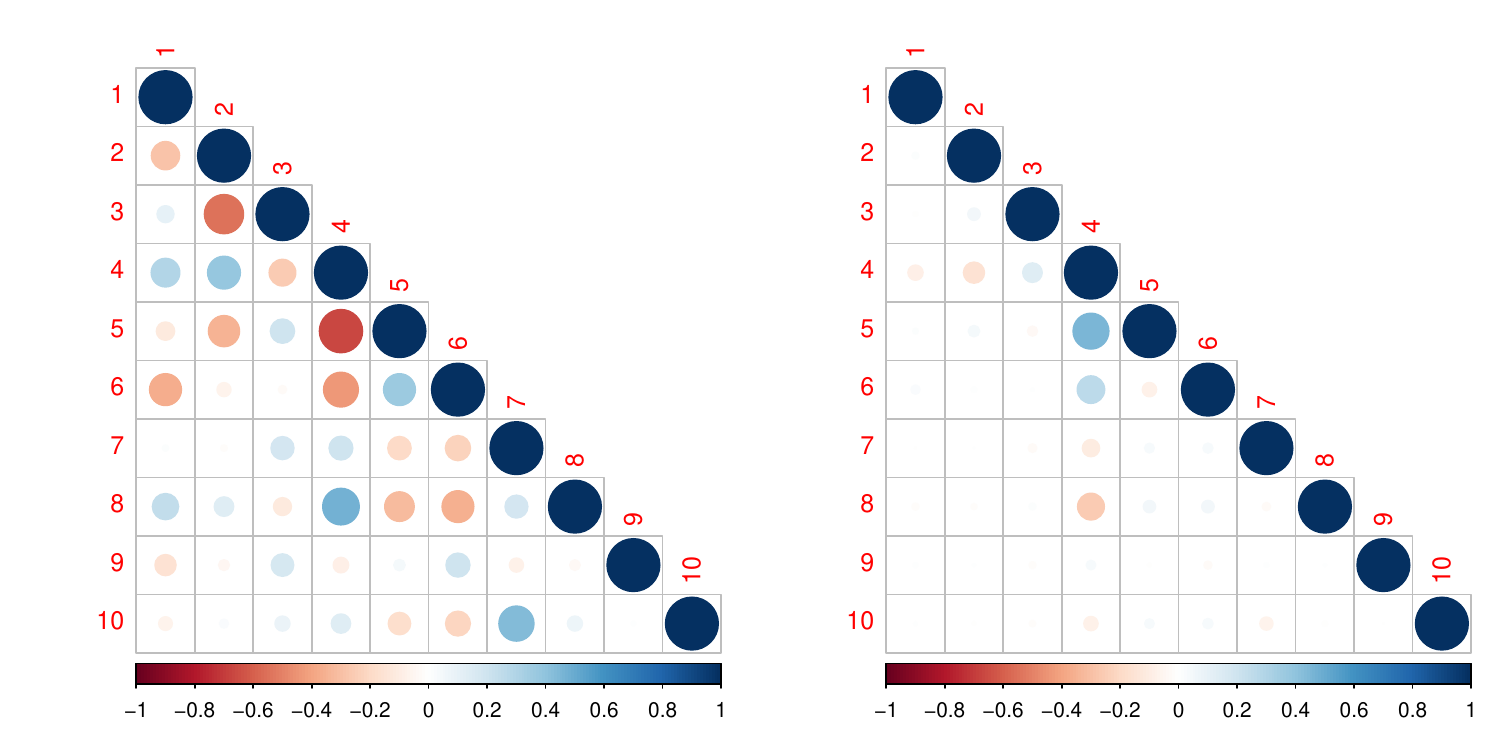}
  \caption{Comparison of correlation plots of first 10 scores at both group of RSDN. Left: $k=1$; Right: $k=0$.}
 \label{fig:corrRSDN}
\end{figure}

\clearpage
\subsection{Correlation of scores in RSDT}\label{sec:corrRSDT}
\begin{table}[h!]
\centering
\scalebox{0.8}{\begin{tabular}{r|rrrrrrrrrr}
  \hline
 & 1 & 2 & 3 & 4 & 5 & 6 & 7 & 8 & 9 & 10 \\ 
  \hline
1 & 1.000 & &&&&&&&& \\ 
  2 & -0.361 & 1.000 & &&&&&&& \\ 
  3 & 0.110 & 0.258 & 1.000 & &&&&&& \\ 
  4 & -0.278 & 0.300 & 0.015 & 1.000 & &&&&& \\ 
  5 & 0.144 & 0.069 & 0.759 & -0.295 & 1.000 & &&&& \\ 
  6 & 0.015 & -0.061 & 0.155 & -0.257 & 0.262 & 1.000 & &&& \\ 
  7 & -0.189 & -0.077 & -0.128 & 0.117 & -0.138 & 0.276 & 1.000 & && \\ 
  8 & 0.094 & -0.079 & 0.307 & -0.099 & 0.367 & 0.036 & -0.158 & 1.000 & & \\ 
  9 & 0.156 & -0.058 & 0.291 & -0.234 & 0.297 & -0.114 & -0.176 & -0.074 & 1.000 &  \\ 
  10 & -0.075 & -0.077 & -0.142 & -0.046 & 0.002 & 0.103 & -0.063 & 0.187 & -0.399 & 1.000 \\ 
   \hline
\end{tabular}}
\caption{Pearson correlations of scores on first $10$ joint basis at group $k=1$ in Scenario RSDT. Correlations are estimated from $500$ samples in total of both groups.}
\label{tab:corrRSDT1}
\end{table}

\begin{table}[h!]
\centering
\scalebox{0.8}{\begin{tabular}{r|rrrrrrrrrr}
  \hline
 & 1 & 2 & 3 & 4 & 5 & 6 & 7 & 8 & 9 & 10 \\ 
  \hline
1 &  & &&&&&&&& \\ 
  2 &\textbf{ 0.000} &  & &&&&&&& \\ 
  3 & 0.102 &\textbf{ 0.000} &  & &&&&&& \\ 
  4 &\textbf{ 0.000} &\textbf{ 0.000} & 0.820 &  & &&&&& \\ 
  5 &\textbf{ 0.032} & 0.302 &\textbf{ 0.000} &\textbf{ 0.000} &  & &&&& \\ 
  6 & 0.820 & 0.360 &\textbf{ 0.020} &\textbf{ 0.000} &\textbf{ 0.000} &  & &&& \\ 
  7 &\textbf{ 0.005} & 0.252 & 0.056 & 0.079 &\textbf{ 0.039} &\textbf{ 0.000} &  & && \\ 
  8 & 0.160 & 0.236 &\textbf{ 0.000} & 0.140 &\textbf{ 0.000} & 0.591 &\textbf{ 0.018} &  && \\ 
  9 &\textbf{ 0.020} & 0.387 &\textbf{ 0.000} &\textbf{ 0.000} &\textbf{ 0.000} & 0.088 &\textbf{ 0.008} & 0.271 &  & \\ 
  10 & 0.263 & 0.253 &\textbf{ 0.034} & 0.495 & 0.976 & 0.124 & 0.345 &\textbf{ 0.005} &\textbf{ 0.000} &  \\ 
   \hline
\end{tabular}}
\caption{P-values from significance test of correlations for scores in Group $k=1$ in Scenario RSDT. $P < 0.05$ is labeled green.}
\label{tab:pvalueRSDT1}
\end{table}

\begin{table}[h!]
\centering
\scalebox{0.8}{\begin{tabular}{r|rrrrrrrrrr}
  \hline
 & 1 & 2 & 3 & 4 & 5 & 6 & 7 & 8 & 9 & 10 \\ 
  \hline
1 & 1.000 & &&&&&&&& \\ 
  2 & 0.022 & 1.000 & &&&&&&& \\ 
  3 & -0.017 & -0.065 & 1.000 & &&&&&& \\ 
  4 & 0.033 & -0.058 & -0.007 & 1.000 & &&&&& \\ 
  5 & -0.026 & -0.019 & -0.562 & 0.170 & 1.000 & &&&& \\ 
  6 & -0.001 & 0.009 & -0.056 & 0.072 & -0.113 & 1.000 & &&&\\ 
  7 & 0.018 & 0.012 & 0.050 & -0.036 & 0.064 & -0.063 & 1.000 & && \\ 
  8 & -0.008 & 0.010 & -0.103 & 0.026 & -0.146 & -0.007 & 0.033 & 1.000 & &\\ 
  9 & -0.012 & 0.010 & -0.091 & 0.057 & -0.111 & 0.021 & 0.035 & 0.013 & 1.000 &  \\ 
  10 & 0.006 & 0.012 & 0.039 & 0.010 & -0.002 & -0.016 & 0.011 & -0.027 & 0.053 & 1.000 \\ 
   \hline
\end{tabular}}
\caption{Pearson correlations of scores on first $10$ joint basis at group $k=0$ in Scenario RSDT. Correlations are estimated from $500$ samples in total of both groups.}
\label{tab:corrRSDT0}
\end{table}

\begin{table}[h!]
\centering
\scalebox{0.8}{\begin{tabular}{r|rrrrrrrrrr}
  \hline
 & 1 & 2 & 3 & 4 & 5 & 6 & 7 & 8 & 9 & 10 \\ 
  \hline
1 &  & &&&&&&&&\\ 
  2 & 0.718 &  & &&&&&&& \\ 
  3 & 0.778 & 0.282 &  & &&&&&& \\ 
  4 & 0.580 & 0.336 & 0.903 &  & &&&&& \\ 
  5 & 0.665 & 0.756 &\textbf{  0.000} &\textbf{  0.005} &  & &&&& \\ 
  6 & 0.982 & 0.881 & 0.351 & 0.230 & 0.060 &  & &&& \\ 
  7 & 0.762 & 0.843 & 0.408 & 0.556 & 0.287 & 0.299 &  & && \\ 
  8 & 0.895 & 0.871 & 0.086 & 0.669 &\textbf{  0.015} & 0.907 & 0.581 &  & & \\ 
  9 & 0.846 & 0.875 & 0.132 & 0.348 & 0.064 & 0.731 & 0.567 & 0.830 &  &  \\ 
  10 & 0.926 & 0.845 & 0.518 & 0.873 & 0.970 & 0.785 & 0.856 & 0.659 & 0.383 &  \\ 
   \hline
\end{tabular}}
\caption{P-values from significance test of correlations for scores in Group $k=0$ in Scenario RSDT. $P < 0.05$ is labeled green.}
\label{tab:pvalueRSDT0}
\end{table}

\begin{figure}[h!]
\centering
      \includegraphics[scale=.6]{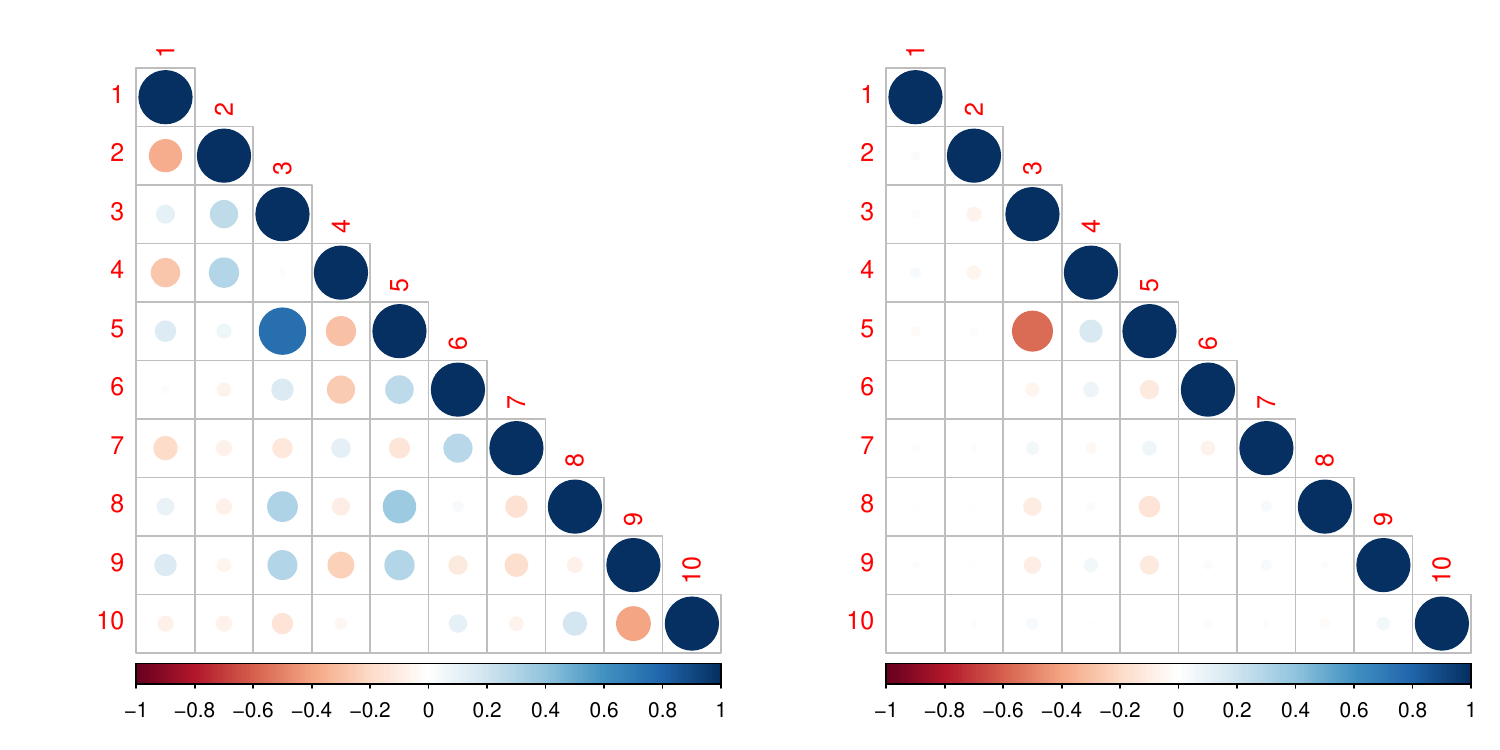}
  \caption{Comparison of correlation plots of first 10 scores at both group of RSDT. Left: $k=1$; Right: $k=0$.}
   \label{fig:corrRSDT}
\end{figure}

\clearpage

\setcounter{equation}{0}
\section{Additional Results for Two Data Examples} \label{sup:dataexample}
\subsection{Fractional Anisotropy Example}

\begin{figure}[h!]
\centering
      \includegraphics[scale=.5]{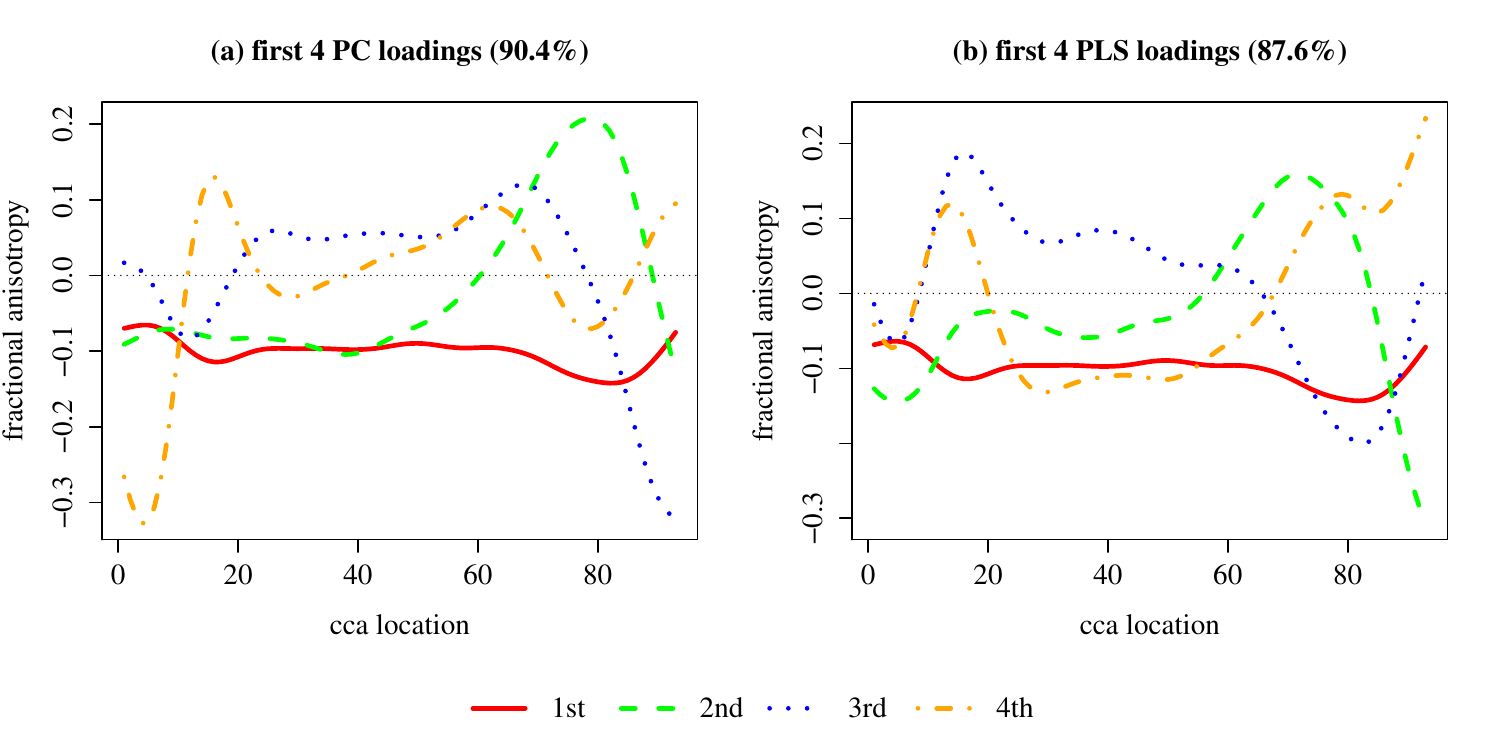}
  \caption{First four loading functions of PC (left) and PLS (right) of the smoothed FA profiles, with percentage of total variation reported in the titles. Both loadings are scaled to unit length for comparison. The first loading functions are red and are roughly horizontal for each method.}
  \label{ccadim}
\end{figure}

\begin{figure}[h!]
\centering
      \includegraphics[scale=.5]{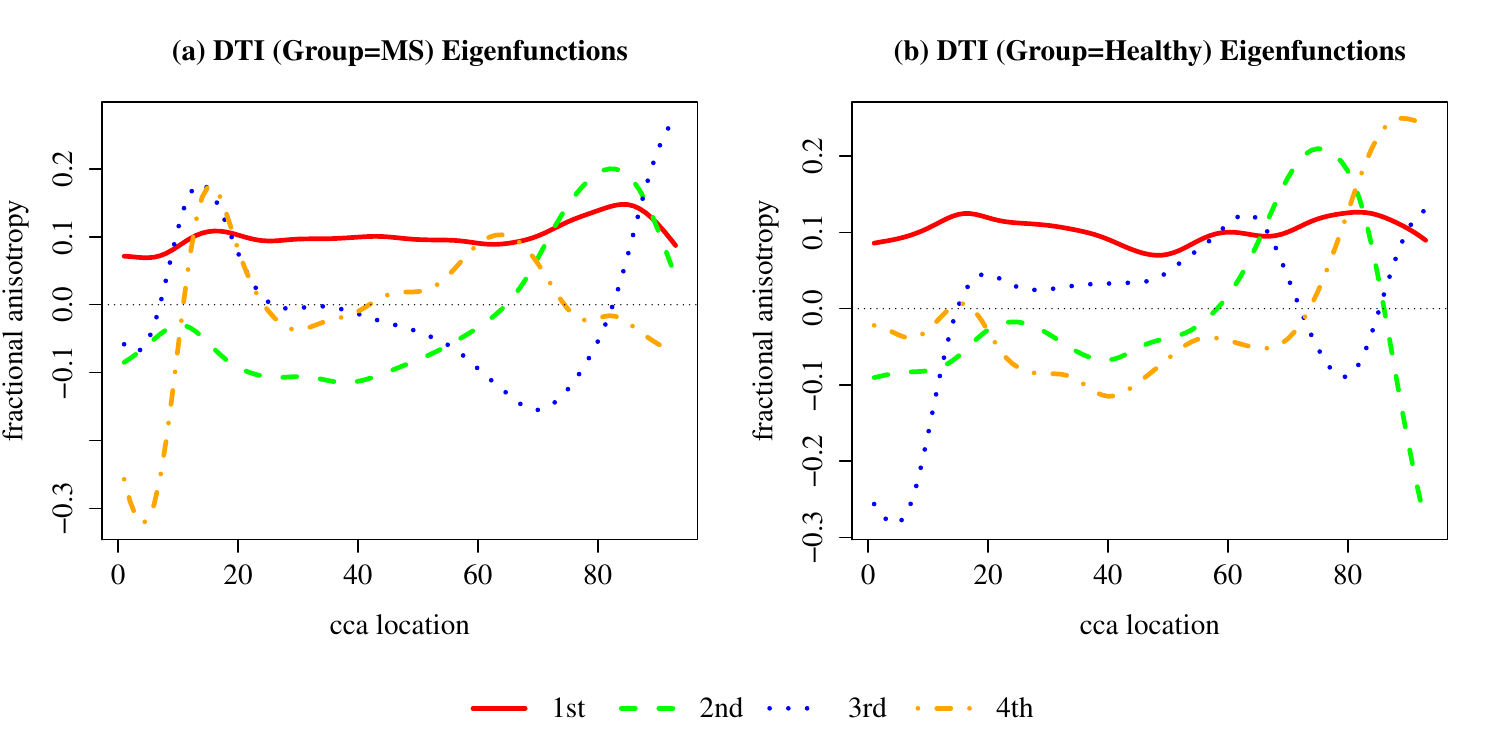}
  \caption{First four group eigenfunctions of smoothed FA profiles in group MS or Healthy.}
  \label{ccagroup}
\end{figure}



\begin{figure}[h!]
\centering
      \includegraphics[scale=.26]{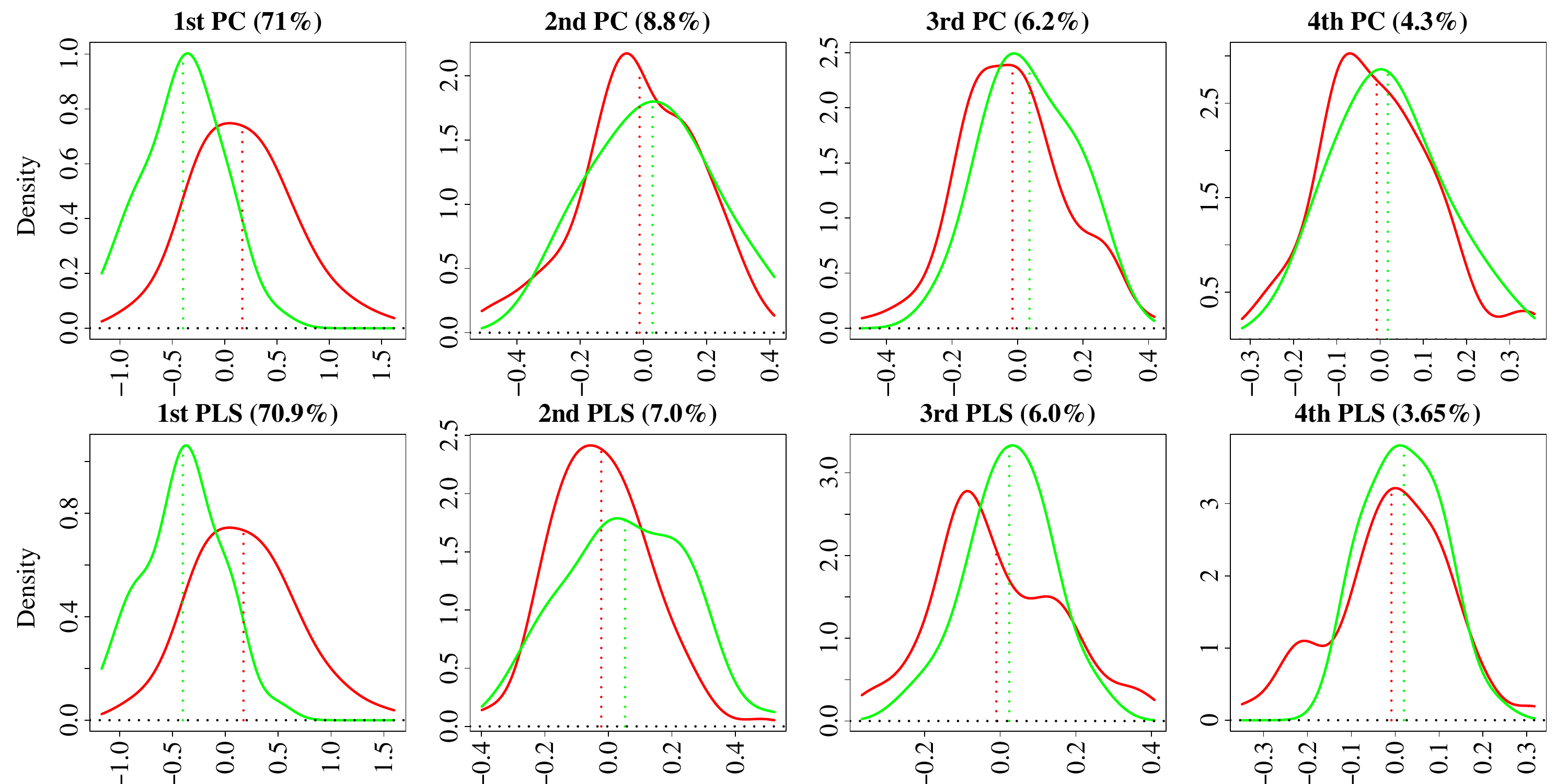} 
  \caption{Estimated densities of scores on first four PC and PLS components in MS (in red) and healthy groups (in green). The proportion of total variation each component explains is included in plot titles. Locations of group score average are labeled with dashed lines.}
  \label{ccadensity}
\end{figure}

In Fig.\ \ref{ccadensity}, we compare the projected score distributions on PC and PLS, with densities estimated by KDE. 
In distinguishing between cases and controls, the first and third PC components are more important than the second one, which captures mostly within-group variation.  Overall,
PLS does not improve over PC, consistent with the results in Table~\ref{tab:ccatable}.

Score correlation tests on first four principal components reveal that, though no significant correlation is found  in MS cases, the 2nd and 3rd components of the control group are positively correlated with Spearman's $\rho$ at $0.525$ and an adjusted $p$-value $2\times 10^{-2}$. Scores on the first four PLS components do not show significance correlations. Therefore, while PC and PLS show almost equal ability in capturing variation with first several components in DTI data, PC exhibits correlation between components in one of the two groups, which may explain the superior performance of PC and of the copula-based classifiers, BCG and BC-t. 

Figure~\ref{ccagroup} show the first four group-specific eigenfunctions.  There are some differences, especially after the first eigenfunctions, which may also contribute to the superior performance of the copula-based classifiers.

\subsection{Additional results of the PM/velocity example} 

\begin{figure}[h!]
\centering
      \includegraphics[scale=.6]{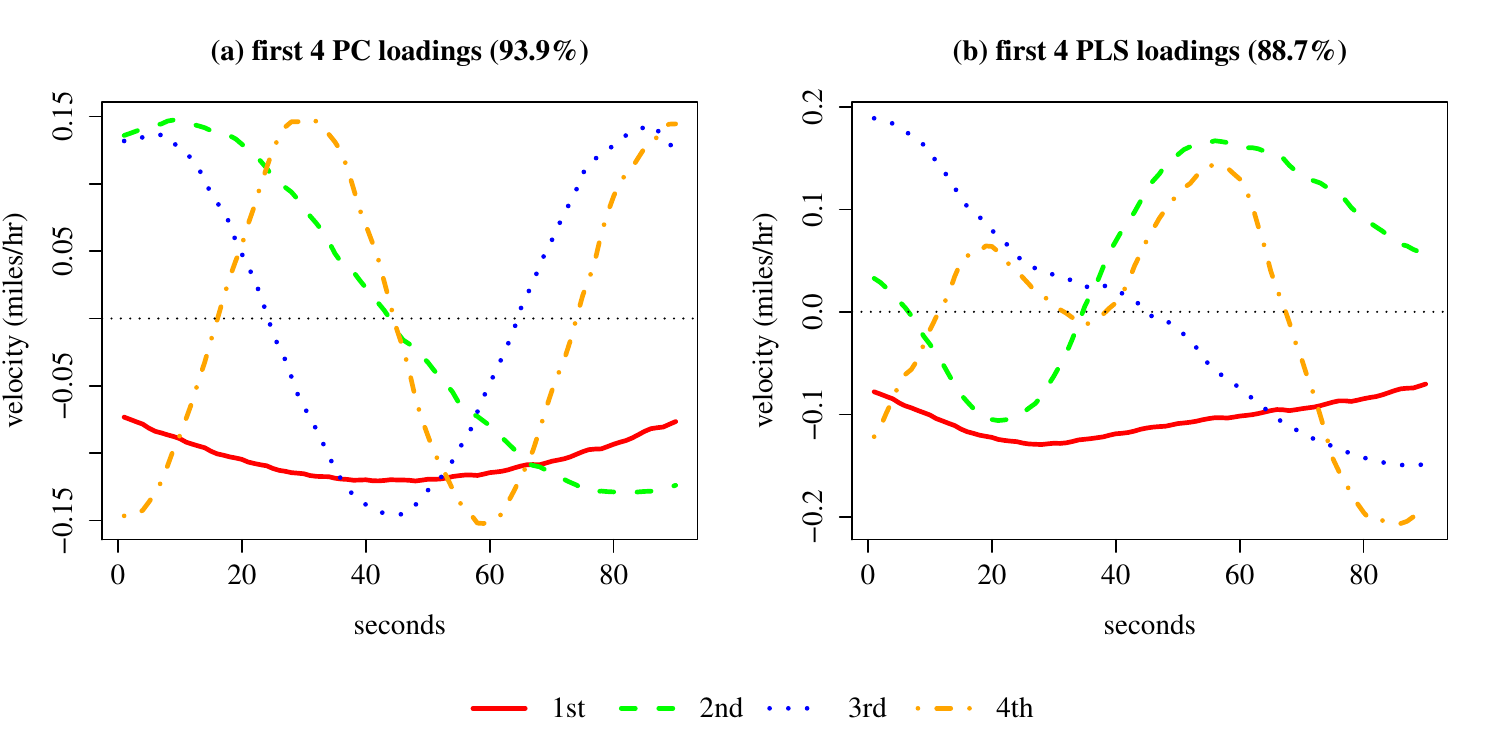}
  \caption{First 4 loading functions on PC (left) and PLS (right) for raw truck velocities, with percentage of total variation reported by first four components in the titles. Both loadings are scaled to unit length.}
  \label{truckdim}
\end{figure}

The first four PC and PLS loading functions are plotted in Fig.\ \ref{truckdim}, with $93.9\%$ of total variation explained by the four PCs, and  $88.7\%$ by PLS components. 
The fractions SSB/SST (between to total sums of squares) of the first four PCs respectively are $2.12\%, 0.37\%, 0.17\%, 6.27\%$, while for PLS they are noticeably larger, $5\%, 13.3\%,  4.71\%, 4.13\%$.  We compare the score distributions  in Fig.\ \ref{truckdensity}, with group means indicated by dashed lines.
The second PLS component with a SSB/SST ratio $13.3\%$ appears  strongest in distinguishing between PM emission groups.

\begin{figure}[h!]
\centering
      \includegraphics[scale=.275]{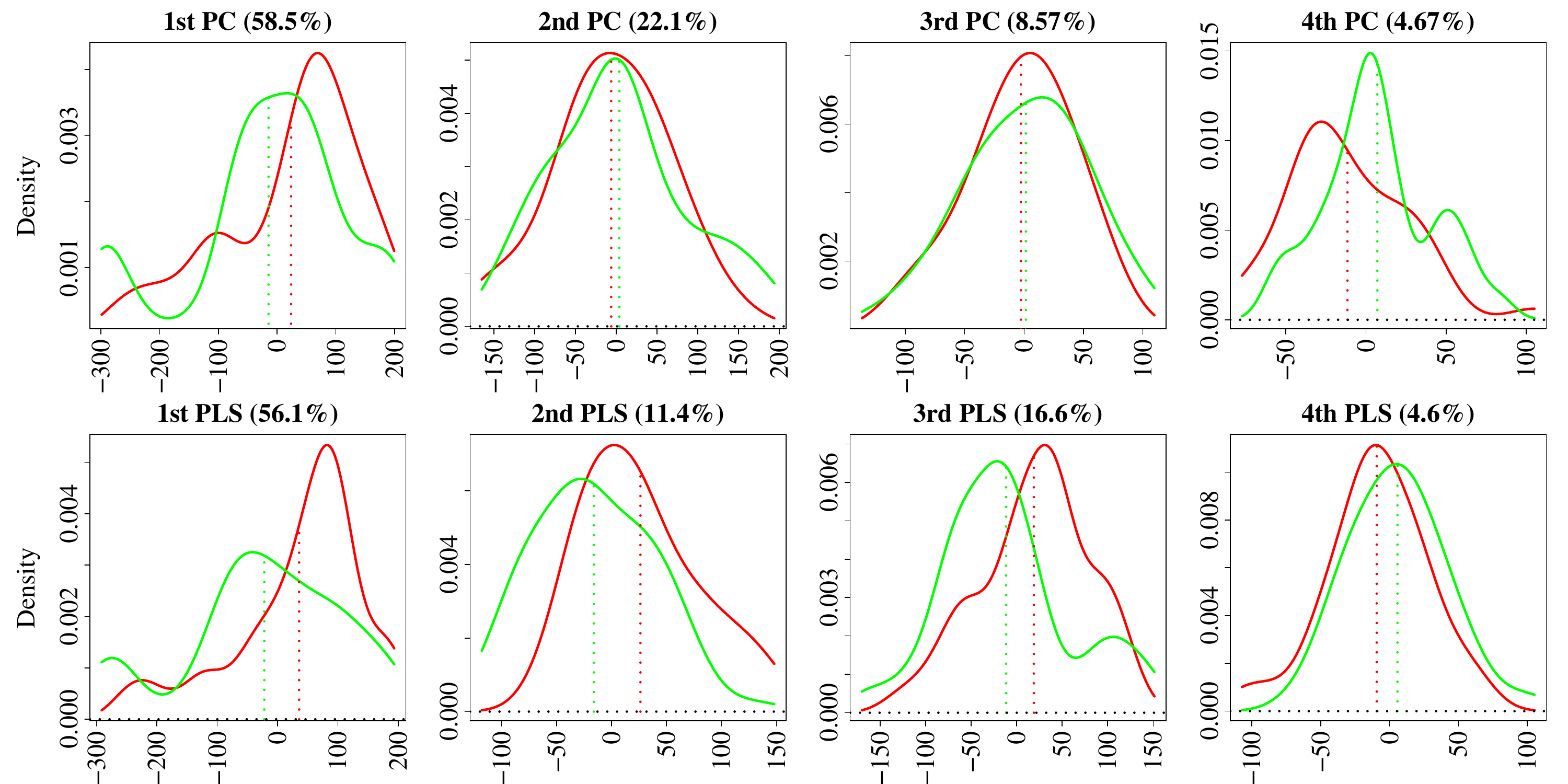}
  \caption{Score densities of first four PC and PLS components in high PM (in \textcolor{red}{red}) and low PM groups (in \textcolor{green}{green}). The proportion of total variation each component explains is included in headlines. The   SSB/SST ratios are $2.12\%, 0.37\%, 0.17\%, 6.27\%$ for PC, and $5\%, 13.3\%,  4.71\%, 4.13\%$ for PLS. The densities are estimated by KDE with direct plug-in bandwidths.  Group means are lindicated by dashed lines.}
  \label{truckdensity}
\end{figure}

PLS components, especially the second one, are able to capture distinctions between the movement patterns causing high and low PM emission. The projected velocity scores of the high PM group on the second PLS component have a positive group mean and a smaller standard deviation, compared to the negative mean and the larger standard deviation of the low PM group. The second PLS loading function, as shown in Fig.\ \ref{truckdim}, starts near~0, and decreases for the first 20 seconds, then is positive for roughly the last 55 seconds. 
(The loading functions are modeling deviations from average values, so a negative value indicates a below-average velocity.)
This pattern is consistent with our earlier finding that while the low PM group has greater variation, the high PM cases have a constant pattern of decelerating over the first $20$ seconds with much lower standard deviation, followed by acceleration with increasing variation. 

\begin{figure}[h!]
\centering
      \includegraphics[scale=.6]{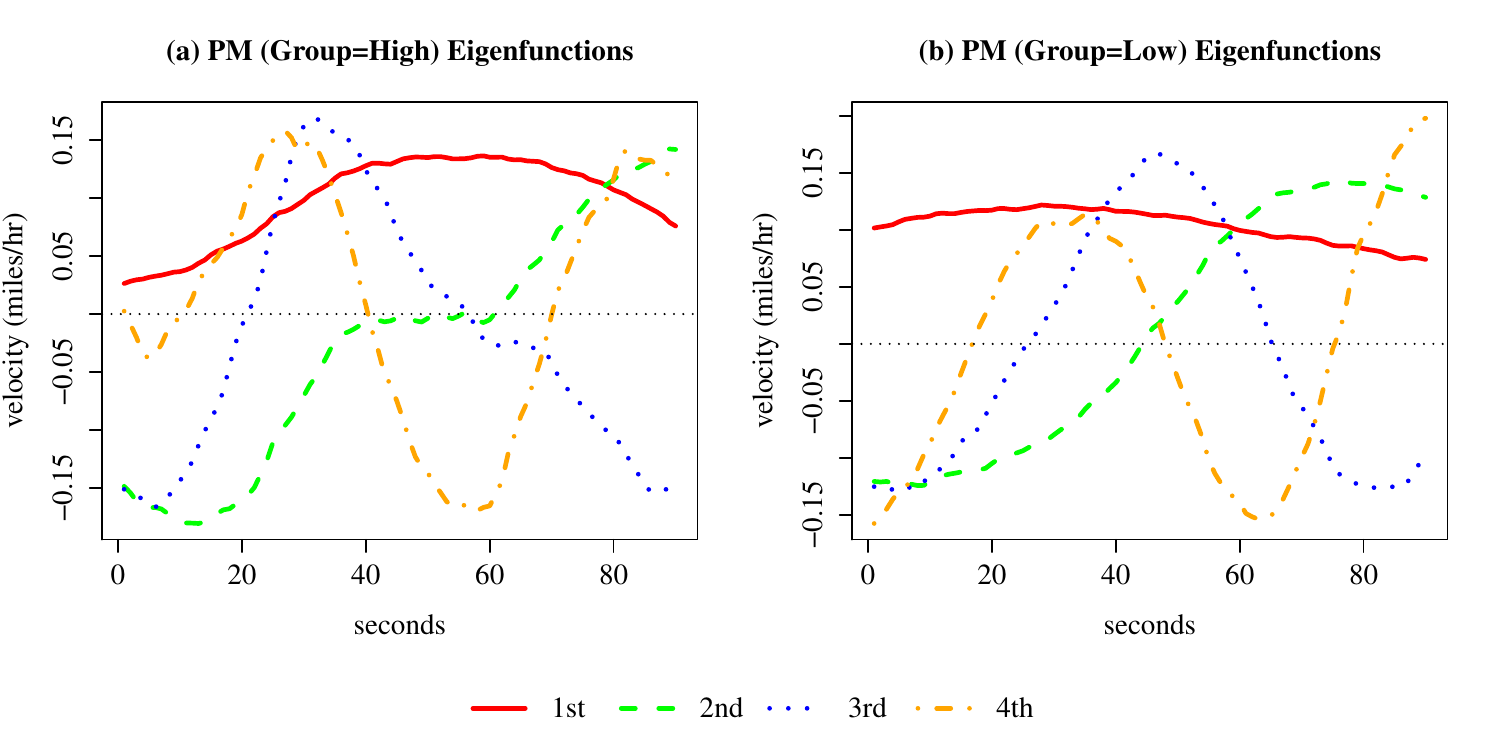}
  \caption{First 4 eigenfunctions of raw truck velocity data in group High or Low.}
  \label{truckgroup}
\end{figure}

\subsection{Group Mean Difference Comparison}
In Fig.\ \ref{sup:realmean}, we compare the projected group mean difference of the two data examples, both on the first $20$ joint eigenfunctions. Apparently, in the first example of DTI data, principal components are able to detect the location difference effectively at about first $5$ basis. On the other hand, in Panel (b), the particulate emission data present a more significant group mean difference, which takes more than $12$ eigenfunctions to fully capture. These two situations validate their different choices of PC and PLS based classifiers.

\begin{figure}[h!]
\centering
      \includegraphics[scale=.43]{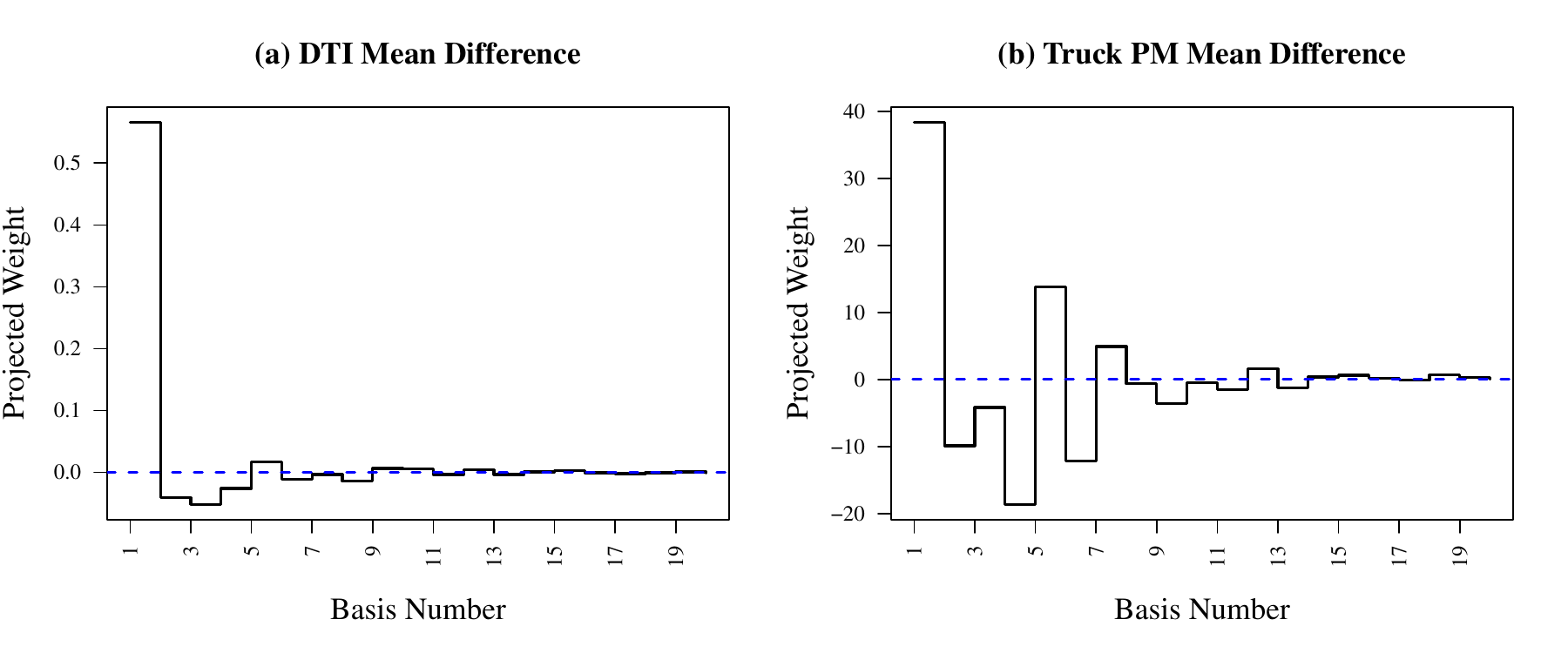}
  \caption{Comparison of projected group mean difference of DTI and PM data, both on the first $20$ joint eigenfunctions. Level $0$ is labeled with a dashed blue line in each plot.}
\label{sup:realmean}
\end{figure}

\clearpage

\section{Proof of Theorem 1}\label{sec:suppTheorem1Proof}
\setcounter{equation}{0}

\subsection{Estimation error of KDE $\hat{f}_{jk}$ on unequal group eigenfunctions} \label{mkde}
Let the class of functions $\mathcal{S}(c)=\{x \in \mathcal{L}^2(\mathcal{T}): \|x\| \le c\}$, $\forall c>0$. We prove Proposition 1 in Section 5.1 of the paper:
\begin{proof}
First let $\hat{g}_{jk} (\hat{x}_j)$ be kernel density estimation (KDE) of standardized scores projected on $\hat{\phi}_j$ at group $k$, and $\hat{g}_j(\hat{x}_j)$ for standardized joint scores, where $\hat{\phi}_j$ and $\hat{\lambda}_j$ are the estimated $j$-th joint eigenfunction and eigenvalue pair from sample eigen-decomposition as illustrated in \cite{DH2011},

\begin{equation} \label{gjk}
\hat{g}_{jk} \left(\hat{x}_j\right)=\dfrac{1}{n_kh}\sum_{i=1}^{n_k} K\left(\dfrac{\langle X_{ik}-x, \hat{\phi}_j\rangle}{\hat{\sigma}_{jk}h}\right), \hat{g}_{j} \left(\hat{x}_j\right)= \dfrac{1}{nh}\sum_{i=1}^{n} K\left(\dfrac{\langle X_{i}-x, \hat{\phi}_j\rangle}{\sqrt{\hat{\lambda}_j}h}\right),
\end{equation}
with $\hat{\sigma}_{jk}$ as sample standard deviation of $\sigma_{jk}=\sqrt{Var \langle X_{ik}, \phi_j\rangle}$, and $h$ is the unit bandwidth for standardized scores. Thus, the estimated marginal density $\hat{f}_{jk}(\hat{x}_j)$ and $\hat{f}_{j}(\hat{x}_j)$ can be correspondingly expressed as
\begin{equation} \label{fjk}
\hat{f}_{jk} \left(\hat{x}_j\right)=\dfrac{1}{\hat{\sigma}_{jk}}\dfrac{1}{n_kh}\sum_{i=1}^{n_k} K\left(\dfrac{\langle X_{ik}-x, \hat{\phi}_j\rangle}{\hat{\sigma}_{jk}h}\right)=\dfrac{1}{\hat{\sigma}_{jk}} \hat{g}_{jk} \left(\hat{x}_j\right), 
\end{equation}
and
\begin{equation} \label{fj}
\hat{f}_{j} \left(\hat{x}_j\right)=\dfrac{1}{\sqrt{\hat{\lambda}_j}} \dfrac{1}{nh}\sum_{i=1}^{n} K\left(\dfrac{\langle X_{i}-x, \hat{\phi}_j\rangle}{\sqrt{\hat{\lambda}_j}h}\right)=\dfrac{1}{\sqrt{\hat{\lambda}_j}}\hat{g}_j\left(\hat{x}_j\right).
\end{equation}
In addition, when $\phi_j$, $\lambda_j$ and $\delta_{jk}$ are known, we use $\bar{f}_{jk}$ and $\bar{f}_{j}$ as below,
\begin{equation} \label{barjk}
\bar{f}_{jk} \left(x_j\right)=\dfrac{1}{\sigma_{jk}}\dfrac{1}{n_kh}\sum_{i=1}^{n_k} K\left(\dfrac{\langle X_{ik}-x, \phi_j\rangle}{\sigma_{jk}h}\right)=\dfrac{1}{\sigma_{jk}} \bar{g}_{jk} \left(x_j\right), 
\end{equation}
and
\begin{equation} \label{barj}
\bar{f}_{j} \left(x_j\right)=\dfrac{1}{\sqrt{\lambda}_j} \dfrac{1}{nh}\sum_{i=1}^{n} K\left(\dfrac{\langle X_{i}-x, \phi_j\rangle}{\sqrt{\lambda_j}h}\right)=\dfrac{1}{\sqrt{\lambda}_j}\bar{g}_j\left(x_j\right).
\end{equation}
 
With Taylor expansion, 
\begin{align}
\hat{\pi}_1 \hat{g}_{j1}\left(\hat{x}_j\right)+\hat{\pi}_0 \hat{g}_{j0}\left(\hat{x}_j\right) &= \dfrac{1}{nh} \sum_{i=1}^{n_1} K\left(\dfrac{\langle X_{i1}-x, \hat{\phi}_j \rangle}{\sqrt{\hat{\lambda}_j}h}\right) \label{I}\\
&+\dfrac{1}{nh} \sum_{i=1}^{n_1}\left(\dfrac{1}{\hat{\sigma}_{j1}}-\dfrac{1}{\sqrt{\hat{\lambda}_j}}\right)\dfrac{1}{h} \langle X_{i1}-x, \hat{\phi}_j \rangle K'\left(\gamma_{ij1}\right) \label{II}\\
&+ \dfrac{1}{nh} \sum_{i=1}^{n_0} K\left(\dfrac{\langle X_{i0}-x, \hat{\phi}_j \rangle}{\sqrt{\hat{\lambda}_j}h}\right) \label{III}\\
&+ \dfrac{1}{nh} \sum_{i=1}^{n_0} \left(\dfrac{1}{\hat{\sigma}_{j0}}-\dfrac{1}{\sqrt{\hat{\lambda}_j}}\right)\dfrac{1}{h} \langle X_{i0}-x, \hat{\phi}_j \rangle K'\left(\gamma_{ij0}\right), \label{IV}
\end{align}
where $\gamma_{ijk}=c_{ijk}\cdot \dfrac{\langle X_{ik}-x, \hat{\phi}_j \rangle}{h}$, with $c_{ijk}$ between $\dfrac{1}{\sqrt{\hat{\lambda}_j}}$ and $\dfrac{1}{\hat{\sigma}_{jk}}$. Since Eq.(\ref{I}) + Eq.(\ref{III}) is $\hat{g}_j\left(\hat{x}_j\right)$,  $\hat{\pi}_1 \hat{g}_{j1}\left(\hat{x}_j\right)+\hat{\pi}_0 \hat{g}_{j0}\left(\hat{x}_j\right) - \hat{g}_j\left(\hat{x}_j\right)$ is sum of the two parts Eq.(\ref{II}) and Eq.(\ref{IV}). 

Then we discuss specifically the case when the kernel function $K$ here is standard Gaussian. We denote the partial term $\dfrac{1}{h} \langle X_{ik}-x, \hat{\phi}_j \rangle K'\left(\gamma_{ijk}\right)$ in Eq.(\ref{II}) and Eq.(\ref{IV}) as $A_{ijk}$. Therefore,
\begin{align} \label{partA}
A_{ijk}&=\dfrac{1}{h} \langle X_{ik}-x, \hat{\phi}_j \rangle K'\left(\gamma_{ijk}\right) \nonumber\\
&=-\dfrac{c_{ijk}}{h^2}\langle X_{ik}-x, \hat{\phi}_j\rangle^2 \exp\left(-\dfrac{1}{2}\dfrac{c_{ijk}^2}{h^2} \langle X_{ik}-x, \hat{\phi}_j\rangle^2\right)\cdot \dfrac{1}{\sqrt{2\pi}}
\end{align}

To show $A_{ijk}=op \left(h^2\right)$, we let
\begin{equation} \label{Ak2}
\left(-\sqrt{2\pi}\right) \cdot A_k \Big / \left(h^2\dfrac{1}{\langle X_{ik}-x, \hat{\phi}_j\rangle^2} \dfrac{1}{c_{ijk}^3}\right)=\left(\dfrac{c_{ijk}}{h}\langle X_{ik}-x, \hat{\phi}_j\rangle \right)^4 \exp\left \{-\dfrac{1}{2}\left(\dfrac{c_{ijk}}{h} \langle X_{ik}-x, \hat{\phi}_j\rangle\right)^2\right \}.
\end{equation}

The term in Eq.(\ref{Ak2}), $|\dfrac{c_{ijk}}{h}\langle X_{ik}-x, \hat{\phi}_j\rangle| \overset{p}{\to} \infty$ by the following steps:
\begin{enumerate}[i)]
\item
$|\langle X_{ik}-x, \hat{\phi}_j\rangle|=|\langle X_{ik}-x, \phi_j\rangle|+Op\left(n^{-1/2}\right)$: from Lemma 3.4 of \cite{hosseini}, $\|\hat{\phi}_j-\phi_j\|=Op\left(n^{-1/2}\right)$. Then $|\langle X_{ik}-x, \hat{\phi}_j-\phi_j\rangle| \le \|X_{ik}-x\| \|\hat{\phi}_j-\phi_j\|=Op\left(n^{-1/2}\right)$, so $|\langle X_{ik}-x, \hat{\phi}_j\rangle|=|\langle X_{ik}-x, \phi_j\rangle|+Op\left(n^{-1/2}\right)=Op\left(1\right)$;
\item
$c_{ijk}$ is between $1/\sqrt{\lambda_j}+Op\left(n^{-1/2}\right)$ and $1/\sigma_{jk}+Op\left(n^{-1/2}\right)$: by Taylor expansion $c_{ijk}$ is somewhere between $1/\sqrt{\hat{\lambda}_j}$ and $1/\hat{\sigma}_{jk}$, where $\hat{\lambda}_j=\lambda_j+Op\left(n^{-1/2}\right)$ (\cite{DH2011}). The estimated $\hat{\sigma}_{jk}^2=\sum_{i=1}^{n_k}\langle X_{ik}-\bar{X}, \hat{\phi}_j\rangle^2/\left(n_k-1\right)$, with $\bar{X}$ the average function. Let $\tilde{\sigma}_{jk}^2=\sum_{i=1}^{n_k}\langle X_{ik}-\bar{X}, \phi_j\rangle^2/\left(n_k-1\right)$, which is well known to be root-n consistent with $\sigma_{jk}^2$. With $\|\hat{\phi}_j-\phi_j\|=Op\left(n^{-1/2}\right)$ again, $\langle X_{ik}-\bar{X}, \hat{\phi}_j\rangle^2-\langle X_{ik}-\bar{X}, \phi_j\rangle^2=Op \left(n^{-1/2}\right)$. So, $\hat{\sigma}_{jk}^2-\tilde{\sigma}_{jk}^2=\left(n_k-1\right)^{-1}\sum_{i=1}^{n_k} \left(\langle X_{ik}-\bar{X}, \hat{\phi}_j\rangle^2-\langle X_{ik}-\bar{X}, \phi_j\rangle^2\right)=Op\left(n^{-1/2}\right)$. Thus $\hat{\sigma}_{jk}^2$ is also root-n consistent with $\sigma_{jk}^2$, and so is $1/\hat{\sigma}_{jk}$ with $1/\sigma_{jk}$ by delta method. Thus $c_{ijk}$ is between $1/\sqrt{\lambda_j}+Op\left(n^{-1/2}\right)$ and $1/\sigma_{jk}+Op\left(n^{-1/2}\right)$, i.e.\ $c_{ijk}=Op \left(1\right)$;
\item
Then with above results, $|c_{ijk} \langle X_{ik}-x, \hat{\phi}_j\rangle|/h$ is between 
\begin{equation} \label{opsigma}
\left|\dfrac{1}{\sigma_{jk}}\langle X_{ik}-x, \phi_j\rangle \right|/h+Op \left(\dfrac{1}{\sqrt{n}h}\right),
\end{equation}
and
\begin{align} \label{oplambda}
&\left|\dfrac{1}{\sqrt{\lambda_j}}\langle X_{ik}-x, \phi_j\rangle \right|+Op \left(\dfrac{1}{\sqrt{n}h}\right) \nonumber\\
&=\dfrac{\sigma_{jk}}{\sqrt{\lambda_j}}\left|\dfrac{1}{\sigma_{jk}}\langle X_{ik}-x, \phi_j\rangle \right|+Op \left(\dfrac{1}{\sqrt{n}h}\right),
\end{align}
where r.v. $\dfrac{1}{\sigma_{jk}}\langle X_{ik}-x, \phi_j\rangle$ is standardized with finite mean.

So $\forall M > 0$, $P\left(|\dfrac{1}{\sigma_{jk}}\langle X_{ik}-x, \phi_j\rangle|/h > M\right)=P\left(|\dfrac{1}{\sigma_{jk}}\langle X_{ik}-x, \phi_j\rangle|> Mh\right) \to 1$ as $n \to \infty$, and then $|\dfrac{1}{\sigma_{jk}}\langle X_{ik}-x, \phi_j\rangle|/h \overset{p}{\to} \infty$. 

Also, $Op \left(\dfrac{1}{\sqrt{n}h}\right)=op(1)$, since $nh^2=n^{1-\delta}h^3 \cdot n^{\delta}h^{-1}$, and $n^{1-\delta}h^3$ for $\delta>0$ is bounded away from zero by assumption. So $nh^2 \to \infty$, and $\dfrac{1}{\sqrt{n}h} \to 0$. Therefore, both Eq.(\ref{opsigma}) and Eq.(\ref{oplambda}) $\overset{p} \to \infty$.
\end{enumerate} 
As a conclusion from i) {-} iii), $|c_{ijk} \langle X_{ik}-x, \hat{\phi}_j\rangle|/h \overset{p} \to \infty$. Then by continuous mapping, Eq.(\ref{Ak2}) $=op \left(1\right)$. Also, $\dfrac{1}{\langle X_{ik}-x, \hat{\phi}_j\rangle^2} \dfrac{1}{c_{ijk}^3}$ is apparently $Op \left(1\right)$ using above results, which in the end shows that $A_{ijk}=op(h^2)$. 

It also shows that $1/\hat{\sigma}_{jk}-1/\sqrt{\hat{\lambda}_j}=1/\sigma_{jk}-1/\sqrt{\lambda_j} + Op \left(n^{-1/2}\right)$. Therefore, from Eq.(\ref{I}){-}(\ref{IV}), we get to the result that 
\begin{equation} \label{sdgj}
\hat{\pi}_1 \hat{g}_{j1}\left(\hat{x}_j\right)+\hat{\pi}_0 \hat{g}_{j0}\left(\hat{x}_j\right)-\hat{g}_j\left(\hat{x}_j\right)=op \left(h\right).
\end{equation}
With similar steps, it also shows that $\hat{\pi}_1 \bar{g}_{j1}\left(x_j\right)+\hat{\pi}_0 g_{j0}\left(x_j\right)-\bar{g}_j\left(x_j\right)=op \left(h\right)$. So $\hat{\pi}_1\left \{ \hat{g}_{j1}\left(\hat{x}_j\right)-\bar{g}_{j1}\left(x_j\right)\right \}+\hat{\pi}_0\left \{\hat{g}_{j0}\left(\hat{x}_j\right)-\bar{g}_{j0}\left(x_j\right)\right \}=\hat{g}_j\left(\hat{x}_j\right)-\bar{g}_j\left(x_j\right)+op\left(h\right)$, and when combined with Theorem 3.1 from \cite{DH2010}, it proves
\begin{align} \label{conclusion1}
&\sup_{x \in \mathcal{S}(c)} \left | \hat{\pi}_1\left \{\hat{g}_{j1}\left(\hat{x}_j\right)-\bar{g}_{j1}\left(x_j\right)\right \}+\hat{\pi}_0\left \{\hat{g}_{j0}\left(\hat{x}_j\right)-\bar{g}_{j0}\left(x_j\right)\right \} \right | \nonumber\\
&=\sup_{x \in \mathcal{S}(c)} \left| \hat{g}_j\left(\hat{x}_j\right) - \bar{g}_j\left(x_j\right) \right|+op \left(h\right) \nonumber \\
&=op\left(\dfrac{1}{\sqrt{nh}}\right) + op \left(h\right) = op\left(h\right).
\end{align}
Then under Assumption A5, $\sup_{x \in \mathcal{S}(c)} \left|\hat{g}_{jk}\left (\hat{x}_j\right)-\bar{g}_{jk}\left(x_j\right)\right|=op\left(h+\sqrt{\dfrac{\log n}{nh}}\right)$, and
\begin{align} \label{diffg}
& \sup_{x \in \mathcal{S}(c)} \left|\hat{g}_{jk}\left(\hat{x}_j\right)-g_{jk}\left(x_j\right)\right| \nonumber\\
& \le \sup_{x \in \mathcal{S}(c)} \left|\hat{g}_{jk}\left (\hat{x}_j\right)-\bar{g}_{jk}\left(x_j\right)\right|+\sup_{x \in \mathcal{S}(c)} \left|\bar{g}_{jk}\left (x_j\right)-g_{jk}\left(x_j\right)\right| \nonumber \\
&=op \left(h+\sqrt{\dfrac{\log n}{nh}}\right)+ Op \left(h+\sqrt{\dfrac{\log n}{nh}}\right)=Op \left(h+\sqrt{\dfrac{\log n}{nh}}\right),
\end{align}
where the second bound in Eq.(\ref{diffg}) is from established results of kernel density estimation like in \cite{stone}. Consequently,
\begin{align} \label{fdiff}
&\sup_{x \in \mathcal{S}(c)} \left|\hat{f}_{jk}\left(\hat{x}_j\right)-f_{jk}\left(x_j\right)\right| \nonumber\\
&=\sup_{x \in \mathcal{S}(c)} \left|\dfrac{1}{\hat{\sigma}_{jk}}\hat{g}_{jk}\left(\hat{x}_j\right)-\dfrac{1}{\sigma_{jk}} g_{jk}\left(x_j\right)\right| \nonumber\\
& \le \sup_{x \in \mathcal{S}(c)}\left| \dfrac{1}{\hat{\sigma}_{jk}} \left \{\hat{g}_{jk}\left(\hat{x}_j\right)-g_{jk}\left(x_j\right)\right \} \right| + \sup_{x \in \mathcal{S}(c)}\left| \left(\dfrac{1}{\hat{\sigma}_{jk}} - \dfrac{1}{\hat{\sigma}_{jk}}\right) g_{jk}\left(x_j\right)\right| \nonumber\\
&=Op\left(h+\sqrt{\dfrac{\log n}{nh}}\right)+ Op\left(\dfrac{1}{\sqrt{n}}\right)=Op\left(h+\sqrt{\dfrac{\log n}{nh}}\right)
\end{align}
\end{proof}

\subsection{Difference between $\hat{u}_{jk}$ and $u_{jk}$}
We need the following Lemma \ref{lemmau} for Theorem 1 proof: 
\begin{lemma} \label{lemmau}
Under A1-A4, $\forall X \in \mathcal{L}^2(\mathcal{T})$, $\hat{u}_{jk}=\Phi^{-1}\left \{\hat{F}_{jk}\left(\langle X, \hat{\phi}_j\rangle\right)\right\}$ is root-n consistent of $u_{jk}=\Phi^{-1}\left \{F_{jk}\left(\langle X, \phi_j \rangle \right)\right\}$
\end{lemma}
\begin{proof}
Let $\hat{u}^*_{jk}=\Phi^{-1}\left \{\hat{F}_{jk}\left(\langle X, \phi_j\rangle\right)\right\}$. Here $\hat{F}_{jk}\left(\langle X, \phi_j\rangle\right)=\dfrac{\sum_{i=1}^{n_k} I\left\{\langle X_{ik}, \phi_j \rangle \le \langle X, \phi_j\rangle \right\}}{n_k+1}$, which easily gives $\hat{u}^*_{jk}-u_{jk}=Op\left(n^{-1/2}\right)$ by CLT and delta method. Then,
\begin{align} \label{eq19}
& \left|\hat{F}_{jk}\left(\langle X, \hat{\phi}_j\rangle\right)-\hat{F}_{jk}\left(\langle X, \phi_j\rangle\right) \right| \nonumber\\
&=\dfrac{\left|\sum_{i=1}^{n_k} I\left\{\langle X_{ik}-X, \hat{\phi}_j \rangle \le 0 \right\} - \sum_{i=1}^{n_k} I\left\{\langle X_{ik}-X, \phi_j \rangle \le 0 \right\}\right|}{n_k+1}\nonumber\\
& \le \dfrac{\sum_{i=1}^{n_k} I\left\{ I\left\{\langle X_{ik}-X, \hat{\phi}_j \rangle \le 0 \right\} \ne I\left\{\langle X_{ik}-X, \phi_j \rangle \le 0 \right\}\right\}}{n_k+1}.
\end{align}
From Eq.(\ref{eq19}),
\begin{equation} \label{split1}
E\left|\hat{F}_{jk}\left(\langle X, \hat{\phi}_j\rangle\right)-\hat{F}_{jk}\left(\langle X, \phi_j\rangle\right) \right| \le \dfrac{1}{n_k+1} \sum_{i=1}^{n_k} P\left( I\left\{\langle X_{ik}-X, \hat{\phi}_j \rangle \le 0 \right\} \ne I\left\{\langle X_{ik}-X, \phi_j \rangle \le 0 \right\} \right),
\end{equation}
so for $I\left\{\langle X_{ik}-X, \hat{\phi}_j \rangle \le 0 \right\} \ne I\left\{\langle X_{ik}-X, \phi_j \rangle \le 0 \right\}$, $\left|\langle X_{ik}-X, \hat{\phi}_j \rangle- \langle X_{ik}-X, \phi_j \rangle\right| > \epsilon_{ijk}$ for some $\epsilon_{ijk}>0$. Then Eq.(\ref{split1}) becomes
\begin{align} \label{Ediff}
E\left|\hat{F}_{jk}\left(\langle X, \hat{\phi}_j\rangle\right)-\hat{F}_{jk}\left(\langle X, \phi_j\rangle\right) \right| &\le \dfrac{1}{n_k+1} \sum_{i=1}^{n_k} P\left(\left|\langle X_{ik}-X, \hat{\phi}_j \rangle- \langle X_{ik}-X, \phi_j \rangle\right| > \epsilon_{ijk}\right) \nonumber\\
&=\dfrac{1}{n_k+1} \sum_{i=1}^{n_k} P\left(\left|\langle X_{ik}-X, \hat{\phi}_j-\phi_j \rangle\right| > \epsilon_{ijk}\right)
\end{align}

By Lemma 3.3 and 3.4 of \cite{hosseini}, as $n \to \infty$, $\sqrt{n}E\left|\langle X_{ik}-X, \hat{\phi}_j-\phi_j \rangle\right| \le \sqrt{E\|X_{ik}-X\|^2}\cdot \sqrt{E\|\sqrt{n}\left(\hat{\phi}_j-\phi_j\right)\|^2} < \infty$. Hence $\forall \epsilon >0$, $\sqrt{n}P\left(\left|\langle X_{ik}-X, \hat{\phi}_j-\phi_j \rangle\right| > \epsilon \right) \le
\left(\sqrt{n}E\left|\langle X_{ik}-X, \hat{\phi}_j-\phi_j \rangle\right|\right)/\epsilon<\infty$ by Markov inequality. 

Continuing from Eq.(\ref{Ediff}), as $n \to \infty$,
\begin{equation} \label{etop}
\sqrt{n}E\left|\hat{F}_{jk}\left(\langle X, \hat{\phi}_j\rangle\right)-\hat{F}_{jk}\left(\langle X, \phi_j\rangle\right) \right| \le \dfrac{n_k}{n_k+1}\left[\sqrt{n} P\left(\left|\langle X_{ik}-X, \hat{\phi}_j-\phi_j \rangle\right| > \epsilon_{ijk} \right)\right] < \infty,
\end{equation}
which proves $\sqrt{n}\left|\hat{F}_{jk}\left(\langle X, \hat{\phi}_j\rangle\right)-\hat{F}_{jk}\left(\langle X, \phi_j\rangle\right) \right|= Op \left(1\right)$. Then with Taylor expansion it easily shows $\hat{u}_{jk}-\hat{u}_{jk}^*=\Phi^{-1}\left(\hat{F}_{jk}\left(\langle X, \hat{\phi}_j\rangle\right)\right)-\Phi^{-1}\left(\hat{F}_{jk}\left(\langle X, \phi_j\rangle\right)\right)=Op\left(n^{-1/2}\right)$,  hence $\hat{u}_{jk}-u_{jk}=Op\left(n^{-1/2}\right)$ too, concluding the lemma.
\end{proof}

\subsection{Difference between $\check{\mathbf{\Omega}}_k^{jj'}$ and $\hat{\mathbf{\Omega}}_k^{jj'}$}
Here $\check{\mathbf{\Omega}}_k$ is estimated correlation matrix at group $k$ using sample rank correlation calculated from scores $\langle X_{ik}, \phi_j\rangle$, while $\hat{\mathbf{\Omega}}_k$ uses $\langle X_{ik}, \hat{\phi}_j\rangle$. For simplicity, we only demonstrate with Kendall's $\tau$, but other rank correlations like Spearman's $\rho$ will have similar results:
\begin{align}
&\hat{\mathbf{\Omega}}_k^{jj'}=\sin \left(\dfrac{\pi}{2} \hat{\rho}_{\tau, k}^{jj'}\right): \hat{\rho}_{\tau, k}^{jj'}=\dfrac{2}{n_k\left(n_k-1\right)}\sum_{1 \le i \le i' \le n_k} \text{sign} \left\{\langle X_{ik}-X_{i'k}, \hat{\phi}_j\rangle \langle X_{ik}-X_{i'k}, \hat{\phi}_{j'}\rangle\right\}\\
&\check{\mathbf{\Omega}}_k^{jj'}=\sin \left(\dfrac{\pi}{2} \check{\rho}_{\tau, k}^{jj'}\right): \check{\rho}_{\tau, k}^{jj'}=\dfrac{2}{n_k\left(n_k-1\right)}\sum_{1 \le i \le i' \le n_k} \text{sign} \left\{\langle X_{ik}-X_{i'k}, \phi_j\rangle \langle X_{ik}-X_{i'k}, \phi_{j'}\rangle\right\}.
\end{align} 
We then propose the following lemma:
\begin{lemma} \label{lemmatau} 
$\left|\hat{\mathbf{\Omega}}_k^{jj'}-\check{\mathbf{\Omega}}_k^{jj'}\right|=Op \left(\dfrac{1}{\sqrt{n}}\right)$, $\forall 1 \le j, j' \le J$, $j \ne j'$.
\end{lemma}

\begin{proof}
\begin{equation} \label{taudiff}
\begin{split}
\left|\hat{\rho}_{\tau, k}^{jj'}-\check{\rho}_{\tau, k}^{jj'}\right| \le \dfrac{4}{n_k\left(n_k-1\right)} \sum_{1 \le i < i' \le n_k} I [&\text{sign} \left\{\langle X_{ik}-X_{i'k}, \hat{\phi}_j\rangle \langle X_{ik}-X_{i'k}, \hat{\phi}_{j'}\rangle\right\} \\
&\ne \text{sign} \left\{\langle X_{ik}-X_{i'k}, \phi_j\rangle \langle X_{ik}-X_{i'k}, \phi_{j'}\rangle\right\}].
\end{split}
\end{equation}

To have unequal signs between $\langle X_{ik}-X_{i'k}, \hat{\phi}_j\rangle \langle X_{ik}-X_{i'k}, \hat{\phi}_{j'}\rangle$ and $\langle X_{ik}-X_{i'k}, \phi_j\rangle \langle X_{ik}-X_{i'k}, \phi_{j'}\rangle$, exactly either sign$\langle X_{ik}-X_{i'k}, \hat{\phi}_j\rangle \ne$ sign$\langle X_{ik}-X_{i'k}, \phi_j\rangle$, or sign$\langle X_{ik}-X_{i'k}, \hat{\phi}_{j'}\rangle \ne$ sign$\langle X_{ik}-X_{i'k}, \phi_{j'}\rangle$. So Eq.(\ref{taudiff}) has expectation
\begin{align} \label{etau}
E\left|\hat{\rho}_{\tau, k}^{jj'}-\check{\rho}_{\tau, k}^{jj'}\right| &\le \dfrac{4}{n_k\left(n_k-1\right)} \sum_{1 \le i < i' \le n_k} P\left(\text{sign} \langle X_{ik}-X_{i'k}, \hat{\phi}_j\rangle \ne \text{sign} \langle X_{ik}-X_{i'k}, \phi_j\rangle \right) \nonumber\\
&+ \dfrac{4}{n_k\left(n_k-1\right)} \sum_{1 \le i < i' \le n_k} P\left(\text{sign} \langle X_{ik}-X_{i'k}, \hat{\phi}_{j'}\rangle \ne \text{sign} \langle X_{ik}-X_{i'k}, \phi_{j'}\rangle \right) \nonumber\\
&\le \dfrac{4}{n_k\left(n_k-1\right)} \sum_{1 \le i < i' \le n_k} P\left(\left|\langle X_{ik}-X_{i'k}, \hat{\phi}_j-\phi_j\rangle\right| > \epsilon_{(i, i')jk}\right) \nonumber\\
&+ \dfrac{4}{n_k\left(n_k-1\right)} \sum_{1 \le i < i' \le n_k} P\left(\left|\langle X_{ik}-X_{i'k}, \hat{\phi}_{j'}-\phi_{j'}\rangle \right|> \epsilon_{(i, i')j'k}\right),
\end{align}
for $\epsilon_{(i, i')jk}$, $\epsilon_{(i, i')j'k}$ $> 0$, with the same reasoning as in Lemma \ref{lemmau}.

With results from proof steps of Lemma \ref{lemmau}, Eq.(\ref{etop}), $E\sqrt{n}\left|\hat{\rho}_{\tau, k}^{jj'}-\check{\rho}_{\tau, k}^{jj'}\right| < \infty$, $\Rightarrow \sqrt{n}\left|\hat{\rho}_{\tau, k}^{jj'}-\check{\rho}_{\tau, k}^{jj'}\right|=Op\left(1\right)$, $\Rightarrow \left|\hat{\rho}_{\tau, k}^{jj'}-\check{\rho}_{\tau, k}^{jj'}\right|=Op\left(\dfrac{1}{\sqrt{n}}\right)$. Thus with Taylor expansion it proves Lemma \ref{lemmatau}.
\end{proof}

\subsection{Asymptotic bound of $\left|\log \hat{Q}_J^*\left(X\right) - \log Q_J^*\left(X\right)\right|$}
\label{sup:correst}
Difference between the Bayes classifier and its estimated version is
\begin{align} 
\left| \log \hat{Q}_J^*\left(X\right) - \log Q_J^*\left(X\right)\right| &\le \sum_{k=0, 1}\sum_{j=1}^J  \left| \left( \log \hat{f}_{jk} \left(\hat{X}_j\right)- \log f_{jk} \left(X_j\right)\right)\right| \label{Qdiff1} \\
&+\dfrac{1}{2}\sum_{k=0, 1}\left|\log |\check{\mathbf{\Omega}}_k|-\log \left|\mathbf{\Omega}_k\right|\right| \label{Qdiff2}\\
&+\dfrac{1}{2}\sum_{k=0, 1}\left|\hat{\mathbf{u}}_k^T\left(\check{\mathbf{\Omega}}_k^{-1}-\mathbf{I}\right) \hat{\mathbf{u}}_k-\mathbf{u}_k^T\left(\mathbf{\Omega}_k^{-1}-\mathbf{I}\right)\mathbf{u}_k\right| \label{Qdiff3}\\
&+\dfrac{1}{2}\sum_{k=0, 1}\left|\log |\hat{\mathbf{\Omega}}_k|-\log |\check{\mathbf{\Omega}}_k|\right| + \dfrac{1}{2}\sum_{k=0, 1}\left|\hat{\mathbf{u}}_k^T\left(\hat{\mathbf{\Omega}}_k^{-1}-\check{\mathbf{\Omega}}_k^{-1}\right) \hat{\mathbf{u}}_k\right|, \label{Qdiff4}
\end{align}
Precision matrix is estimated using nonparanormal SKEPTIC with the graphical Dantzig selector described in \cite{yuan10} and \cite{Liu}. Asymptotic behavior of Eq.(\ref{Qdiff1}) is previously discussed in Section \ref{mkde}, $\hat{X}_j=\langle X, \hat{\phi}_j\rangle$.

\subsubsection{Bound of Eq.(\ref{Qdiff3})} \label{Q3}
To bound Eq.(\ref{Qdiff3}), we denote $\tilde{\mathbf{u}}_k=\hat{\mathbf{u}}_k-\mathbf{u}_k$, $\mathbf{M}_k=\check{\mathbf{\Omega}}_k^{-1}-\mathbf{\Omega}_k^{-1}$, where $\hat{\mathbf{u}}_k$ is a length $J$ vector with entries $\hat{u}_{jk}$ as defined above.
\begin{align} \label{part3}
\hat{\mathbf{u}}_k^T\left(\check{\mathbf{\Omega}}_k^{-1}-\mathbf{I}\right) \hat{\mathbf{u}}_k-\mathbf{u}_k^T\left(\mathbf{\Omega}_k^{-1}-\mathbf{I}\right)\mathbf{u}_k &= \mathbf{u}_k^T \mathbf{M}_k\mathbf{u}_k+2\mathbf{u}_k^T \mathbf{\Omega}_k^{-1}\tilde{\mathbf{u}}_k+2\mathbf{u}_k^T\mathbf{M}_k\tilde{\mathbf{u}}_k  \nonumber\\
&-2\mathbf{u}_k^T\tilde{\mathbf{u}}_k+\tilde{\mathbf{u}}_k^T\mathbf{\Omega}_k^{-1}\tilde{\mathbf{u}}_k+\tilde{\mathbf{u}}_k^T \mathbf{M}_k\tilde{\mathbf{u}}_k-\tilde{\mathbf{u}}_k^T\tilde{\mathbf{u}}_k
\end{align}
We discuss the asymptotic bound of each part in Eq.(\ref{part3}) from a) to f). For convenience of notation, $\|\cdot\|$ is for $\|\cdot\|_2$
\begin{enumerate}[a)]
\item
$\mathbf{u}_k^T\mathbf{M}_k\mathbf{u}_k \le \|\mathbf{u}_k\|^2 \cdot \|\mathbf{M}_k\|=Op \left(J\right) \cdot Op \left(M\sqrt{\dfrac{\log J}{n}}\right)=Op \left(MJ\sqrt{\dfrac{\log J}{n}}\right)$, where the bound on the norm of matrix difference comes from Theorem 4.4 in \cite{Liu}, and the fact that $\mathbf{\Omega}_k \in \mathcal{C}\left(\kappa, \tau, M, J\right)$;
\item
\begin{align}
2\mathbf{u}_k^T \mathbf{\Omega}_k^{-1}\tilde{\mathbf{u}}_k&=2\mathbf{u}_k^T \mathbf{\Omega}_k^{-1} Op\left(\dfrac{1}{\sqrt{n}}\right) \mathbf{1} \nonumber\\
&=Op\left(\dfrac{1}{\sqrt{n}}\right) \mathbf{u}_k^T \mathbf{\Omega}_k^{-1}\mathbf{1} \le Op\left(\dfrac{1}{\sqrt{n}}\right) \|\mathbf{u}_k\|\|\mathbf{\Omega}_k^{-1}\mathbf{1}\| \nonumber\\
&=Op\left(\dfrac{1}{\sqrt{n}}\right) \cdot Op\left(\sqrt{J}\right) \cdot Op\left(\sqrt{J}\right)=Op\left(\dfrac{J}{\sqrt{n}}\right),
\end{align}
where we have $\tilde{\mathbf{u}}_k=Op\left(\dfrac{1}{\sqrt{n}}\right) \mathbf{1}$ from Lemma \ref{lemmau}, and $\|\mathbf{\Omega}_k^{-1}\|_1 \le \kappa$;

\item
\begin{align}
2\mathbf{u}_k^T \mathbf{M}_k\tilde{\mathbf{u}}_k & \le 2\|\mathbf{u}_k\| \|\mathbf{M}_k\| \|\tilde{\mathbf{u}}_k\| \nonumber\\
&= Op \left(\sqrt{J}\right) \cdot Op \left(M\sqrt{\dfrac{\log J}{n}}\right)\cdot Op \left(\sqrt{\dfrac{J}{n}}\right)=Op \left(\dfrac{JM}{n}\sqrt{\log J}\right)
\end{align}

\item
\begin{align}
-2\mathbf{u}_k^T\tilde{\mathbf{u}}_k-\tilde{\mathbf{u}}_k^T\tilde{\mathbf{u}}_k=-\left(\hat{\mathbf{u}}_k+\mathbf{u}_k\right)^T\left(\hat{\mathbf{u}}_k-\mathbf{u}_k\right)=\|\mathbf{u}_k\|^2-\|\hat{\mathbf{u}}_k\|^2=Op \left(\dfrac{J}{\sqrt{n}}\right)
\end{align}

\item
\begin{align}
\tilde{\mathbf{u}}_k^T\mathbf{\Omega}_k^{-1}\tilde{\mathbf{u}}_k=Op \left(\dfrac{1}{\sqrt{n}}\right) \mathbf{1}^T \mathbf{\Omega}_k^{-1}Op \left(\dfrac{1}{\sqrt{n}}\right) \mathbf{1}=Op \left(\dfrac{J}{n}\right)
\end{align}

\item
\begin{align}
\tilde{\mathbf{u}}_k^T \mathbf{M}_k\tilde{\mathbf{u}}_k \le \|\tilde{\mathbf{u}}_k\|^2 \|\mathbf{M}_k\|=Op\left(\dfrac{MJ}{n}\sqrt{\dfrac{\log J}{n}}\right)
\end{align}
\end{enumerate}
In sum, Eq.(\ref{Qdiff3})$=Op \left(MJ\sqrt{\dfrac{\log J}{n}}\right)$

\subsubsection{Bound of Eq.(\ref{Qdiff2})} \label{bdet}
Log determinant difference in Eq.(\ref{Qdiff2}) can be bounded using Lemma 12 in \cite{Singh}:
\begin{equation}
\left|\log |\check{\mathbf{\Omega}}_k|-\log \left|\mathbf{\Omega}_k\right|\right| \le \dfrac{1}{\lambda^*}\|\check{\mathbf{\Omega}}_k-\mathbf{\Omega}_k\|_F,
\end{equation}
where $\lambda^*$ is the minimum among all eigenvalues of $\check{\mathbf{\Omega}}_k$ and $\mathbf{\Omega}_k$. Also, by Theorem 4.2 in \cite{Liu}, $\sup_{jj'} \left|\check{\mathbf{\Omega}}_k^{jj'}-\mathbf{\Omega}_k^{jj'}\right|=Op\left(\sqrt{\dfrac{\log J}{n}}\right)$. Thus, $\left|\log |\check{\mathbf{\Omega}}_k|-\log \left|\mathbf{\Omega}_k\right|\right|=Op \left(J\sqrt{\dfrac{\log J}{n}}\right)$.

\subsubsection{Bound of Eq.(\ref{Qdiff4})} \label{Q4}
With similar steps in Section \ref{bdet}, the first part in Eq.(\ref{Qdiff4}) is bounded as $\left|\log |\hat{\mathbf{\Omega}}_k|-\log |\check{\mathbf{\Omega}}_k|\right|=Op\left(\dfrac{J}{\sqrt{n}}\right)$, due to Lemma \ref{lemmatau}. For the second part, 
\begin{align}
\left|\hat{\mathbf{u}}_k^T\left(\hat{\mathbf{\Omega}}_k^{-1}-\check{\mathbf{\Omega}}_k^{-1}\right) \hat{\mathbf{u}}_k\right| &= \left|\hat{\mathbf{u}}_k^T \check{\mathbf{\Omega}}_k^{-1}\left(\check{\mathbf{\Omega}}_k-\hat{\mathbf{\Omega}}_k\right) \hat{\mathbf{\Omega}}_k^{-1}\hat{\mathbf{u}}_k\right| \nonumber\\
&\le \|\hat{\mathbf{u}}_k^T \check{\mathbf{\Omega}}_k^{-1}\| \|\check{\mathbf{\Omega}}_k-\hat{\mathbf{\Omega}}_k\| \|\hat{\mathbf{\Omega}}_k^{-1}\hat{\mathbf{u}}_k\|=Op\left(\dfrac{J^2}{\sqrt{n}}\right).
\end{align}
Thus, Eq.(\ref{Qdiff2}), Eq.(\ref{Qdiff3}) and Eq.(\ref{Qdiff4}) in sum are $Op \left(MJ\sqrt{\dfrac{\log J}{n}}\right)+Op\left(\dfrac{J^2}{\sqrt{n}}\right)$. 

\subsection{Proof of Theorem 1}\label{sec:suppTheorem1Proof}
\begin{proof}
We here inherit the idea in \cite{DMY2017} to only consider the case when $f_{j1}$ and $f_{j0}$ have common supports for simplicity. When $f_{j1}$ and $f_{j0}$ have unequal supports, we can divide the scenario into two parts: first, consider when the score of the target data $X$ fall into the common support of both densities, which is similar to what we discuss here; second, consider when the score only belongs to one support, which would be trivial to prove that $\log \hat{Q}_J^*\left(X\right)$ and $\log Q_J^*\left(X\right)$ always share the same sign. For detailed reasoning please refer to the Supplementary Material of \cite{DMY2017}.

For all $ \epsilon >0$, when $n$ is big enough, with parameters $c, C_{jk}, C_{T_1}, C_{T_2}$ dependent on $\epsilon$, we build the following sets:
\begin{itemize}
\item
$S_1=\left\{\|X\| \le c \right\}=\left\{X \in \mathcal{S}\left(c\right)\right\}$ s.t.\ $P\left(S_1\right) \ge 1-\epsilon/4$;
\item
By Proposition 1, let $S_2^{jk}=\left\{\sup_{x \in \mathcal{S}(c)} |\hat{f}_{jk}(\hat{x}_j)-f_{jk}(x_j)|/ \left(h+\sqrt{\dfrac{\log n}{nh}}\right) \le C_{jk}\right\}$, and $P\left(S_2^{jk}\right) \ge 1-2^{-\left(j+3\right)}$, for $j \ge1$, $k=0, 1$;
\item
Let $T_1=$ Eq.(\ref{Qdiff2}) $+$ Eq.(\ref{Qdiff3}). $T_1=Op\left(MJ\sqrt{\dfrac{\log J}{n}}\right)$ by Section \ref{Q3} and \ref{bdet}. $S_{T_1}=\left\{T_1/\left(MJ\sqrt{\dfrac{\log J}{n}}\right) \le C_{T_1}\right\}$, $P\left(S_{T_1}\right) \ge 1-\epsilon/4$;
\item
Let $T_2=$ Eq.(\ref{Qdiff4}). $T_2=Op\left(\dfrac{J^2}{\sqrt{n}}\right)$ by Section \ref{Q4}. $S_{T_2}=\left\{T_2/\left(\dfrac{J^2}{\sqrt{n}}\right) \le C_{T_2}\right\}$, $P\left(S_{T_2}\right) \ge 1-\epsilon/4$;
\item
Let $S^{jk}_3=\left\{\langle X, \phi_j \rangle \in \text{support}\left(f_{jk}\right)\right\}$. $P\left(S^{jk}_3\right)=1$. 
\end{itemize}

Let $S=S_1 \left\{\bigcap_{j \ge 1, k=0, 1} S_2^{jk}\right\} \cap S_{T_1} \cap S_{T_2} \left\{\bigcap_{j \ge 1, k=0, 1} S_3^{jk}\right\}$, $P\left(S\right)=1-P\left(S^c\right) \ge 1-\epsilon$. Since $\left(h+\sqrt{\dfrac{\log n}{nh}}\right) \to 0$, there exists $a_n \to \infty$ an increasing sequence which satisfies $a_n\left(h+\sqrt{\dfrac{\log n}{nh}}\right)=o\left(1\right)$. With $\mathcal{U}_{jk}=\left\{x: \langle x, \phi_j\rangle \in \text{support}\left(f_{jk}\right) \right\}$, $\mathcal{U}=\bigcap_{j \ge 1, k=0, 1} \mathcal{U}_{jk}$, and $d_{jk}=\min\left\{1, \inf_{x \in \mathcal{S}(c) \cap \mathcal{U}} f_{jk}\left(x_j\right)\right\}$, there is already a nondecreasing sequence $J_0\left(n\right)$ built by \cite{DMY2017}, which we can directly apply here: $$J_0\left(n\right)=\sup \left\{J' \ge 1: \sum_{j \le J', k=0, 1} \dfrac{M_{jk}}{d_{jk}} \le a_n\right\}.$$ It guarantees that Eq.(\ref{Qdiff1}): $ \sum_{k=0, 1}\sum_{j=1}^J  \left| \left( \log \hat{f}_{jk} \left(\hat{X}_j\right)- \log f_{jk} \left(X_j\right)\right)\right|=o\left(1\right)$ on the set $S$. 

Also, $T_1 \le MJ\sqrt{\log J} \cdot \dfrac{C_{T_1}}{\sqrt{n}}$ on $S$, subject to the condition in setup that $MJ\sqrt{\log J}=o\left(\sqrt{n}\right)$. As $\dfrac{C_{T_1}}{\sqrt{n}} \to 0$, $\exists b_n \to \infty$ and $b_n\dfrac{C_{T_1}}{\sqrt{n}} \to 0$. We here define $$J_1\left(n\right)=\sup \left\{J' \ge 1: M'J'\sqrt{\log J'} \le b_n\right\}.$$Then the nondecreasing $J_1$ satisfies the constraint $MJ\sqrt{\log J}=o\left(\sqrt{n}\right)$ and also guarantees $T_1=o\left(1\right)$ on $S$.

For $T_2 \le \dfrac{C_{T_2}}{\sqrt{n}} J^2$ on $S$, again $\exists c_n \to \infty$ and $c_n\dfrac{C_{T_2}}{\sqrt{n}} \to 0$. Let $$J_2\left(n\right)=\floor{\sqrt{c_n}}.$$ Then the sequence $J_2$ is nondecreasing and $T_2=o\left(1\right)$ on $S$ choosing $J=J_2$.

In sum, let $J^*\left(n\right)=\min \left\{J_0 \left(n\right), J_1 \left(n\right), J_2 \left(n\right)\right\}$, then $\left| \log \hat{Q}_J^*\left(X\right) - \log Q_J^*\left(X\right)\right| \to 0$ at $J=J^*\left(n\right)$ on $S$. With Assumption 4, the ratios $f_{j1}(X_j)/f_{j0}(X_j)$ are atomless, which therefore concludes $$P\left(S \cap \left\{\mathds{1}\left\{\log \hat{Q}^*_J\left(X\right) \ge 0 \right\} \ne \mathds{1}\left\{\log Q^*_J\left(X\right) \ge 0\right\}\right\}\right) \to 0.$$

\end{proof}

\section{Proofs of Theorem 2 \& 3}
\setcounter{equation}{0}
\subsection{Optimality of functional Bayes classifier on truncated scores} \label{section11}
The optimality of Bayes classification in multivariate case can be easily extended to the functional setting with first $J$ truncated scores: for a new case $X \in \mathcal{L}^2(\mathcal{T})$, the functional Bayes classifier $q^*_J=\mathds{1}\{\log Q^*_J(X)>0\}$, where
\begin{equation} \label{eq:supp1}
\log Q^*_J \left(X\right)=\log \left (\dfrac{\pi_1}{\pi_0} \right )+\displaystyle \sum_{j=1}^J \log \left \{ \dfrac{f_{j1}(X_j)}{f_{j0}(X_j)} \right \} + \log \left \{\dfrac{c_1\{ F_{11}(X_1), \ldots, F_{J1}(X_J)\}}{c_0\{ F_{10}(X_1), \ldots, F_{J0}(X_J)\}} \right \},
\end{equation}
achieves lower misclassification rate than any other classifier using the first $J$ scores $X_j=\langle X, \psi_j\rangle$, $j=1, \ldots, J$.
\begin{proof}
Let $q_J(X) = k$ be any classifier assigning $X$ to group $k$ based on its first $J$ scores. Define $D_k=\left\{\left(X_1, \ldots, X_J\right): q_J(X)=k\right\}$, $\mathds{1}_{D_k}=\mathds{1}\left\{\left(X_1, \ldots, X_J\right) \in D_k\right\}$. Then the misclassification rate of $q_J(X)$, denoted $\text{err}(q_J(X))$, is
\begin{align} \label{eq:supp2}
\text{err}\left\{q_J\left(X\right)\right\} &= P\left(q_J\left(X\right)=1, Y=0\right)+P\left(q_J\left(X\right)=0, Y=1\right) \nonumber \\
&=E\left[P\left(q_J\left(X\right)=1, Y=0 | X_1, \ldots, X_J\right)+P\left(q_J\left(X\right)=0, Y=1 | X_1, \ldots, X_J\right)\right]\nonumber \\
&=E\left[\mathds{1}_{D_1}P\left(Y=0 | X_1, \ldots, X_J\right)+\mathds{1}_{D_0}P\left(Y=1 | X_1, \ldots, X_J\right)\right]
\end{align} 
Thus, letting the corresponding functions $D_k^*$ and $\mathds{1}_{D_k^*}$ of Bayes classifier $q_J^*(X)$ being similar to $D_k$ and $\mathds{1}_{D_k}$, the difference between the error rates of $q_J(X)$ and $q_J^*(X)$ is
\begin{align} \label{eq:supp3}
\text{err}\left\{q_J\left(X\right)\right \}-\text{err}\left\{q_J^*\left(X\right)\right\} =&E[\left(\mathds{1}_{D_1}-\mathds{1}_{D_1^*}\right)P\left(Y=0 | X_1, \ldots, X_J\right) \nonumber\\
&+\left(\mathds{1}_{D_0}-\mathds{1}_{D_0^*}\right)P\left(Y=1 | X_1, \ldots, X_J\right)]
\end{align}
When $q_J(X)=0$, $q_J^*(X)=1$, $P\left(Y=1 | X_1, \ldots, X_J\right)> P\left(Y=0 | X_1, \ldots, X_J\right)$ by the definition of Bayes classification; and $P\left(Y=1 | X_1, \ldots, X_J\right)] > P\left(Y=0 | X_1, \ldots, X_J\right)$ when $q_J(X)=1$, $q_J^*(X)=0$. Therefore Eq.(\ref{eq:supp3}) is nonnegative, which proves the optimality of Bayes classification on truncated functional scores.
\end{proof}

\subsection{Theorem 2}\label{sec:suppTheorem2Proof}
\begin{proof}
When $X$ is Gaussian process under both $Y=0$ and $1$, let $\mathbf{X}_J=\left(X_1, \ldots, X_J\right)^T$, then the log ratio of $Q^*_J(X)$ is
\begin{equation} \label{ratio}
\log Q^*_J(X)=-\dfrac{1}{2}\left(\mathbf{X}_J-\vec{\mu}_J\right)^T \mathbf{R}_1^{-1}\left(\mathbf{X}_J-\vec{\mu}_J\right)+\dfrac{1}{2}\mathbf{X}_J^T \mathbf{R}_0^{-1}\mathbf{X}_J+\log \sqrt{\dfrac{|R_0|}{|R_1|}}
\end{equation}
At $k=0$, $\mathbf{X}_J^T \mathbf{R}_0^{-1}\mathbf{X}_J$ has central chi-square distribution with $J$ degrees of freedom, while $(\mathbf{X}_J-\vec{\mu}_J)^T \mathbf{R}_1^{-1}(\mathbf{X}_J-\vec{\mu}_J)$ is distributed generalized chi-squared. 

Eigendecomposition gives $\mathbf{R}_0^{1/2} \mathbf{R}_1^{-1} \mathbf{R}_0^{1/2}=\mathbf{P}^T \mathbf{\Delta} \mathbf{P}$, where $\mathbf{\Delta}$ is a diagonal matrix diag$\{\Delta_1, \ldots, \Delta_J\}$. Also determinant of $\mathbf{R}_0^{1/2} \mathbf{R}_1^{-1} \mathbf{R}_0^{1/2}$ is $\prod_{j=1}^J \frac{d_{j0}}{d_{j1}}=\prod_{j=1}^J\Delta_j$. We let $\mathbf{Z}=\mathbf{R}_0^{-1/2}\mathbf{X}_J$, $\mathbf{U}=\mathbf{P}\mathbf{Z}$. At $k=0$, $U_j$, as the $j$-th entry of vector $\mathbf{U}$, has standard Gaussian distribution; at $k=1$, $U_j \sim N(-b_j, 1/\Delta_j)$, with $b_j$ the $j$-th entry of $\mathbf{b}=-\mathbf{P}\mathbf{R}_0^{-1/2}\vec{\mu}_J$. $U_j$ and $U_{j'}$ are uncorrelated $\forall 1 \le j, j' \le J$, for both $k=0$ and $1$. 

Then Eq.(\ref{ratio}) is transformed into
\begin{align} \label{uform}
\log Q^*_J(X)&=-\dfrac{1}{2}\left(\mathbf{U}+\mathbf{b}\right)^T \mathbf{\Delta}\left(\mathbf{U}+\mathbf{b}\right)+\dfrac{1}{2}\mathbf{U}^T\mathbf{U}+\log \sqrt{\dfrac{|R_0|}{|R_1|}}\nonumber \\
&=-\dfrac{1}{2}\sum_{j=1}^J \Delta_j \left(U_j+b_j\right)^2+\dfrac{1}{2}\sum_{j=1}^J U_j^2 + \dfrac{1}{2}\sum_{j=1}^J \log \Delta_j
\end{align}

Eq. (\ref{uform}) thus fits into Lemma 3 in the Supplementary Material of \cite{DMY2017}, with which we conclude directly that perfect classification of $\mathds{1}\{\log Q^*_J(X)>0\}$ is achieved when either $\sum_{j=1}^{\infty} b_j^2 = \infty$, or $\sum_{j=1}^{\infty} (\Delta_j-1)^2 = \infty$, as $J \to \infty$. Otherwise $\log Q^*_J(X)$ converges almost surely to some random variable with finite mean and variance, thus err$\left(\mathds{1}\{\log Q^*_J(X)>0\}\right) \not \to 0$.

\end{proof}

\subsection{Proof of Theorem 3}\label{sec:proofTheorem3}
First, we provide a quick proof about the distribution of $u_{jk}|Y=k$ as mentioned in Section 5.3: $P\left[u_{jk} \le u|Y=k\right]=P\left[\Phi^{-1}\left(F_{jk}\left(X_{j}\right)\right) \le u |Y=k\right]=P\left[F_{jk}\left(X_j\right) \le \Phi\left(u\right)|Y=k\right]$. Since $F_{jk}\left(X_j\right)$ is a uniformly distributed variable at $Y=k$ (\cite{RM2015}), $P\left[u_{jk} \le u|Y=k\right]=\Phi\left(u\right)$. Thus $u_{jk}|Y=k \sim N(0, 1)$.

Second, we prove the claim that if a sequence of random variables $a_n > 0$ is $op\left(1\right)$, the conditional sequence $a_n|Y=k$, where $Y$ is binary with $k=0, 1$, is also convergent in probability to $0$:
\begin{proof}
To show $a_n|Y=k=op\left(1\right)$, we need to show $\forall \epsilon, \xi >0$, $\exists N_{\epsilon, \xi}$ such that, when $n \ge N_{\epsilon, \xi}$, $P\left(a_n > \epsilon|Y=k\right) < \xi$. 

Since $a_n=op\left(1\right)$, and $P\left(a_n > \epsilon\right)=P\left(a_n > \epsilon|Y=1\right)\pi_1+P\left(a_n > \epsilon|Y=0\right)\pi_0$, there exists $N_{\epsilon, \xi}'$ such that for $n \ge N_{\epsilon, \xi}'$, $P\left(a_n > \epsilon\right) < \pi_k \xi$, $\Rightarrow P\left(a_n > \epsilon|Y=k\right)\pi_k < \pi_k \xi$, $\Rightarrow P\left(a_n > \epsilon|Y=k\right)< \xi$. Thus it is proved that $\forall \epsilon, \xi$, such $N_{\epsilon, \xi}$ exists, and $N_{\epsilon, \xi} \le N_{\epsilon, \xi}'$, which concludes $a_n|Y \overset{p}{\to} 0$.
\end{proof}

Finally, to learn the asymptotic properties, we rely on the optimality of functional Bayes classification on truncated scores as discussed above. Any classifier on the same set of scores provides an upper bound of the error rate of the Bayes classifier $\mathds{1}\{\log Q^*_J(X)>0\}$. Therefore, let $\Gamma_J$ be the collection of all decision rules $\gamma_J$ using truncated scores $X_1, \ldots, X_J$, err$\left(\mathds{1}\{\log Q^*_J(X)>0\}\right) \le \min_{\gamma_J \in \Gamma_J} \text{err}\left(\gamma_J\right)$. Then perfect classification exists as long as there exists some classifier with asymptotic error rate converging to $0$. In the proof below, we build some decision rules with customized functions $T^a_j(X)$, etc., developed from the summand of $\log Q^*_J(X)$:

\begin{proof}
\begin{enumerate}[a)]
\item
For the first case, let $T^a_j(X)$ be defined as
\begin{equation}
T^a_{j}(X)=\log \dfrac{f_{j1}\left(X_{j}\right)}{f_{j0}\left(X_{j}\right)}\Big /\dfrac{\sqrt{\omega_{j1}}}{\sqrt{\omega_{j0}}}+\dfrac{1}{\omega_{j0}}\left(\mathbf{V}_{j0}^T\mathbf{u}_0\right)^2 = \log g_j+\left(\mathbf{V}_{j0}^T\mathbf{u}_0\right)^2/\omega_{j0},
\end{equation}
where $\mathbf{V}_{j0}$ as mentioned is $j$-th column of matrix $\mathbf{V}_0$ from the eigendecomposition $\mathbf{\Omega}_0=\mathbf{V}_0 \mathbf{D}_0 \mathbf{V}_0^T$. 

At $Y=0$, $\left(\mathbf{V}_{j0}^T\mathbf{u}_0\right)^2/\omega_{j0}$ follows $\chi^2_1$. Since there exists a subsequence $g^*_r=g_{j_r}$ of $g_j$ such that $g_{j_r} \overset{p}{\to} 0$, the subsequence is also $op \left(1\right)$ conditioned at $Y=0$, as proved previously. Therefore,
\begin{align} \label{eqa}
&P\left(T^a_{j_r}\left(X\right)>0|Y=0\right)=P\left(\log g_{j_r}+\left(\mathbf{V}_{j_r0}^T\mathbf{u}_0\right)^2/\omega_{j_r0} > 0 |Y=0\right) \nonumber\\ 
&=P\left(\log g_{j_r}+\left(\mathbf{V}_{j_r0}^T\mathbf{u}_0\right)^2/\omega_{j_r0} + C_a> C_a |Y=0\right), \forall C_a \in \mathds{R}^+\nonumber\\
&\le P\left(\log g_{j_r}+C_a >0 \cup \left(\mathbf{V}_{j_r0}^T\mathbf{u}_0\right)^2/\omega_{j_r0} >C_a |Y=0\right)\nonumber\\
&\le P\left(\log g_{j_r}+C_a >0|Y=0\right)+P\left(\left(\mathbf{V}_{j_r0}^T\mathbf{u}_0\right)^2/\omega_{j_r0} >C_a |Y=0\right)\nonumber\\
&=P\left(g_{j_r}>\exp\left\{-C_a\right\}|Y=0\right)+1-F_{\chi^2_1}\left(C_a\right) \nonumber \\
&\to 1-F_{\chi^2_1}\left(C_a\right), 
\end{align}
where $F_{\chi^2_1}$ is CDF of Chi-square distribution with d.f.\ $1$. As the inequality in Eq.(\ref{eqa}) exists $\forall C_a \in \mathds{R}^+$,  $P\left(\log g_{j_r}+\left(\mathbf{V}_{j_r0}^T\mathbf{u}_0\right)^2/\omega_{j_r0} > 0 |Y=0\right) \le \lim_{C_a \to \infty} 1-F_{\chi^2_1}\left(C_a\right)=0$.

At $Y=1$, 
\begin{align} \label{eqaa}
&P\left(\log g_{j_r}+\left(\mathbf{V}_{j_r0}^T\mathbf{u}_0\right)^2/\omega_{j_r0} < 0 |Y=1\right) \nonumber\\ 
&=P\left(s_{j_r0}\log g_{j_r}+s_{j_r0}\cdot\dfrac{\left(\mathbf{V}_{j_r0}^T\mathbf{u}_0\right)^2}{\omega_{j_r0}} < 0 |Y=1\right) \nonumber\\
&\le P\left(s_{j_r0}\log g_{j_r}+\epsilon <0|Y=1\right)+P\left(s_{j_r0}\cdot\dfrac{\left(\mathbf{V}_{j_r0}^T\mathbf{u}_0\right)^2}{\omega_{j_r0}} < \epsilon |Y=1\right), \forall \epsilon > 0 \nonumber\\
&\le P\left(|s_{j_r0}\log g_{j_r}|>\epsilon |Y=1\right)+P\left(\left|\sqrt{\dfrac{s_{j_r0}}{\omega_{j_r0}}}\mathbf{V}_{j_r0}^T\mathbf{u}_0\right| < \sqrt{\epsilon} |Y=1\right), \forall \epsilon > 0,
\end{align}
with $s_{j_r0}=1/\text{var}\left(V_{j_r0}^T\mathbf{u}_0/\sqrt{\omega_{j_r0}} |Y=1\right)$, as defined in Section 5.3. Thus $\sqrt{\dfrac{s_{j_r0}}{\omega_{j_r0}}}V_{j_r0}^T\mathbf{u}_0$ in the second probability part in Eq.(\ref{eqaa}) has unit variance. When $s_{j_r0} \to 0$, $s_{j_r0}\log g_{j_r} \overset{p}{\to} 0$ by continuous mapping and Slutsky's Theorem, so both probabilities in Eq.(\ref{eqaa}) go to $0$ when $\epsilon \to 0$. Consequently Eq.(\ref{eqaa}) converges to $0$, and the error rates of the sequence of decision rules $\mathds{1}\{T_{j_r}^a(X)>0\}$ are 
\begin{equation}
\text{err}\left(\mathds{1}\{T_{j_r}^a(X)>0\}\right)=P\left(T_{j_r}^a(X)>0 | Y=0\right)\pi_0+P\left(T_{j_r}^a(X)<0 | Y=1\right)\pi_1 \to 0.
\end{equation}
Therefore, the misclassification rate of $\mathds{1}\{\log Q^*_J(X)>0\}$ is asymptotically $0$ in this case.

\item
For the second case when the subsequence $1/g_{j_r} = op(1)$, the reasoning steps are similar. The term $T_j^b(X)$ is designed to build the decision rule here:
\begin{equation} \label{eqb}
T^b_{j}(X)=\log \dfrac{f_{j1}\left(X_{j}\right)}{f_{j0}\left(X_{j}\right)}\Big /\dfrac{\sqrt{\omega_{j1}}}{\sqrt{\omega_{j0}}}-\dfrac{1}{\omega_{j1}}\left(\mathbf{V}_{j1}^T\mathbf{u}_1\right)^2 = \log g_j-\left(\mathbf{V}_{j1}^T\mathbf{u}_1\right)^2/\omega_{j1}.
\end{equation}
Then at $Y=1$, $\left(\mathbf{V}_{j1}^T\mathbf{u}_1\right)^2/\omega_{j1}$ is $\chi_1^2$. Also, when $1/g_{j_r} = op(1)$, 
\begin{align} \label{eqbb}
&P\left(T^b_{j_r}\left(X\right)<0|Y=1\right)=P\left(\log g_{j_r}-\left(\mathbf{V}_{j_r1}^T\mathbf{u}_1\right)^2/\omega_{j_r1} < 0 |Y=1\right) \nonumber\\ 
&=P\left(\log g_{j_r}-\left(\mathbf{V}_{j_r1}^T\mathbf{u}_1\right)^2/\omega_{j_r1} + C_b< C_b |Y=1\right), \forall C_b \in \mathds{R}^+\nonumber\\
&\le P\left(\log g_{j_r} <C_b|Y=1\right)+P\left(\left(\mathbf{V}_{j_r1}^T\mathbf{u}_1\right)^2/\omega_{j_r1} >C_b |Y=1\right)\nonumber\\
&=P\left(g_{j_r}<\exp\left\{C_b\right\}|Y=1\right)+1-F_{\chi^2_1}\left(C_b\right) \nonumber \\
&\to 1-F_{\chi^2_1}\left(C_b\right), \forall C_b \in \mathds{R}^+,
\end{align}
since $1/g_{j_r}$ converges to $0$ in probability, i.e., $g_{j_r} \overset{p}{\to} \infty$. The error rate at $Y=1$ goes to $0$ as the inequality in Eq.(\ref{eqbb}) exists $\forall C_b \in \mathds{R}^+$. 

At $Y=0$, similarly to case a),
\begin{align} \label{eqbbb}
&P\left(\log g_{j_r}-\left(\mathbf{V}_{j_r1}^T\mathbf{u}_1\right)^2/\omega_{j_r1} > 0 |Y=0\right) \nonumber\\ 
&=P\left(s_{j_r1}\log g_{j_r}-s_{j_r1}\cdot\dfrac{\left(\mathbf{V}_{j_r1}^T\mathbf{u}_1\right)^2}{\omega_{j_r1}} > 0 |Y=0\right) \nonumber\\
&\le P\left(s_{j_r1}\log g_{j_r}>\epsilon |Y=0\right)+P\left(\epsilon-s_{j_r1}\cdot\dfrac{\left(\mathbf{V}_{j_r1}^T\mathbf{u}_1\right)^2}{\omega_{j_r1}} >0 |Y=0\right), \forall \epsilon > 0 \nonumber\\
&\le P\left(|s_{j_r1}\log g_{j_r}|>\epsilon |Y=0\right)+P\left(\left|\sqrt{\dfrac{s_{j_r1}}{\omega_{j_r1}}}\mathbf{V}_{j_r1}^T\mathbf{u}_1\right| < \sqrt{\epsilon} |Y=0\right), \forall \epsilon > 0,
\end{align}
and $s_{j_r1}=1/\text{var}\left(\mathbf{V}_{j_r1}^T\mathbf{u}_1/\sqrt{\omega_{j_r1}} |Y=0\right)$. Then again, when $s_{j_r1} \to 0$ and $g_{j_r} \overset{p}{\to} \infty$, $s_{j_r1}\log g_{j_r}$ is $op(1)$. Eq.(\ref{eqbbb}) goes to $0$ when $\epsilon \to 0$, and therefore asymptotic misclassification rate of the Bayes classifier is bounded up by $0$ in this case.  
\item
The third case uses $T_j^c(X)$ which is a combination of $T_j^a(X)$ and $T_j^b(X)$:
\begin{align}
T_j^c&=\log \dfrac{f_{j1}\left(X_{j}\right)}{f_{j0}\left(X_{j}\right)}\Big /\dfrac{\sqrt{\omega_{j1}}}{\sqrt{\omega_{j0}}}+\dfrac{1}{\omega_{j0}}\left(\mathbf{V}_{j0}^T\mathbf{u}_0\right)^2-\dfrac{1}{\omega_{j1}}\left(\mathbf{V}_{j1}^T\mathbf{u}_1\right)^2 \nonumber\\
&= \log g_j+\left(\mathbf{V}_{j0}^T\mathbf{u}_0\right)^2/\omega_{j0}-\left(\mathbf{V}_{j1}^T\mathbf{u}_1\right)^2/\omega_{j1}.
\end{align}
Then at $Y=0$, since $1/g_{j_r} \overset{p}{\to} 0$, and $s_{j_r1} \to 0$, the random variables $s_{j_r1}\log g_{j_r}$ and $s_{j_r1}\left(\mathbf{V}_{j_r0}^T\mathbf{u}_0\right)^2/\omega_{j_r0}$ are both $op(1)$, therefore,
\begin{align} \label{eqc}
P\left(T_{j_r}^c>0|Y=0\right) &= P\left(\log g_{j_r}+\left(\mathbf{V}_{j_r0}^T\mathbf{u}_0\right)^2/\omega_{j_r0}-\left(\mathbf{V}_{j_r1}^T\mathbf{u}_1\right)^2/\omega_{j_r1}>0|Y=0\right) \nonumber \\
&=P\left(s_{j_r1}\log g_{j_r}+s_{j_r1}\left(\mathbf{V}_{j_r0}^T\mathbf{u}_0\right)^2/\omega_{j_r0}-\left(\sqrt{\dfrac{s_{j_r1}}{\omega_{j_r1}}}\mathbf{V}_{j_r1}^T\mathbf{u}_1\right)^2>0|Y=0\right) \nonumber \\
& \le P\left(s_{j_r1}\log g_{j_r}+s_{j_r1}\left(\mathbf{V}_{j_r0}^T\mathbf{u}_0\right)^2/\omega_{j_r0}>\epsilon|Y=0\right) \nonumber \\ &+P\left(\left(\sqrt{\dfrac{s_{j_r1}}{\omega_{j_r1}}}\mathbf{V}_{j_r1}^T\mathbf{u}_1\right)^2<\epsilon|Y=0\right), \forall \epsilon > 0\nonumber \\
&\to P\left(\left|\sqrt{\dfrac{s_{j_r1}}{\omega_{j_r1}}}\mathbf{V}_{j_r1}^T\mathbf{u}_1\right|<\epsilon|Y=0\right), \forall \epsilon > 0,
\end{align}
and similar to case (b), $\sqrt{\dfrac{s_{j_r1}}{\omega_{j_r1}}}\mathbf{V}_{j_r1}^T\mathbf{u}_1$ has unit variance. Eq.(\ref{eqc}) goes to $0$ when $\epsilon \to 0$.

At $Y=1$, following previous steps, it is easy to find that $P\left(T_{j_r}^c<0|Y=1\right) \to 0$ when $g_{j_r} \to 0$ and $s_{j_r0} \to 0$ conditioned on $Y=1$, and therefore the proof is omitted here. In sum, the sufficiency of case (c) for perfect classification is verified.

\item
The last case uses $T_j^d=T_j^c$, where 
\begin{align} \label{eqd1}
P\left(T_{j_r}^d>0|Y=0\right)&=P\left(\log g_{j_r}+\left(\mathbf{V}_{j_r0}^T\mathbf{u}_0\right)^2/\omega_{j_r0}-\left(\mathbf{V}_{j_r1}^T\mathbf{u}_1\right)^2/\omega_{j_r1}>0|Y=0\right) \nonumber \\
& \le P\left(\log g_{j_r}+\left(\mathbf{V}_{j_r0}^T\mathbf{u}_0\right)^2/\omega_{j_r0}>0|Y=0\right),
\end{align}
and
\begin{align}\label{eqd2}
P\left(T_{j_r}^d<0|Y=1\right)&=P\left(\log g_{j_r}+\left(\mathbf{V}_{j_r0}^T\mathbf{u}_0\right)^2/\omega_{j_r0}-\left(\mathbf{V}_{j_r1}^T\mathbf{u}_1\right)^2/\omega_{j_r1}<0|Y=1\right) \nonumber \\
& \le P\left(\log g_{j_r}-\left(\mathbf{V}_{j_r1}^T\mathbf{u}_1\right)^2/\omega_{j_r1}<0|Y=1\right).
\end{align}
Eq.(\ref{eqd1}) with $g_{j_r} \overset{p}{\to} 0$ is already proved to go to $0$ in case (a), and Eq.(\ref{eqd2}) with $1/g_{j_r} \overset{p}{\to} 0$ converges to $0$ as shown in case (b), which complete the proof.
\end{enumerate}

\end{proof}

\lhead[]{}\rhead[\fancyplain{}\leftmark\footnotesize]{\fancyplain{}\rightmark\footnotesize{} }
\bibhang=1.7pc
\bibsep=2pt
\fontsize{9}{14pt plus.8pt minus .6pt}\selectfont
\renewcommand\bibname{\large \bf References}
\expandafter\ifx\csname
natexlab\endcsname\relax\def\natexlab#1{#1}\fi
\expandafter\ifx\csname url\endcsname\relax
  \def\url#1{\texttt{#1}}\fi
\expandafter\ifx\csname urlprefix\endcsname\relax\def\urlprefix{URL}\fi

\bibliography{references}


\end{document}